\documentclass[pdflatex,sn-mathphys-num]{sn-jnl}

\usepackage{graphicx}%
\usepackage{multirow}%
\usepackage{amsmath,amssymb,amsfonts}%
\usepackage{amsthm}%
\usepackage{mathrsfs}%
\usepackage[title]{appendix}%
\usepackage{xcolor}%
\usepackage{textcomp}%
\usepackage{manyfoot}%
\usepackage{booktabs}%
\usepackage{algorithm}%
\usepackage{algorithmicx}%
\usepackage{algpseudocode}%
\usepackage{listings}%
\usepackage{booktabs}
\usepackage{array}
\usepackage{graphicx}
\usepackage{verbatim}
\usepackage{diagbox}

\usepackage{chngcntr}

\makeatletter
\let\c@figure\c@table
\let\thefigure\thetable
\makeatother

\theoremstyle{thmstyleone}
\newtheorem{theorem}{Theorem}

\theoremstyle{thmstyletwo}

\theoremstyle{thmstylethree}

\begin{document}

\title[Penalized likelihood estimation of probability density functions using compositional splines]{Penalized likelihood estimation of probability density functions using compositional splines}

\author*[1]{\fnm{Stanislav} \sur{\v{S}kor\v{n}a}}\email{stanislav.skorna01@upol.cz}

\author[1]{\fnm{Jitka} \sur{Machalov\'a}}\email{jitka.machalova@upol.cz}

\affil*[1]{\orgdiv{Department of Mathematical Analysis and Applications of Mathematics, Faculty of Science}, \orgname{Palack\'y University Olomouc}, \orgaddress{\street{17. listopadu 12}, \city{Olomouc}, \postcode{77900}, \state{Czech Republic}}}

\abstract{Probability density functions are commonly estimated through preliminary smoothing or aggregation procedures, e.g., histograms or kernel density estimation, before subsequent functional representation and functional data analyses. Such a two-stage approach can lead to additional approximation bias and weaken the direct connection between the observed data and the underlying distributional structure. In this paper, we propose a penalized maximum likelihood framework for direct estimation of probability density functions from raw observations within the framework of Bayes Hilbert spaces while preserving the compositional geometry of densities. The methodology is based on the centered log-ratio (clr) transformation, an isometric isomorphism between Bayes Hilbert spaces and the standard Lebesgue space of square integrable functions with zero integral, enabling efficient spline representations. The clr transformed densities are represented using $Z\!B$-spline basis functions, while their smoothness is controlled through quadratic penalties imposed on the spline coefficients. The proposed framework is developed for univariate and bivariate densities. In the latter case, it naturally incorporates the orthogonal decomposition into independent and interactive parts together with the corresponding geometric marginals. 
The performance is evaluated in a simulation study involving multiple complex scenarios and compared with kernel smoothing. Finally, the applicability of the framework is illustrated using empirical geochemical data.}


\maketitle

\section{Introduction}\label{Introduction}
Samples of probability density functions arise naturally in modern applications as a result of large-scale data collection and their aggregation at the level of distributions rather than individual observations. Typical examples arise in studies of age \cite{hron16, Eckardt2024}, income \cite{Maier2021, Wang2026}, particle size \cite{Fayad2026}, and anthropometric characteristics \cite{hron23}, where each observational unit is represented by a sample of measurements over a common compact domain. In such a setting, the primary object of interest is not the individual observations themselves but the underlying distributional structure they form.

To enable further statistical analysis, typically within the framework of functional data analysis \cite{ramsay2005, kokoszka17}, these samples must be represented as functional objects. In practical applications, however, probability densities are typically not observed directly but are first estimated from raw data using preliminary procedures such as histograms \cite{machalova2021, hron23, Skorna2026} or kernel density estimators \cite{Grygar2024, Grygar2026}. Subsequently, basis expansions are commonly employed to obtain smooth functional representations, with $B$-splines \cite{deboor1978, dierckx1993} being a particularly popular choice due to their flexibility, local support, and computational efficiency.

Accordingly, a variety of spline-based methods have been developed for the estimation of probability density functions, differing in their treatment of domain structure, dependence patterns, and smoothness constraints. For instance, \cite{Mohaoui2025} employs variation-diminishing spline approximation to obtain shape-preserving density estimates, while \cite{Das2025} proposes a likelihood-based approach for estimating bivariate densities on irregular spatial domains using penalized splines. In a broader context, densities have also been treated as functional data objects \cite{Petersen2022}, enabling the application of standard functional data analysis tools through suitable transformations and basis representations. These approximations, however, remain based on preliminary density estimates obtained prior to the spline representation.

Nevertheless, most existing approaches rely on classical representations of densities as elements of standard function spaces and thus neglect their intrinsic compositional structure, where only relative information is relevant and densities are defined up to a multiplicative constant. Consequently, standard operations may not adequately reflect the geometric structure of probability densities, particularly with respect to positivity and relative scale. A suitable framework that naturally respects these properties is provided by Bayes Hilbert spaces \cite{egozcue06, boogaart14}, where probability densities are treated as compositional objects endowed with a Hilbert space structure. This framework has already been applied in a variety of recent developments, e.g., in regression modeling of densities in Bayes Hilbert space \cite{Maier2021}, where densities are treated as functional responses.

Within this framework, the centered log-ratio (clr) transformation establishes an isometric isomorphism between the Bayes space and a subspace consisting of square-integrable functions having a zero integral, which allows the use of standard methods from functional analysis. Early spline representations of clr-transformed densities were based on classical $B$-splines, together with additional constraints on the spline coefficients to enforce the zero-integral property \cite{machalova16}. This constrained formulation was subsequently extended to the bivariate setting \cite{hron23}, where coefficient constraints were likewise required to enforce the zero-integral property. To eliminate the need for such constraints, $Z\!B$-splines were introduced in \cite{machalova2021}. By construction, each $Z\!B$-spline has a zero integral; therefore, no additional restrictions on the spline coefficients are needed. The construction was later generalized to the bivariate setting \cite{Skorna2026}, providing a flexible spline representation for bivariate clr-transformed densities. Maximum penalized likelihood combined with $Z\!B$-spline representation was subsequently employed in \cite{Mondon2026} for univariate density estimation as a preprocessing step for outlier detection. Nevertheless, a unified framework that controls smoothness through penalties on the spline coefficients and extends the likelihood-based estimation to bivariate densities is still lacking.

To address this gap, we propose a penalized maximum likelihood framework for estimating probability density functions directly from raw observations within the Bayes Hilbert space setting. Rather than first estimating a density and subsequently approximating it by splines, the proposed methodology estimates the spline coefficients themselves through a penalized likelihood functional. The resulting framework combines likelihood-based estimation, spline representation, and compositional geometry within a unified model while naturally preserving the structural properties of probability densities.

The approach is developed for both univariate and bivariate densities, where in the latter case, the model naturally reflects the orthogonal decomposition into independent and interactive parts \cite{genest22, hron23, Skorna2026}. The resulting optimization problem is solved using efficient gradient-based methods.

The paper is organized as follows. In the next section, the main concepts of Bayes Hilbert spaces are introduced. Section \ref{Splines} summarizes the univariate and bivariate $Z\!B$-spline representation in zero-integral spaces. Section \ref{PML} derives the formulation of penalized maximum likelihood. Section \ref{rhoselection} describes the cross-validation procedure used for the selection of penalization parameters. Section \ref{Simulation} investigates the performance of the proposed methodology in a simulation study, and Section \ref{Application} demonstrates the applicability of the method to real geochemical data. Finally, Section \ref{Conclusion} concludes with a discussion and outlook. The proofs of the theoretical results are provided in Appendix~\ref{app:proofs}, while Appendices~\ref{app:simulation} and \ref{app:application} contain supplementary materials for the simulation study and the empirical application, respectively.

\section{Bayes Hilbert spaces} \label{Bayes}
 
Bayes Hilbert spaces provide a functional framework for the analysis of probability density functions by equipping them with a Hilbert space structure while preserving their compositional structure \cite{egozcue06, boogaart14}. The fundamental idea is that probability densities are considered up to multiplication by a positive constant, i.e., as equivalence classes of positive functions differing only by a scaling factor. This reflects the fact that the relevant information is contained in the shape of the density rather than in its absolute magnitude.

\subsection{Univariate case}
In the univariate setting, let $I = [a,b]$ be a compact interval and consider the corresponding Bayes space ${\cal B}^{2}(I)$, consisting of equivalence classes of strictly positive densities whose logarithms are square-integrable \cite{egozcue06, boogaart14}. For densities \(f,g\in{\cal B}^{2}(I)\) and \(\alpha\in\mathbb{R}\), the operations of perturbation (addition) and powering (scalar multiplication), which endow \({\cal B}^{2}(I)\) with a vector space structure, are defined as
$$
(f\oplus g)(x)=_{{\cal B}^{2}(I)} f(x)\cdot g(x),
\qquad
(\alpha\odot f)(x)=_{{\cal B}^{2}(I)} f^{\alpha}(x),
$$
where \(=_{{\cal B}^{2}(I)}\) denotes equality up to multiplication by a positive constant. Equipped with this vector space structure, $\mathcal{B}^2(I)$ becomes a Hilbert space with an inner product
\begin{equation*}
\langle f, g \rangle_{\mathcal{B}^2(I)} =
\frac{1}{2\,(b-a)}
\int_I \int_I
\ln \frac{f(x)}{f(t)} \,
\ln \frac{g(x)}{g(t)}
\, \mbox{d}x\, \mbox{d}t.
\end{equation*}
The corresponding norm and distance are induced in the standard way as
\begin{equation}\label{B_norm}
\|f\|_{\mathcal{B}^2(I)}=\sqrt{\langle f, f\rangle_{\mathcal{B}^2(I)}},\quad d_{\mathcal{B}^2(I)}(f,g)=\|f\ominus g\|_{\mathcal{B}^2(I)},
\end{equation}
where $f\ominus g=f\oplus[(-1)\odot g]$.

For practical purposes, it is convenient to work with the centered log-ratio (clr) transformation, an isometric isomorphism between ${\cal B}^{2}(I)$ and the space $L_0^2(I)$, the space of square-integrable functions with a zero integral. Consequently, every clr transformed probability density function satisfies this constraint. The clr transformation is defined as
$$
\mbox{clr}(f)(x) \equiv f^c(x) = \ln f(x) - \frac{1}{(b-a)}\int_I \ln f(t)\,\mathrm{d}t
$$
and preserves both linear structure and inner products. Moreover, it holds that
$$
\mbox{clr}(f\,\oplus\,g) = \mbox{clr}(f) + \mbox{clr}(g)\qquad \mbox{clr}(\alpha\,\odot\,f) = \alpha\cdot \mbox{clr}(f).
$$
The clr transformation thus provides an equivalent representation of densities in the linear space $L_0^2(I)$. Since the transformation is bijective, every clr-transformed density can be mapped back to the original Bayes space $\mathcal{B}^2(I)$ using the inverse clr transformation
\begin{equation}\label{B_density}
f(x) \equiv \mbox{clr}^{-1}(f^c)(x) =_{\mathcal{B}^2(I)} \exp\big(f^c(x)\big).
\end{equation}
See \cite{boogaart14, egozcue06} for further details.

\subsection{Bivariate case}
The framework naturally extends to the bivariate setting by considering a product domain $\Omega = [a,b] \times [c,d] \subset \mathbb{R}^2$ and the corresponding Bayes space ${\cal B}^{2}(\Omega)$ of equivalence classes of strictly positive bivariate densities that are square-integrable in the logarithmic sense \cite{hron23, genest22}. For densities $f,g\in\mathcal{B}^2(\Omega)$ and $\alpha\in\mathbb{R}$, the operations of perturbation  and powering are defined analogously as in the univariate case
$$
(f\oplus g)(x,y) =_{\mathcal{B}^2(\Omega)} f(x,y) \cdot g(x,y),\qquad (\alpha \odot f)(x,y) =_{\mathcal{B}^2(\Omega)} f^{\alpha}(x,y).
$$
Equipped with these operations, ${\mathcal B}^{2}(\Omega)$ forms a Hilbert space with an inner product
\begin{equation*}
\left\langle f, g\right\rangle_{\mathcal{B}^2(\Omega)} = \frac{1}{2\,(d-c)\,(b-a)} \iint_{\Omega} \iint_{\Omega} \ln \frac{f(x,y)}{f(t,s)}\, \ln \frac{g(x,y)}{g(t,s)} \, \mbox{d}x\,\mbox{d}y\,\mbox{d}t\,\mbox{d}s.
\end{equation*}
The corresponding norm and distance are defined as in \eqref{B_norm}. The clr transformation extends to the bivariate setting as
$$
\mbox{clr}(f)(x,y) \equiv f^c(x,y) = \ln f(x,y) - \frac{1}{(b-a)(d-c)} \iint_{\Omega} \ln f(t,s)\,\mathrm{d}t\,\mathrm{d}s,
$$
and constitutes an isometric isomorphism between ${\cal B}^{2}(\Omega)$ and $L_0^2(\Omega)$, the space of square-integrable functions with zero integral over $\Omega$.

An important advantage of the bivariate Bayes space framework is that it enables a separation of the marginal structure from the dependence structure of a density. This is achieved through an orthogonal decomposition based on geometric marginals \cite{hron23, genest22}. Specifically, for $f\in{\mathcal B}^{2}(\Omega)$, the geometric marginals are defined as
\begin{align*}
f_{1}(x) &=_{{\cal B}^{2}(\Omega)} \exp\left\{\frac{1}{d-c}\int_c^d\ln f(x,t)\,\mathrm{d}t\right\}, \\
f_{2}(y) &=_{{\cal B}^{2}(\Omega)}
\exp\left\{\frac{1}{b-a}\int_a^b\ln f(t,y)\,\mathrm{d}t \right\}.
\end{align*}

These functions correspond to orthogonal projections of $f$ onto subspaces of functions depending only on $x$ and $y$, respectively. Based on these projections, the density $f$ admits a unique orthogonal decomposition into an independent part and an interactive part,
$$
f = f_{ind} \oplus f_{int}, \qquad
f_{ind} = f_1 \oplus f_2, \qquad
f_{int} = f \ominus f_{ind},
$$
where the independent component captures the marginal structure, and the interactive component describes the remaining dependence. In the clr representation, this decomposition takes a particularly simple additive form
$$
f^c = f_{int}^c + f_1^c + f_2^c,
$$
with
$$
f_1^c(x) = \frac{1}{d-c}\int_c^d f^c(x,t)\,\mathrm{d}t, \qquad
f_2^c(y) = \frac{1}{b-a}\int_a^b f^c(t,y)\,\mathrm{d}t.
$$
This decomposition provides an interpretable separation between marginal structure and dependence, which is particularly useful in statistical modeling of multivariate densities.

Finally, the clr transformed densities $f^c(x,y)$ can be mapped back to the Bayes space using the inverse clr transformation and are represented as
\begin{equation}\label{B_density_2D}
f(x,y) \equiv \mathrm{clr}^{-1}(f^c)(x,y) =_{\mathcal{B}^2(\Omega)} \exp(f^c(x,y)),
\end{equation}
as well as their corresponding decomposition parts.

\section{Spline representation in zero-integral spaces}\label{Splines}

In this section, we introduce univariate and bivariate $Z\!B$-spline representations in the zero-integral spaces $L_0^2(I)$ and $L_0^2(\Omega)$, respectively. In Section~\ref{PML}, these spline representations will be used to approximate the clr transformations of unknown probability densities within a penalized likelihood framework. We first present the univariate construction and then extend it to the bivariate setting.

\subsection{Univariate $Z\!B$-spline representation}
\label{univariateZB}

In this subsection, we briefly review the construction of the univariate $Z\!B$-spline representation following \cite{machalova2021}.
Let $I = [a,b]\subset\mathbb{R}$ be a compact interval and let 
\begin{equation*}    
\Delta\lambda\, = \, \{\lambda_i\}_{i=0}^{g+1}, \quad  a=\lambda_{0}<\lambda_{1}<\ldots<\lambda_{g}<\lambda_{g+1}=b
\end{equation*}
be a sequence of knots on $I$. In the following, we will consider the corresponding extended sequence of knots $\Delta\Lambda=\{\lambda_i\}_{i=-k}^{g+k+1}$ with the coincident boundary knots
\begin{equation}\label{knotsx_ext}
\lambda_{-k}=\cdots =\lambda_{0}=a <\lambda_{1}<\ldots<\lambda_{g}< b =\lambda_{g+1} = \cdots =\lambda_{g+k+1}.
\end{equation}
Subsequently, we denote by $\mathcal{Z}_k^{\Delta\lambda}(I)$ the vector space of splines $s_k(x)$ of degree $k\in\mathbb{N}$ with knots $\Delta\lambda$ having a zero integral
$$
\int_{I} s_k(x)\,\mbox{d}x = 0.
$$
Therefore, $\mathcal{Z}_{k}^{\Delta\lambda}(I)\subset L^2_0(I)$. It is shown in 
\cite{machalova2021} that $\dim\left(\mathcal{Z}_k^{\Delta\lambda}(I)\right) = g+k$ and any spline $s_k(x)\in\mathcal{Z}_k^{\Delta\lambda}(I)$ can be represented as a linear combination of $Z\!B$-spline  basis functions, i.e.
\begin{equation}\label{sk}
s_k(x) = \sum_{i=-k}^{g-1}\,z_iZ_i^{k+1}(x),
\end{equation}
where $z_i$ are its coefficients and $Z_i^{k+1}(x)$, $i=-k,\dots,g-1$, are $Z\!B$-splines defined as
\begin{equation}\label{ZB-spline}
Z_i^{k+1}(x) = \frac{\mbox{d}}{\mbox{d}x} B_i^{k+2}(x).
\end{equation}
Here $B_i^{k+2}(x)$ stands for a classical $B$-spline of degree $k+1$ (order $k+2$), see \cite{dierckx1993}. The coincident boundary knots in~\eqref{knotsx_ext} are essential to ensure that all $Z\!B$-splines satisfy 
$$
\int_{I} Z_i^{k+1}(x)\,\mbox{d}x = 0
$$
for $i=-k,\ldots,g-1$. See \cite{machalova2021} for further details. Therefore, every spline $s_k(x)\in\mathcal{Z}_k^{\Delta\lambda}(I)$ in the form \eqref{sk} automatically satisfies the zero-integral constraint. For notational convenience, let $\textbf{z} = \left(z_{-k},\ldots,z_{g-1}\right)^{\top}$ and 
\begin{equation}\label{Zk_matrix}
\textbf{Z}_{k+1}(x) = \left(Z_{-k}^{k+1}(x),\ldots, Z_{g-1}^{k+1}(x) \right).
\end{equation} 
With this notation, representation \eqref{sk} can be written in matrix form as
\begin{equation}\label{sk_matrix}
s_k(x) = \textbf{Z}_{k+1}(x)\,\textbf{z}.
\end{equation}
Since $s_k\in \mathcal{Z}_k^{\Delta\lambda}(I)\subset L_0^2(I)$, it can be mapped to the Bayes space $\mathcal{B}^2(I)$ by the inverse clr transformation \eqref{B_density}. The corresponding compositional spline $\xi_k\in \mathcal{B}^2(I)$ is therefore defined as
\begin{equation*}
\xi_k(x)=\operatorname{clr}^{-1}\bigl(s_k\bigr)(x)
=\operatorname{clr}^{-1}\bigl(\mathbf{Z}_{k+1}(x)\mathbf{z}\bigr).
\end{equation*}
Thus, the coefficient vector $\mathbf{z}$ provides a finite-dimensional parameterization of the compositional spline $\xi_k$ in $\mathcal{B}^2(I)$.

\subsection{Bivariate $Z\!B$-spline representation}\label{bivariateZB}
In this subsection, we summarize the extension of the $Z\!B$-spline representation to the bivariate setting,  following \cite{Skorna2026}.
Let $\Omega = [a,b]\times[c,d]$ be a compact subset of $\mathbb{R}^2$, and let $\mathcal{Z}_{kl}^{\Delta\lambda,\Delta\mu}(\Omega)$ denote the vector space of splines $s_{kl}(x,y)$ of degree $k\in\mathbb{N}$ with an extended sequence of knots \eqref{knotsx_ext} in $x$ and of degree $l\in\mathbb{N}$ with a corresponding extended sequence of knots $\Delta M=\{\mu_j\}_{j=-l}^{h+l+1}$
$$
\mu_{-l}=\cdots=\mu_0=c <\mu_1<\cdots<\mu_h <d= \mu_{h+1}=\cdots=\mu_{h+l+1}
$$
in $y$. It was shown in \cite{Skorna2026}, that
$\dim({\cal Z}_{kl}^{\Delta\lambda,\Delta\mu}(\Omega)) = (g+k+1)(h+l+1)-1.$
Consequently, every spline $s_{kl}(x,y)\in{\cal Z}^{\Delta\lambda,\,\Delta\mu}_{kl}(\Omega)$ has a unique representation
\begin{equation}\label{ZBrepr}
s_{kl}(x,y) \, = \, \sum_{i=-k}^{g-1}\sum_{j=-l}^{h-1}z_{ij}\,Z_{ij}^{k+1,l+1}(x,y) + \sum_{i=-k}^{g-1}v_{i}\,Z_{i}^{k+1}(x,y) + \sum_{j=-l}^{h-1}\,u_{j}\,Z_{j}^{l+1}(x,y),
\end{equation}
where $z_{ij}$, $v_{i}$, $u_{j}$, $i=-k,\dots,g-1,\ j=-l,\dots,h-1$ are the corresponding spline coefficients and
\begin{equation*}
\begin{aligned}
Z_{ij}^{k+1,l+1}(x,y)
&= Z_i^{k+1}(x)\,Z_j^{l+1}(y),\\
Z_i^{k+1}(x,y)
&= Z_i^{k+1}(x),\\
Z_j^{l+1}(x,y)
&= Z_j^{l+1}(y).
\end{aligned}
\end{equation*}
which form the corresponding basis functions. Here, $Z_i^{k+1}(x)$ denotes the univariate $Z\!B$-splines defined in \eqref{ZB-spline}, while $Z_j^{l+1}(y)$ is constructed analogously in the $y$-direction. To simplify the notation, we can express the spline $s_{kl}(x,y)$ from \eqref{ZBrepr} in matrix form as
\begin{align}\label{ZBrepr_matrix}
s_{kl}(x,y) &= \left(\textbf{Z}_{k+1}(x),\ 1\right)
\begin{pmatrix}
\textbf{Z} & \textbf{v}\\
\textbf{u}^{\top} & 0
\end{pmatrix}
\begin{pmatrix}
\textbf{Z}_{l+1}^{\top}(y)\\
1
\end{pmatrix}\\ \nonumber
&= \textbf{Z}_{k+1}(x)\textbf{Z}\textbf{Z}^{\top}_{l+1}(y) + \textbf{Z}_{k+1}(x)\textbf{v} + \textbf{u}^{\top}\textbf{Z}^{\top}_{l+1}(y)
\end{align}
with $Z\!B$-spline coefficients $\mathbf{Z}=\left(z_{ij}\right)_{i=-k, j=-l}^{g-1, h-1}$,  $\mathbf{v}=(v_{-k},\dots,v_{g-1})^{\top}$ and $\mathbf{u}=(u_{-l},\dots,u_{h-1})^{\top}$ and vectors of univariate $Z\!B$-spline basis $\textbf{Z}_{k+1}(x)$ defined in \eqref{Zk_matrix} and $\textbf{Z}_{l+1}(y) = \left(Z_{-l}^{l+1}(y),\ldots, Z_{h-1}^{l+1}(y) \right)$. Moreover, using the column-wise vectorization operator $\mbox{cs}(\cdot)$ together with identities $\mbox{cs}(\textbf{AXB})= (\textbf{B}^{\top}\otimes \textbf{A})\,\mbox{cs}(\textbf{X})$ and $\textbf{u}^{\top}\textbf{Z}^{\top}_{l+1}(y) = \textbf{Z}_{l+1}(y)\textbf{u}$, the spline $s_{kl}(x,y)$ can be expressed as
\begin{equation}\label{ZB_repr_short}
s_{kl}(x,y) = \boldsymbol{\Phi}(x,y)\,\boldsymbol{\theta},
\end{equation}
where $\boldsymbol{\theta}=(\mbox{cs}(\textbf{Z})^{\top},\textbf{v}^{\top},\textbf{u}^{\top})^{\top}\in\mathbb{R}^{(g+k+1)(h+l+1)-1}$ is the coefficient vector and $\boldsymbol{\Phi}(x,y)=\left(\textbf{Z}_{l+1}(y)\otimes\textbf{Z}_{k+1}(x),\ \textbf{Z}_{k+1}(x),\ \textbf{Z}_{l+1}(y)\right)$ is the vector of basis $Z\!B$-splines.

The representation \eqref{ZBrepr} naturally reflects the orthogonal decomposition introduced in Section~\ref{Bayes}. Specifically, the individual terms in \eqref{ZBrepr} determine an orthogonal decomposition of the spline $s_{kl}$ into its independent and interactive parts. The terms 
$$
s_k^1(x,y) = 
\sum_{i=-k}^{g-1}v_{i}\,Z_{i}^{k+1}(x,y),\quad
s_l^2(x,y) = 
\sum_{j=-l}^{h-1}u_{j}\,Z_{j}^{l+1}(x,y),
$$
represent the corresponding geometric marginal parts and together form the independent part
$$
s_{kl}^{\mathrm{ind}}(x,y)=s_k^1(x,y)+s_l^2(x,y).
$$
The remaining term
$$
s_{kl}^{\mathrm{int}}(x,y) \, = \, \sum_{i=-k}^{g-1}\sum_{j=-l}^{h-1}z_{ij}\,Z_{ij}^{k+1,l+1}(x,y)
$$
represents the interactive component. Consequently,
$$
s_{kl}(x,y) = s_{kl}^{\mathrm{ind}}(x,y)
+ s_{kl}^{\mathrm{int}}(x,y).
$$
See \cite{Skorna2026} for further details. 
Since $s_{kl}\in Z_{kl}^{\Delta\lambda,\Delta\mu}(\Omega) \subset L_0^2(\Omega)$, it can be mapped to the Bayes space $\mathcal{B}^2(\Omega)$ by the inverse clr transformation \eqref{B_density_2D}. The corresponding bivariate compositional spline $\xi_{kl}\in \mathcal{B}^2(\Omega)$ is therefore defined as 
\begin{equation*}
\xi_{kl}(x,y) = \operatorname{clr}^{-1}\bigl(s_{kl}\bigr)(x,y).
\end{equation*}

In view of representations~\eqref{ZBrepr}, \eqref{ZBrepr_matrix} and \eqref{ZB_repr_short}, the coefficient matrix $\mathbf{Z}$ (or coefficient vector $\mbox{cs}(\textbf{Z})$, respectively) and the coefficient vectors $\mathbf{u}$ and $\mathbf{v}$ uniquely determine both the spline $s_{kl}$ and the corresponding compositional spline $\xi_{kl}$.

\section{Penalized maximum likelihood}\label{PML}

This section introduces the main methodological contribution of the paper, namely a penalized maximum likelihood approach for estimating probability densities directly from raw observations in Bayes Hilbert spaces. In contrast to the commonly used two-step approach \cite{machalova2021, Skorna2026}, in which a histogram or kernel density estimate is first constructed and subsequently approximated by splines, the proposed method estimates the $Z\!B$-spline coefficients directly from the raw observations. Smoothness is controlled through quadratic penalties imposed on these coefficients.

The univariate formulation is presented first, followed by its extension to the bivariate setting, where smoothness is controlled separately in the two coordinate directions.

\subsection{Penalized maximum likelihood for univariate densities}\label{PML1D}

Let $x_n\in I\subset \mathbb{R}$, $n=1,\dots N$ be independent observations from an unknown density $p\in\mathcal{B}^2(I)$. Following Section~\ref{univariateZB}, its clr transformation is approximated by a spline $s_k\in\mathcal{Z}_k^{\Delta\lambda}(I)$ of the form \eqref{sk_matrix}, where $\mathbf{z}$ is the vector of unknown $Z\!B$-spline coefficients. Choosing the unit-integral representative of $\operatorname{clr}^{-1}(s_k)$, the density associated with the vector of coefficients $\mathbf{z}$ is given by
\begin{equation}\label{density_z}
p_{\mathbf{z}}(x) =
\frac{\exp\bigl(s_k(x)\bigr)}
{\displaystyle\int_I\exp\bigl(s_k(t)\bigr)\,\mathrm{d}t}
=
\frac{\exp\bigl(\mathbf{Z}_{k+1}(x)\mathbf{z}\bigr)}
{\displaystyle\int_I
\exp\bigl(\mathbf{Z}_{k+1}(t)\mathbf{z}\bigr)\,\mathrm{d}t}.
\end{equation}
The coefficient vector $\mathbf{z}$ is estimated by maximizing the likelihood
$L(\mathbf{z}) = \prod_{n=1}^{N} p_{\mathbf{z}}(x_n)$.
Equivalently, since the logarithm is strictly increasing, we maximize the log-likelihood
\begin{equation*}
l(\mathbf{z})
=
\ln L(\mathbf{z})
=
\sum_{n=1}^{N}\ln p_{\mathbf{z}}(x_n).
\end{equation*}
Substituting the expression \eqref{density_z} for $p_{\mathbf{z}}$ yields
\begin{align*}
l(\mathbf{z}) 
& = \sum_{n=1}^N\ln\left(\frac{\exp(s_k(x_n))}{\int_{I}\exp(s_k(t))\,\mbox{d}t}\right) = \sum_{n=1}^N\ln\left(\exp(s_k(x_n))\right) - \sum_{n=1}^N\ln\int_{I}\exp(s_k(t))\,\mbox{d}t \\
& = \sum_{n=1}^Ns_k(x_n) - N\ln\int_{I}\exp(s_k(t))\,\mbox{d}t = \sum_{n=1}^{N}
\mathbf{Z}_{k+1}(x_n)\mathbf{z} - N\ln \int_I
\exp\bigl(\mathbf{Z}_{k+1}(t)\mathbf{z}\bigr)\,\mathrm{d}t
\nonumber\\
& = \mathbf{c}^{\top}\mathbf{z} - N\ln \int_I \exp\bigl(\mathbf{Z}_{k+1}(t)\mathbf{z}\bigr)\,\mathrm{d}t,
\end{align*}
where $c_i = \sum_{n=1}^N Z_i^{k+1}(x_n)$, $i=-k,\dots,g-1$, and $\textbf{c} = (c_{-k},\dots,c_{g-1})^{\top}$. 
Since maximizing the log-likelihood $l(\mathbf{z})$ is equivalent to minimizing its negative, the estimation problem can be reformulated as the standard unconstrained optimization problem 
\begin{equation*}
\min_{\mathbf{z}} \left\{ -l(\mathbf{z}) \right\} =
\min_{\mathbf{z}} \left\{-\mathbf{c}^{\top}\mathbf{z} + N\ln \int_I \exp\bigl(\mathbf{Z}_{k+1} (t)\mathbf{z}\bigr)\,\mathrm{d}t \right\}.
\end{equation*}
The negative log-likelihood controls the fit of the estimated density to the observed data, but it does not impose any restrictions on its smoothness. In particular, for a sufficiently flexible spline representation, neighboring coefficients may vary considerably, which can result in undesirable local oscillations of the estimated density. To control this behavior, we augment the objective function with a penalty on discrete differences of the spline coefficients, following the $P$-spline approach of \cite{eilers1996,eilers2021}.
For a given difference order 
$\mathrm{d}\in\{1,2,\dots,g+k-1\}$, the $\mathrm{d}$th-order discrete differences of $\mathbf z$ are given by
$$
\Delta^\mathrm{d}\textbf{z} = \textbf{D}_\mathrm{d}\,\textbf{z},
$$
where
$\textbf{D}_\mathrm{d}\in\mathbb{R}^{g+k-d,\,g+k}$ is the corresponding difference matrix with entries
\begin{equation}\label{D_matrix}
\mathbf{D}_\mathrm{d}(i,j) =
\begin{cases}
(-1)^{\mathrm{d}+i-j} \dbinom{\mathrm{d}}{j-i}, & \text{if } i \leq j \leq i + \mathrm{d}, \\
0, & \text{otherwise},
\end{cases}
\end{equation}
for $i = 1, \dots, g+k-d$ and $j = 1, \dots, g+k$. The corresponding quadratic penalty is therefore given by 
$$
\|\mathbf{D}_\mathrm{d}\,\mathbf{z}\|_2^2 = \mathbf{z}^\top\mathbf{P}_\mathrm{d}\,\mathbf{z},
$$
where $\mathbf{P}_\mathrm{d} = \mathbf{D}_\mathrm{d}^\top \mathbf{D}_\mathrm{d}.$
Combining the negative log-likelihood with the quadratic penalty, we define the penalized negative log-likelihood function as
\begin{equation}\label{penalized_likelihood}
l_{\rho}(\mathbf{z}) = -\mathbf{c}^{\top}\mathbf{z} + N \ln \int_I \exp\bigl(\mathbf{Z}_{k+1} (t)\mathbf{z}\bigr)\,\mathrm{d}t + \rho\,\mathbf{z}^{\top}\mathbf{P}_d\, \mathbf{z},
\end{equation}
where $\rho>0$ is the penalization parameter controlling the balance between the fit to the observed data and the smoothness of the estimated density. The penalized maximum likelihood estimate of the coefficient vector is then given by
$$
\textbf{z}_{\rho}^* = \operatorname*{argmin}_{\mathbf{z}}\ l_{\rho}(\textbf{z}).
$$

In the following, we assume that the observed sample is nondegenerate in the sense that
$$
\frac{\mathbf{c}^\top}{N} \in \operatorname{int}\operatorname{conv} \left\{ \mathbf{Z}_{k+1}(t): t\in I \right\},
$$
where $\operatorname{conv}(\cdot)$ denotes the convex hull of a set and $\operatorname{int}(\cdot)$ denotes its interior. Intuitively, this condition requires that the empirical representation of the sample in the $Z\!B$-spline basis lies in the interior, rather than on the boundary, of the corresponding convex support.

For fixed values of the penalization parameter $\rho$ and the difference order $\mathrm{d}$, the following theorem establishes the existence and uniqueness of the minimizer and, consequently, of the corresponding spline.

\begin{theorem}\label{thm:existence_uniqueness_univariate}
Let $\rho>0$ and $\mathrm{d}\in\{1,\ldots,g+k-1\}$ be fixed and assume
$$
\frac{\mathbf{c}^\top}{N} \in \operatorname{int}\operatorname{conv} \left\{ \mathbf{Z}_{k+1}(t): t\in I \right\}.
$$
Then there exists a unique coefficient vector $\mathbf{z}_{\rho}^*\in\mathbb{R}^{g+k}$ such that
$$
l_{\rho}\bigl(\mathbf{z}^*_{\rho}\bigr) \, = \, \min_{\mathbf{z}} l_{\rho}(\mathbf{z}).
$$
Since the $Z\!B$-splines form a basis of $\mathcal{Z}_{k}^{\Delta\lambda}(I)$, the vector $\mathbf{z}^*_{\rho}$ uniquely determines the spline $s^*_k(x) = \mathbf{Z}_{k+1}(x) \mathbf{z}^*_{\rho} \in\mathcal{Z}_{k}^{\Delta\lambda}(I).$
\end{theorem}

The proof is provided in Appendix~\ref{app:proofs}. It is based on the strict convexity of the penalized negative log-likelihood and on an existence argument for the corresponding unpenalized likelihood problem under the nondegeneracy assumption.

For a fixed spline space $\mathcal{Z}_k^{\Delta\lambda}(I)$, specified by the spline degree $k$ and the knot sequence $\Delta\Lambda$, and for given values of $\mathrm{d}$ and $\rho$, the minimizer $\mathbf{z}_\rho^*$ is computed numerically using a quasi-Newton method. The gradient of the penalized negative log-likelihood is derived analytically; its explicit expression is given by \eqref{eq:grad} in Appendix~\ref{app:proofs}. The integrals appearing in the objective function and in the analytical expression of the gradient are evaluated numerically by quadrature.

\begin{algorithm}[H]
\caption{Penalized likelihood estimation of a univariate density}
\label{alg:univariate}
\begin{algorithmic}[1]
\Require Observations $x_1,\ldots,x_N$, penalization parameter $\rho$, difference order $\mathrm{d}$, and tolerance $\varepsilon > 0$.
\State Compute $\mathbf{c}$ and construct
$\mathbf{P}_\mathrm{d}=\mathbf{D}_\mathrm{d}^{\top}\mathbf{D}_\mathrm{d}$.
\State Set $\mathbf{z}^{(0)}=\mathbf{0}$ and $r=0$.
\Repeat
\State Evaluate $l_\rho(\mathbf{z}^{(r)})$ and    $\nabla_{\mathbf z}l_\rho(\mathbf{z}^{(r)})$.
\State Compute $\mathbf{z}^{(r+1)}$ by a quasi-Newton step.
\State Set $r\leftarrow r+1$.
\Until{$\|\nabla_{\mathbf z} l_\rho(\mathbf{z}^{(r)})\|_2 <\varepsilon$}
\State Set $\mathbf{z}_\rho^*=\mathbf{z}^{(r)}$ and
\[
s_k^*(x)=\mathbf{Z}_{k+1}(x)\mathbf{z}_\rho^*.
\]
\State Return
\[
p_{\mathbf z_\rho^*}(x) = \frac{\exp\bigl(s_k^*(x)\bigr)} {\displaystyle\int_I\exp\bigl(s_k^*(t)\bigr)\,\mathrm{d}t}.
\]
\end{algorithmic}
\end{algorithm}

The quasi-Newton procedure is terminated when the Euclidean norm of the gradient falls below a prescribed tolerance. Upon convergence, the resulting coefficient vector provides a numerical approximation of the unique minimizer $\mathbf{z}_{\rho}^{*}$ established in Theorem~\ref{thm:existence_uniqueness_univariate}.

\subsection{Penalized maximum likelihood for bivariate densities}
\label{PML2D}

We now extend the penalized likelihood framework introduced in Section~\ref{PML1D}  to the bivariate setting. Let $(x_n,y_n)\in\Omega\subset\mathbb{R}^2$, $n=1,\dots,N$ be independent observations from an unknown density $p\in\mathcal{B}^2(\Omega)$. Following Section~\ref{bivariateZB}, its clr transformation is approximated by a spline $s_{kl}\in\mathcal{Z}_{kl}^{\Delta\lambda,\Delta\mu}(\Omega)$ represented by \eqref{ZBrepr}, or equivalently by \eqref{ZBrepr_matrix} and \eqref{ZB_repr_short} in matrix form. 

The derivation follows the same steps as in the univariate setting. We first formulate the corresponding normalized density and log-likelihood and then introduce penalties controlling smoothness in the two coordinate directions.

Choosing the unit-integral representative of $\operatorname{clr}^{-1}(s_{kl})$, the density associated with the coefficient matrix $\mathbf{Z}$ and the vectors of coefficients $\mathbf{v}$, and $\mathbf{u}$, or equivalently with the coefficient vector $\boldsymbol{\theta} = (\mbox{cs}(\textbf{Z})^{\top},\,\textbf{v}^{\top},\,\textbf{u}^{\top})^{\top}$, is given by 
\begin{equation}\label{density_bivariate}
p_{\boldsymbol{\theta}}(x,y) = \frac{\exp\bigl(s_{kl}(x,y)\bigr)}
{\displaystyle
\iint_{\Omega}\exp\bigl(s_{kl}(t,s)\bigr)\,\mathrm{d}t\,\mathrm{d}s}.
\end{equation}
The coefficients are estimated by maximizing the likelihood $L(\boldsymbol{\theta}) = \prod_{n=1}^{N} p_{\boldsymbol{\theta}}(x_n,y_n)$. Equivalently, we maximize the log-likelihood
\[
l(\boldsymbol{\theta})
=
\sum_{n=1}^{N}
\ln p_{\boldsymbol{\theta}}(x_n,y_n).
\]
Substituting \eqref{density_bivariate} yields
\begin{equation}\label{loglikelihood_bivariate}
l(\boldsymbol{\theta}) =
\sum_{n=1}^{N}s_{kl}(x_n,y_n) - N\ln \iint_{\Omega} 
\exp\bigl(s_{kl}(t,s)\bigr)\,\mathrm{d}t\,\mathrm{d}s.
\end{equation}
Using representation \eqref{ZBrepr}, the first term in \eqref{loglikelihood_bivariate} can be written as 
\begin{align}
l_1 &=
\sum_{n=1}^N \left(
\sum_{i=-k}^{g-1}\sum_{j=-l}^{h-1} z_{ij}Z_i^{k+1}(x_n)Z_j^{l+1}(y_n) +
\sum_{i=-k}^{g-1}v_iZ_i^{k+1}(x_n,y_n) +
\sum_{j=-l}^{h-1}u_jZ_j^{l+1}(x_n,y_n)
\right)
\nonumber\\
&=
\sum_{n=1}^N \sum_{i=-k}^{g-1}\sum_{j=-l}^{h-1} z_{ij}Z_i^{k+1}(x_n)Z_j^{l+1}(y_n) + \sum_{n=1}^N
\sum_{i=-k}^{g-1}v_iZ_i^{k+1}(x_n) +
\sum_{n=1}^N \sum_{j=-l}^{h-1}u_jZ_j^{l+1}(y_n)\nonumber\\
& = 
\sum_{i=-k}^{g-1}\sum_{j=-l}^{h-1} b_{ij}z_{ij}
+ \mathbf{c}^{\top}\mathbf{v}
+ \mathbf{d}^{\top}\mathbf{u} \nonumber,
\end{align}
where
$$
b_{ij} =  \sum_{n=1}^N Z_i^{k+1}(x_n)Z_j^{l+1}(y_n), \qquad c_i = \sum_{n=1}^N Z_i^{k+1}(x_n),\qquad  d_j = \sum_{n=1}^N Z_j^{l+1}(y_n),
$$
for $i=-k,\dots,g-1$ and $j=-l,\dots,h-1$ with 
$$\textbf{c} = (c_{-k},\dots,c_{g-1})^{\top},\qquad \textbf{d} = (d_{-l},\dots,d_{h-1})^{\top}.$$ 
We further define
$\mathbf B=(b_{ij})_{i=-k,j=-l}^{g-1,h-1}$ and $\mathbf b=\operatorname{cs}(\mathbf B)$. Then the first term satisfies $\sum_{i=-k}^{g-1}\sum_{j=-l}^{h-1} z_{ij}b_{ij} = \mathbf b^\top\operatorname{cs}(\mathbf Z)$,
and hence
$$
l_1 =
\mathbf b^\top\operatorname{cs}(\mathbf Z) + \mathbf c^\top\mathbf v +
\mathbf d^\top\mathbf u.
$$
Recalling the coefficient vector defined in \eqref{ZB_repr_short} as
$$
\boldsymbol{\theta} = 
\begin{pmatrix}
\operatorname{cs}(\mathbf Z)\\
\mathbf v\\
\mathbf u
\end{pmatrix},
$$
we further define
\begin{equation}\label{def:h}
\mathbf h =
\begin{pmatrix}
\mathbf b\\
\mathbf c\\
\mathbf d
\end{pmatrix}.
\end{equation}
Hence, the first term of the log-likelihood can be written compactly as
$$
l_1=\mathbf h^\top\boldsymbol{\theta}.
$$
Using the matrix representation \eqref{ZBrepr_matrix}, the second term in \eqref{loglikelihood_bivariate} takes the form
\begin{equation*}
l_2 = N\ln\left(\iint_\Omega\exp{(\textbf{Z}_{k+1}(t)\,\textbf{Z}\,\textbf{Z}^{\top}_{l+1}(s) + \textbf{Z}_{k+1}(t)\,\textbf{v}} + \textbf{u}^{\top}\textbf{Z}^{\top}_{l+1}(s))\,\mbox{d}t\,\mbox{d}s\right).
\end{equation*}
or equivalently
\begin{equation*}
l_2 = N\ln\iint_{\Omega}\exp(\boldsymbol{\Phi}(t,s)\boldsymbol{\theta})\,\mbox{d}t\,\mbox{d}s,
\end{equation*}
using the representation \eqref{ZB_repr_short}.

Having expressed both terms of the log-likelihood in matrix form, the estimation problem can be formulated directly in terms of the unknown spline coefficients. Since maximizing the log-likelihood is equivalent to minimizing its negative, we consider the corresponding negative log-likelihood. As in the univariate case, this criterion controls the fit to the observed data but does not impose smoothness on the resulting bivariate density estimate. We therefore introduce separate penalties in the two coordinate directions.

Let $\textbf{D}_{\mathrm{d}_x}\in\mathbb{R}^{g+k-\mathrm{d}_x,g+k}$ and $\textbf{D}_{\mathrm{d}_y}\in\mathbb{R}^{h+l-\mathrm{d}_y,h+l}$ denote the matrices of difference orders $\mathrm{d}_x\in\{1,2,\dots,g+k-1\}$ and $\mathrm{d}_y\in\{1,2,\dots,h+l-1\}$, respectively, having entries defined by \eqref{D_matrix}. The resulting penalties for each coordinate direction are given by
\begin{align*}
\Vert\textbf{D}_{\mathrm{d}_x}\textbf{Z}\Vert^2_F &+ \Vert\textbf{D}_{\mathrm{d}_x}\,\textbf{v}\Vert^2_F, \\
\Vert\textbf{Z}\,\textbf{D}^{\top}_{\mathrm{d}_y}\Vert^2_F & + \Vert\textbf{D}_{\mathrm{d}_y}\textbf{u}\Vert^2_F,
\end{align*}
where $\Vert\cdot\Vert_F$ denotes the Frobenius norm. This choice is a direct extension of the univariate penalty  $\Vert\textbf{D}_\mathrm{d}\textbf{z}\Vert^2_2$ to the bivariate setting by penalizing discrete differences of the coefficient matrix $\textbf{Z}$ and the vectors of coefficients $\textbf{u}$ and $\textbf{v}$ in a unified manner. Since the Frobenius norm coincides with the Euclidean norm for vectors, i.e.
\begin{equation*}
\Vert\textbf{a}\Vert^2_F = \mathrm{tr}(\textbf{a}^{\top}\textbf{a}) = \textbf{a}^{\top}\textbf{a} = \Vert\textbf{a}\Vert^2_2,
\end{equation*}
the respective second penalty terms reduce to the standard quadratic forms
\begin{align*}
\Vert\textbf{D}_{\mathrm{d}_x}\,\textbf{v}\Vert^2_F &= \Vert\textbf{D}_{\mathrm{d}_x}\,\textbf{v}\Vert^2_2 = \textbf{v}^{\top}\textbf{P}_{\mathrm{d}_x}\,\textbf{v},\\
\Vert\textbf{D}_{\mathrm{d}_y}\,\textbf{u}\Vert^2_F &= \Vert\textbf{D}_{\mathrm{d}_y}\,\textbf{u}\Vert^2_2 =  \textbf{u}^{\top}\textbf{P}_{\mathrm{d}_y}\,\textbf{u},
\end{align*}
with $\mathbf P_{\mathrm{d}_x}=\mathbf D_{\mathrm{d}_x}^\top\mathbf D_{\mathrm{d}_x}$ and $\mathbf P_{\mathrm{d}_y}=\mathbf D_{\mathrm{d}_y}^\top\mathbf D_{\mathrm{d}_y}$.

Using the vectorization $\operatorname{cs}(\cdot)$, together with the identity
$\|\mathbf A\|_F^2 = \|\operatorname{cs}(\mathbf{A} )\|_2^2$ and the properties of Kronecker tensor product $\otimes$, the penalty term involving the coefficient matrix $\mathbf Z$ in the $x$-direction can be rewritten as follows
\begin{align*}
\Vert\mathbf D_{\mathrm{d}_x}\mathbf Z\Vert_F^2 & = \Vert\operatorname{cs}(\mathbf D_{\mathrm{d}_x}\mathbf Z)\Vert_2^2  = \Vert(\mathbf I_{h+l}\otimes\mathbf D_{\mathrm{d}_x}) \operatorname{cs}(\mathbf Z)\Vert_2^2 \\
& = \mbox{cs}(\textbf{Z})^{\top}(\textbf{I}_{h+l}\otimes\textbf{D}_{\mathrm{d}_x})^{\top} (\textbf{I}_{h+l}\otimes \textbf{D}_{\mathrm{d}_x})\,\mbox{cs}(\textbf{Z})\\
&= \mbox{cs}(\textbf{Z})^{\top}(\textbf{I}_{h+l}\otimes\textbf{P}_{\mathrm{d}_x})\,\mbox{cs}(\textbf{Z}) = \mbox{cs}(\textbf{Z})^{\top}\widetilde{\mathbf{P}}_{\mathrm{d}_x}\,\mbox{cs}(\textbf{Z}),
\end{align*}
where
$\widetilde{\mathbf{P}}_{\mathrm{d}_x} = \textbf{I}_{h+l}\otimes\textbf{P}_{\mathrm{d}_x}$. Similarly, the penalty term in the $y$-direction satisfies
$$
\Vert\textbf{Z}\,\textbf{D}^{\top}_{\mathrm{d}_y}\Vert^2_F = \mathrm{cs}(\textbf{Z})^{\top} \widetilde{\mathbf{P}}_{\mathrm{d}_y}\mathrm{cs}(\textbf{Z})
$$
with $\widetilde{\mathbf{P}}_{\mathrm{d}_y} = \textbf{P}_{\mathrm{d}_y}\otimes\textbf{I}_{g+k}$. Combining the negative log-likelihood with the quadratic penalties, we define the penalized negative log-likelihood
\begin{align*}
l_{\rho_x,\rho_y}(\mbox{cs}(\textbf{Z}),\textbf{v},\textbf{u}) &= 
- \mathbf{b}^{\top}\mbox{cs}(\mathbf{Z})
- \mathbf{c}^{\top}\mathbf{v}
- \mathbf{d}^{\top}\mathbf{u} \\ &\quad+N\ln\left(\iint_\Omega\exp{(\textbf{Z}_{k+1}(t)\,\textbf{Z}\,\textbf{Z}^{\top}_{l+1}(s) + \textbf{Z}_{k+1}(t)\,\textbf{v}} + \textbf{u}^{\top}\textbf{Z}^{\top}_{l+1}(s))\,\mbox{d}t\,\mbox{d}s\right) \\
&\quad + \rho_x\left(\mbox{cs}(\textbf{Z})^{\top} \widetilde{\mathbf{P}}_{\mathrm{d}_x}\,\mbox{cs}(\textbf{Z}) + \textbf{v}^{\top}\textbf{P}_{\mathrm{d}_x}\,\textbf{v}\right) +\rho_y\left(\mathrm{cs}(\textbf{Z})^{\top} \widetilde{\mathbf{P}}_{\mathrm{d}_y}\,\mbox{cs}(\textbf{Z})+\textbf{u}^{\top}\textbf{P}_{\mathrm{d}_y}\,\textbf{u}\right),
\end{align*}
where $\rho_x>0$ and $\rho_y>0$ denote the smoothing parameters associated with the two coordinate directions, controlling the trade-off between the fit  to the observed data and the smoothness of the estimated density. The penalized negative log-likelihood function can be expressed analogously with \eqref{penalized_likelihood} in terms of the coefficient vector $\boldsymbol{\theta}$ using the spline representation \eqref{ZB_repr_short} and \eqref{def:h} yielding
$$
l_{\rho_x,\rho_y}(\boldsymbol{\theta}) = -\textbf{h}^{\top}\boldsymbol{\theta} + N\ln\left(\iint_\Omega \exp(\boldsymbol{\Phi}(t,s)\boldsymbol{\theta})\,\mbox{d}t\,\mbox{d}s\right) + \boldsymbol{\theta}^{\top}\mathbb{P}_{\rho_x,\rho_y}\boldsymbol{\theta},
$$
where
$$
\mathbb{P}_{\rho_x,\rho_y} = \begin{pmatrix}
    \rho_x\widetilde{\mathbf{P}}_{\textrm{d}_x} + \rho_y\widetilde{\mathbf{P}}_{\textrm{d}_y} & \textbf{0} & \textbf{0}\\
    \textbf{0} & \rho_x\textbf{P}_{\textrm{d}_x} & \textbf{0}\\
    \textbf{0} & \textbf{0} & \rho_y\textbf{P}_{\textrm{d}_y}.
\end{pmatrix}.
$$
The penalized maximum likelihood estimate of the coefficient vector $\boldsymbol{\theta}$ is then obtained by minimizing the penalized objective function,
$$
\boldsymbol{\theta}^*_{\rho_x,\rho_y} =
\operatorname*{argmin}_{\substack{\boldsymbol{\theta}}}
l_{\rho_x,\rho_y}(\boldsymbol{\theta}).
$$
Analogously with the univariate case, we assume that the observed sample is nondegenerate, i.e.,
$$
\frac{\mathbf{\textbf{h}}^\top}{N} \in \operatorname{int}\operatorname{conv} \left\{\boldsymbol{\Phi}(t,s): (t,s)\in \Omega \right\},
$$
meaning that the observations lie in the interior of the corresponding convex support. The following Theorem states that for fixed values of penalization parameters $\rho_x$, $\rho_y$ and the difference orders $\mathrm{d}_x$, $\mathrm{d}_y$ the optimization problem admits a unique solution.

\begin{theorem}\label{thm:existence_uniqueness_bivariate}
Let $\rho_x>0$, $\rho_y>0$ and $\mathrm{d}_x\in\{1,\ldots,g+k-1\}$, $\mathrm{d}_y\in\{1,\ldots,h+l-1\}$ be fixed and assume
$$
\frac{\mathbf{h}^\top}{N} \in \operatorname{int}\operatorname{conv} \left\{\boldsymbol{\Phi}(t,s): (t,s)\in \Omega \right\}.
$$
Then there exists a unique vector of coefficients $\boldsymbol{\theta}^*\in\mathbb{R}^{(g+k+1)(h+l+1)-1}$ such that
$$
l_{\rho_x,\rho_y}(\boldsymbol{\theta}^*_{\rho_x,\rho_y}) \, = \, \min_{\substack{\boldsymbol{\theta}}} l_{\rho_x,\rho_y}(\boldsymbol{\theta}).
$$
Since the $Z\!B$-splines form a basis of $\mathcal{Z}_{kl}^{\Delta\lambda,\Delta\mu}(\Omega)$, the vector $\boldsymbol{\theta}^*_{\rho_x,\rho_y}$ uniquely determines the spline $s^*_{kl}(x,y) = \boldsymbol{\Phi}(x,y)\boldsymbol{\theta}^*_{\rho_x,\rho_y}\in\mathcal{Z}^{\Delta\lambda,\Delta\mu}_{kl}(\Omega)$.
\end{theorem}

The proof proceeds analogously to the univariate case and is provided in Appendix~\ref{app:proofs}.

For a fixed spline space $\mathcal{Z}_{kl}^{\Delta\lambda,\Delta\mu}(\Omega)$, specified by the spline degrees $k,\ l$ and the knot sequences $\Delta\Lambda,\ \Delta M$, and for given values of $\mathrm{d}_x,\ \mathrm{d}_y$ and $\rho_x,\ \rho_y$, the minimizer $\boldsymbol{\theta}^*_{\rho_x,\rho_y}$ is computed numerically using a quasi-Newton method. The gradient of the penalized negative log-likelihood is derived analytically; its explicit expression is given by \eqref{eq:grad2D} in Appendix~\ref{app:proofs}. The integrals appearing in the objective function and its gradient are evaluated by numerical quadrature.

\begin{algorithm}[H]
\caption{Penalized likelihood estimation of a bivariate density}
\label{alg:bivariate}
\begin{algorithmic}[1]
\Require Observations $(x_1,y_1),\ldots,(x_N,y_N)$, penalization parameters $\rho_x$, $\rho_y$,
difference orders $\mathrm{d}_x$, $\mathrm{d}_y$, and tolerance $\varepsilon > 0$.
\State Compute $\mathbf{b}$, $\mathbf{c}$, $\mathbf{d}$ and form $\textbf{h}$. 
\State Construct
$\mathbf{P}_{\mathrm{d}_x}=\mathbf{D}_{\mathrm{d}_x}^{\top}\mathbf{D}_{\mathrm{d}_x}$, $\mathbf{P}_{\mathrm{d}_y}=\mathbf{D}_{\mathrm{d}_y}^{\top}\mathbf{D}_{\mathrm{d}_y}$ together with $\widetilde{\mathbf{P}}_{\mathrm{d}_x} = \textbf{I}_{h+l}\otimes\textbf{P}_{\mathrm{d}_x}$, $\widetilde{\mathbf{P}}_{\mathrm{d}_y} = \textbf{P}_{\mathrm{d}_y}\otimes\textbf{I}_{g+k}$ and form $\mathbb{P}_{\rho_x,\rho_y}$.
\State Set $\boldsymbol{\theta}^{(0)}=\mathbf{0}$ and $r=0$.
\Repeat
\State Evaluate $l_{\rho_x,\rho_y}(\boldsymbol{\theta}^{(r)})$ and $\nabla_{\boldsymbol{\theta}}\,l_{\rho_x,\rho_y}(\boldsymbol{\theta}^{(r)})$.
\State Compute $\boldsymbol{\theta}^{(r+1)}$ by a quasi-Newton step.
\State Set $r\leftarrow r+1$.
\Until{$\|\nabla_{\boldsymbol{\theta}}\,l_{\rho_x,\rho_y}(\boldsymbol{\theta}^{(r)})\|_2 <\varepsilon$}
\State Set $\boldsymbol{\theta}^*_{\rho_x,\rho_y} = \boldsymbol{\theta}^{(r)}$ and
$$
s_{kl}^*(x,y) = \boldsymbol{\Phi}(x,y)\boldsymbol{\theta}^*
$$
\State Return
\[
p_{\boldsymbol{\theta}^*_{\rho_x,\rho_y}}(x,y) = \frac{\exp\bigl(s_{kl}^*(x,y)\bigr)} {\displaystyle\iint_\Omega\exp\bigl(s_{kl}^*(t,s)\bigr)\,\mathrm{d}t\mathrm{d}s}.
\]
\end{algorithmic}
\end{algorithm}

The quasi-Newton method is terminated when the Euclidean norm of the gradient falls below a prescribed tolerance. Upon convergence, the resulting  coefficient vector $\boldsymbol{\theta}$ provides numerical approximations of the unique solution characterized in Theorem~\ref{thm:existence_uniqueness_bivariate}.

\section{Selection of penalization parameters}\label{rhoselection}
The proposed penalized maximum likelihood framework depends on several parameters controlling the flexibility and smoothness of the resulting density estimate. The flexibility of the spline representation is primarily determined by the number and position of the inner knots, whereas smoothness is controlled by the penalization parameters. Since the placement of knots may substantially affect the quality of the spline approximation, several adaptive knot selection strategies have been proposed, e.g., in \cite{Michel2021, Goepp2025}. In this work, however, we restrict our attention to equidistantly spaced inner knots \cite{eilers2021} and focus on the effects of the number of inner knots and the penalization parameters. 

A large number of inner knots increases the flexibility of the spline representation and allows more complex density structures to be captured. However, insufficient penalization may lead to local oscillations and overfitting, whereas excessive penalization may oversmooth important structural features of the density.

To determine the optimal penalization parameters in a data-driven manner, we employ $K$-fold cross-validation. The observed sample is randomly partitioned into $K$ approximately equal folds. For each fold, the density function is estimated as described in Section \ref{PML} using the remaining $K-1$ subsets. The resulting estimate is subsequently evaluated on the omitted fold.

In the univariate setting, the corresponding cross-validation criterion is defined as
\begin{equation}\label{CV_1D}
CV(\rho) = \frac{1}{K} \sum_{k=1}^{K}
\left[ - \sum_{x_i\in V_k} \ln \hat f_{\rho}^{(-k)}(x_i) \right],
\end{equation}
where $V_k$ denotes the validation set in the $k$-th fold and $\hat f_{\rho}^{(-k)}$ stands for the estimate of the density function obtained from the corresponding training sample. 
For each fixed candidate value of the penalization parameter $\rho$, the corresponding cross-validation criterion is obtained by averaging the negative log-likelihood values over all $K$ validation folds according to \eqref{CV_1D}. This procedure is repeated over a predefined grid of candidate values $\mathcal{G} = \{\rho_1,\ldots,\rho_M\}$. The optimal penalization parameter is subsequently selected as the candidate yielding the smallest cross-validation criterion
$$
\rho^{*} = \operatorname*{argmin}_{\rho\in\mathcal{G}} CV(\rho).
$$

For bivariate densities, the corresponding cross-validation criterion is given analogously by
\begin{equation}\label{CV_2D}
CV(\rho_x,\rho_y) = \frac{1}{K}\sum_{k=1}^{K}\left[-\sum_{(x_i,y_i)\in V_k}\ln\hat f_{\rho_x,\rho_y}^{(-k)}(x_i,y_i)\right].
\end{equation}
where $V_k$ again stands for the validation set in the $k$-th fold, as in the univariate case, and $\hat f_{\rho_x,\rho_y}^{(-k)}$ denotes the corresponding bivariate density estimate from the training sample. For each fixed candidate pair of penalization parameters $(\rho_x,\rho_y)$, the corresponding cross-validation criterion is obtained by averaging the negative log-likelihood values over all $K$ validation folds according to \eqref{CV_2D}. This procedure is repeated over a predefined dense grid of candidate pairs $\mathcal{G}=\{(\rho_x^{1},\rho_y^{1}),\ldots,(\rho_x^{M},\rho_y^{M})\}$. The optimal penalization parameters are subsequently selected as the candidate pair yielding the smallest cross-validation criterion,
$$
(\rho_x^{*},\rho_y^{*})=
\operatorname*{argmin}_{(\rho_x,\rho_y)\in\mathcal{G}}
CV(\rho_x,\rho_y).
$$
The role of penalization becomes even more important in the bivariate setting due to the substantially increased flexibility of the spline representation and the potentially different structural complexity in the two coordinate directions.

Moreover, since the flexibility of the resulting density estimate is jointly determined by the number of inner knots and the penalization parameter in the univariate case, or parameters in the bivariate case, the final model is selected as the combination  
yielding the smallest cross-validated criterion \eqref{CV_1D} or \eqref{CV_2D}.

\section{Simulation study}\label{Simulation}
The aim of the simulation study is to investigate the properties of the proposed methodology for densities of increasing complexity and to examine the influence of spline flexibility and the penalization parameters on the resulting estimates.

\subsection{Univariate scenario}\label{Simulation1D}
First, we considered the density of the standard normal distribution
$$
f_1(x) = \frac{1}{\sqrt{2\pi}} 
\exp\left(-\frac{x^2}{2}\right)
$$
with the domain $I_1 = [-4,4]$. In the second setting, we considered a mixture of two univariate normal distributions
$$
f_2(x) = w_1\,\mathcal{N}(x \mid \mu_1, \sigma_1^2) + w_2\,\mathcal{N}(x \mid \mu_2, \sigma_2^2),
$$
where $\mu_1 = -2,\, \sigma_1 = 0.4$ and $\mu_2 = 2,\, \sigma_2 = 0.7$. We assumed weights $w_1 = w_2 = 1/2$ and 
the domain was set to $I_2=[-4,4]$. In the third scenario, we considered a mixture of three univariate normal distributions
$$
f_3(x) = \sum_{k=1}^3 w_k \,\mathcal{N}(x \mid \mu_k, \sigma_k^2),
$$
with $\mu_1 = -3$, $\sigma_1 = 0.3$, $\mu_2 = 0$, $\sigma_2 = 0.2$, $\mu_3 = 2$, $\sigma_3 = 0.2.$
The weights were set to $w_1=w_2=w_3=1/3$ and we considered the domain $I_3 = [-4,3]$. The asymmetric domain was selected according to the locations of the outer mixture components, centered at $\mu_1=-3$ and $\mu_3=2$, so as to cover the relevant range of the density on both sides. The three theoretical densities are shown in Figure \ref{theoretical_dens}. The density $f_1(x)$ has a simple smooth unimodal structure, whereas $f_2(x)$ and $f_3(x)$ represent more complex bimodal and trimodal structures, respectively.

\begin{figure}[h]
    \centering
    \includegraphics[width=0.3\textwidth]{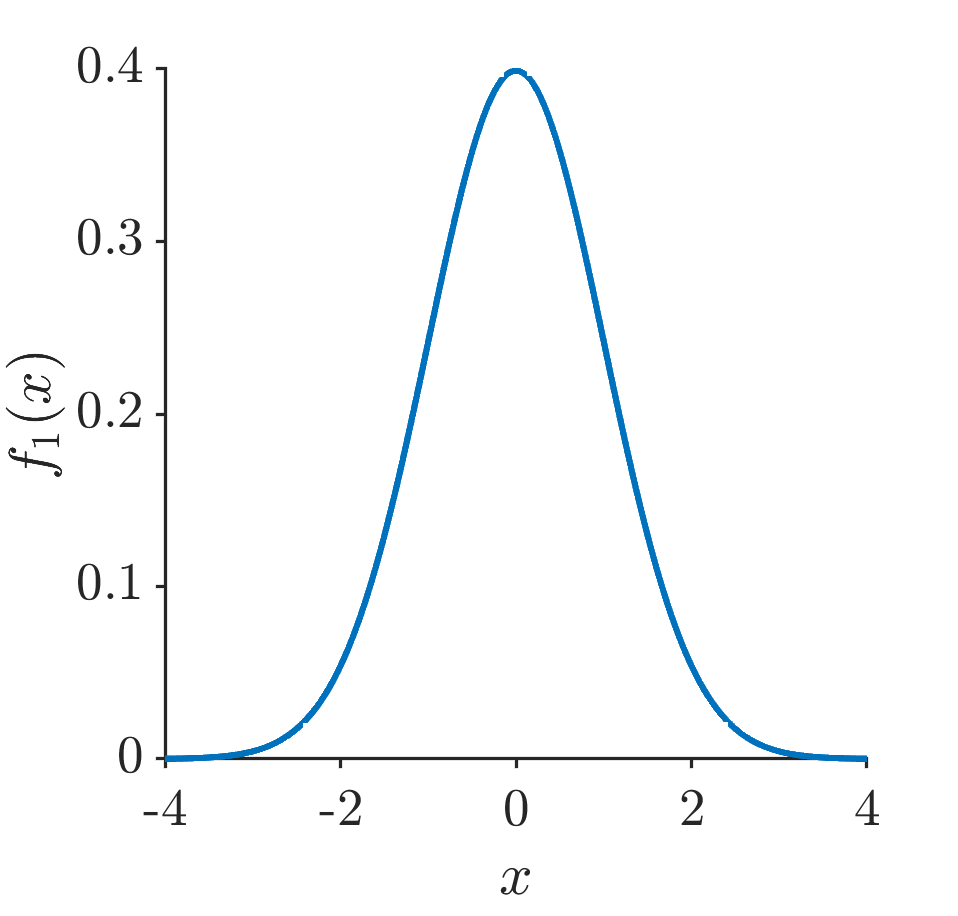}
    \includegraphics[width=0.3\textwidth]{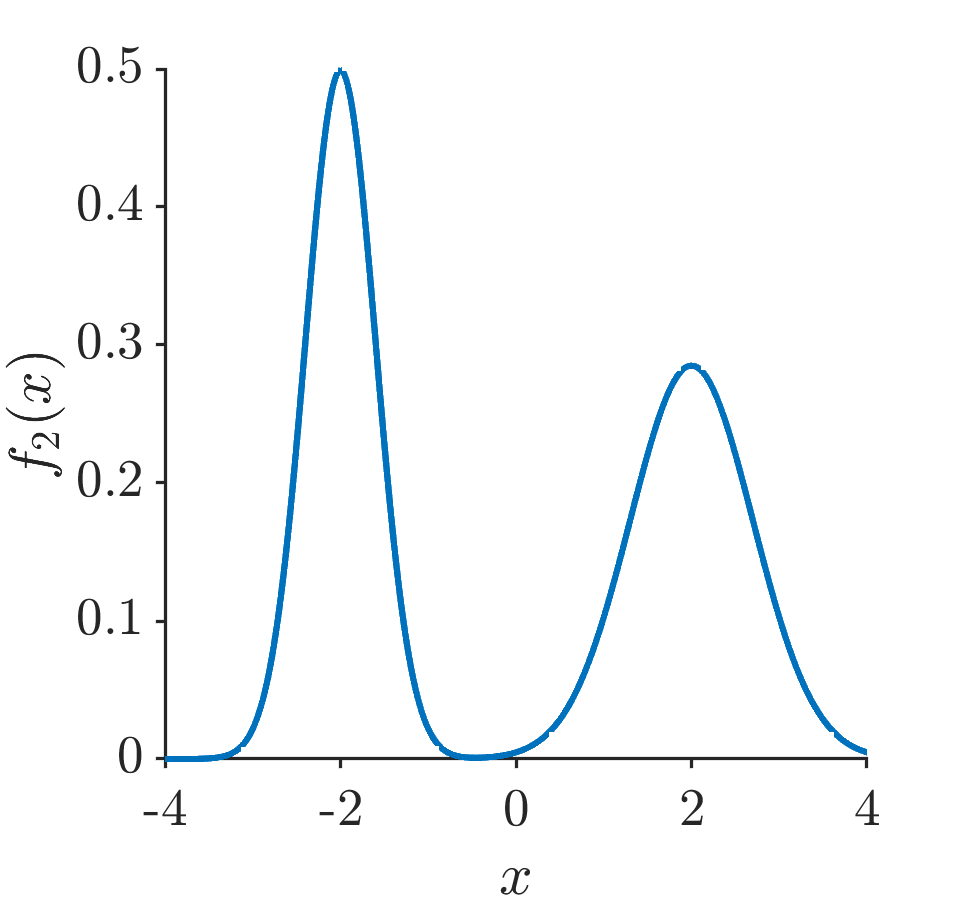}
    \includegraphics[width=0.3\textwidth]{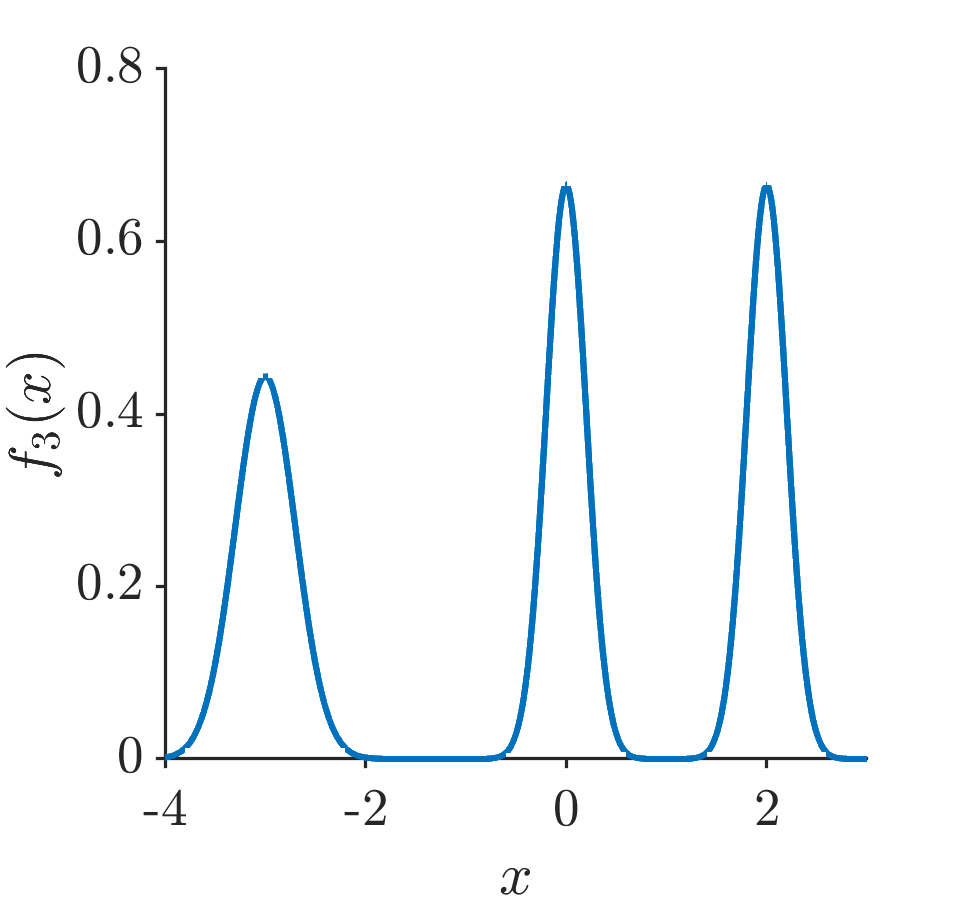}
    \caption{Theoretical densities $f_1(x)$ (left), $f_2(x)$ (middle) and $f_3(x)$ (right).}
    \label{theoretical_dens}
\end{figure}

For each scenario, samples of sizes $N_1=3000$ and $N_2=5000$ were generated. We repeated each simulation 100 times and estimated the density functions via penalized maximum likelihood approach as described in Section \ref{PML1D} with the following setting. Given the shapes of the densities, we employed quadratic splines, i.e., $k=2$, and first-order differences, i.e., $\mathrm{d}=1$. Equidistant knot sequences, i.e., sequences with equally spaced inner knots, were used for 
$g = 1,2,\dots,10$. For each simulation run and each value of $g$, the optimal value of the penalization parameter $\rho$ was determined by $K$-fold cross validation, as described in Section \ref{rhoselection}, using $K_1=5$ and $K_2=10$ folds. Tables \ref{opt_rho_f1}, \ref{opt_rho_f2} and \ref{opt_rho_f3} report the mean optimal values of $\rho$ over the 100 simulation runs for the densities $f_1(x)$, $f_2(x)$, and $f_3(x)$, respectively.

\begin{table}[h]
\centering
\small
\setlength{\tabcolsep}{5pt}
\begin{tabular}{cc|cccccccccc}
\toprule
 &  & \multicolumn{10}{c}{$g$} \\
\cmidrule(lr){3-12}
$N$ & $K$ & 1 & 2 & 3 & 4 & 5 & 6 & 7 & 8 & 9 & 10 \\
\midrule
3000 & 5  & 0.036 & 0.072 & 0.149 & 0.192 & 0.225 & 0.250 & 0.272 & 0.294 & 0.326 & 0.358 \\
3000 & 10 & 0.037 & 0.074 & 0.149 & 0.191 & 0.235 & 0.258 & 0.286 & 0.303 & 0.326 & 0.366 \\
5000 & 5  & 0.041 & 0.079 & 0.162 & 0.203 & 0.257 & 0.285 & 0.304 & 0.320 & 0.335 & 0.369 \\
5000 & 10 & 0.039 & 0.079 & 0.164 & 0.209 & 0.265 & 0.298 & 0.324 & 0.330 & 0.350 & 0.387 \\
\bottomrule
\end{tabular}
\caption{Mean values of optimal penalization parameters $\rho$ for estimating the density $f_1(x)$ using different numbers of sampled values ($N_1 = 3000$, $N_2 = 5000$) and different number of folds $K$ used in \eqref{CV_1D} ($K_1 = 5$, $K_2 = 10$) (rows) for different number of inner knots $g$ (columns).}
\label{opt_rho_f1}
\end{table}

The mean values of the optimal penalization parameter exhibit different patterns across the three simulation scenarios. For the unimodal density $f_1(x)$, the mean optimal value of $\rho$ increases steadily with the number of inner knots, as shown in Table~\ref{opt_rho_f1}. For $f_2(x)$, the substantially larger mean value obtained for $g=1$ is followed by a sharp decrease for $g=2$, after which the mean values generally increase, although not monotonically, see Table~\ref{opt_rho_f2}. In contrast, for $f_3(x)$, no clear monotonic relationship between $g$ and the mean optimal value of $\rho$ is observed, see Table~\ref{opt_rho_f3}. Since a separate equidistant knot sequence is constructed for each value of $g$, changing $g$ affects not only the dimension of the spline space but also the locations of the knots. Therefore, a monotonic dependence of the mean optimal penalization parameter on $g$ should not be expected in general. Nevertheless, for each fixed value of $g$, the mean optimal values are very similar across the considered sample sizes and numbers of cross-validation folds.

Since the true reference densities $f_i(x)$, $i=1,2,3$ are known in the simulation study, we evaluated the estimation accuracy using the normalized integrated squared error (NISE). For each simulation run $r$ and each fixed number of inner knots $g$, the optimal penalization parameter $\rho^*_{r,g}$ was selected by cross-validation. The resulting coefficient vector $\mathbf{z}^*_{r,g}$ determines the estimated spline
\[
s^*_{k,r,g}(x) = 
\mathbf{Z}_{k+1}(x)\mathbf{z}^*_{r,g},
\]
obtained by the penalized maximum likelihood procedure. Its inverse clr transformation defines the corresponding estimated compositional spline
\[
\widehat{f}_{r,g}(x) = \operatorname{clr}^{-1}
\bigl(s^*_{k,r,g}\bigr)(x) =_{\mathcal{B}^2(I)}
\exp\bigl(s^*_{k,r,g}(x)\bigr).
\]
The normalized integrated squared error is defined in $\mathcal{B}^2(I)$ as
\[
\operatorname{NISE}_{r,g} = \frac{1}{b-a} 
\left\| f(x)\ominus\widehat{f}_{r,g}(x) \right\|_{\mathcal{B}^2(I)}^2, \qquad r=1,\ldots,100.
\]
Using the isometry of the clr transformation, it can equivalently be evaluated in $L_0^2(I)$ as
\[
\operatorname{NISE}_{r,g} = \frac{1}{b-a}
\left\| \operatorname{clr}(f)(x)-s^{*}_{k,r,g}(x)
\right\|_{L_0^2(I)}^2, \qquad r=1,\ldots,100.
\]
For each density and each combination of the sample size $N$ and the number of cross-validation folds $K$, the values of $g$ were compared using the median NISE over the 100 simulation runs. The distributions of the resulting NISE values for the penalized likelihood estimates of $f_1(x)$, $f_2(x)$, and $f_3(x)$ are displayed in Figures~\ref{NISE_f1}, \ref{NISE_f2}, and \ref{NISE_f3}, respectively.

\begin{figure}[h]
    \centering
    \includegraphics[width=0.4\textwidth]{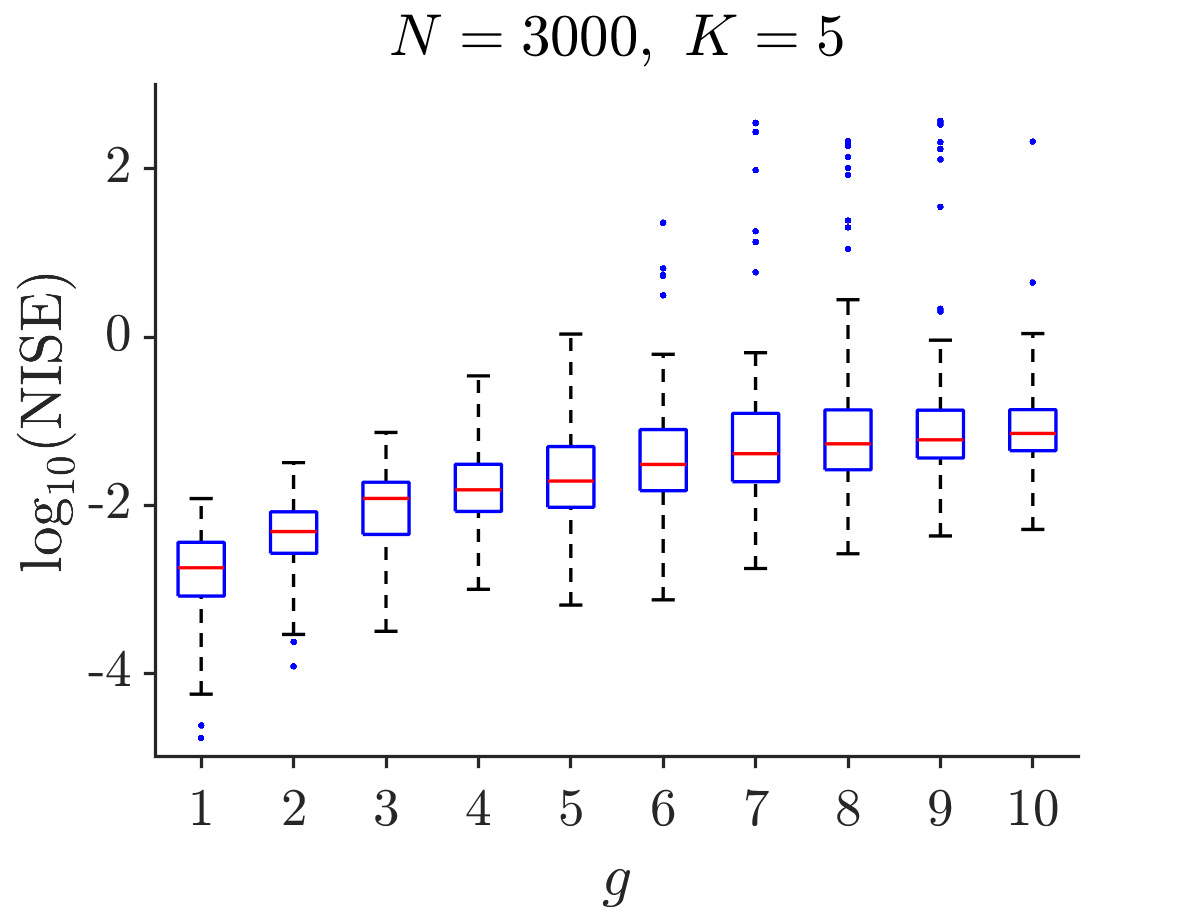}\qquad
    \includegraphics[width=0.4\textwidth]{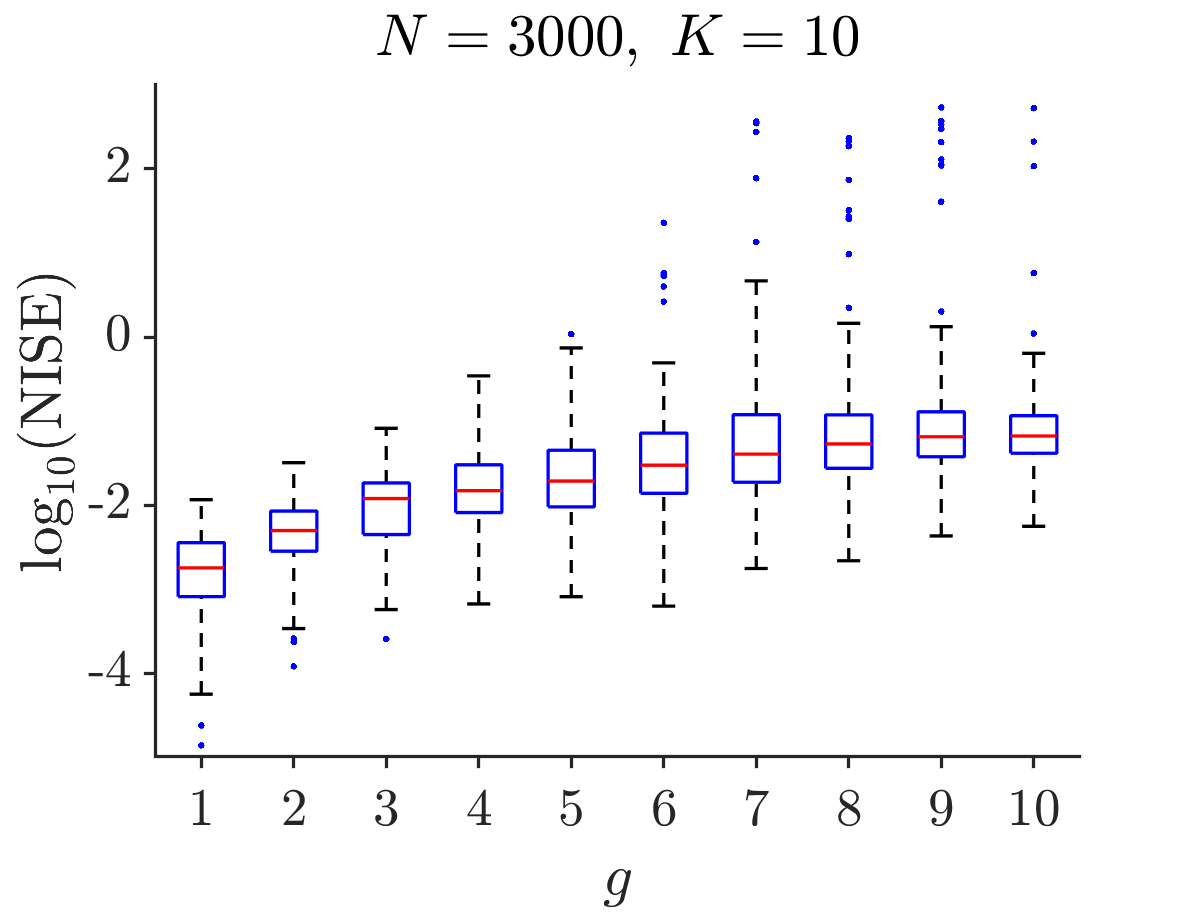}\\[0.5cm]
    \includegraphics[width=0.4\textwidth]{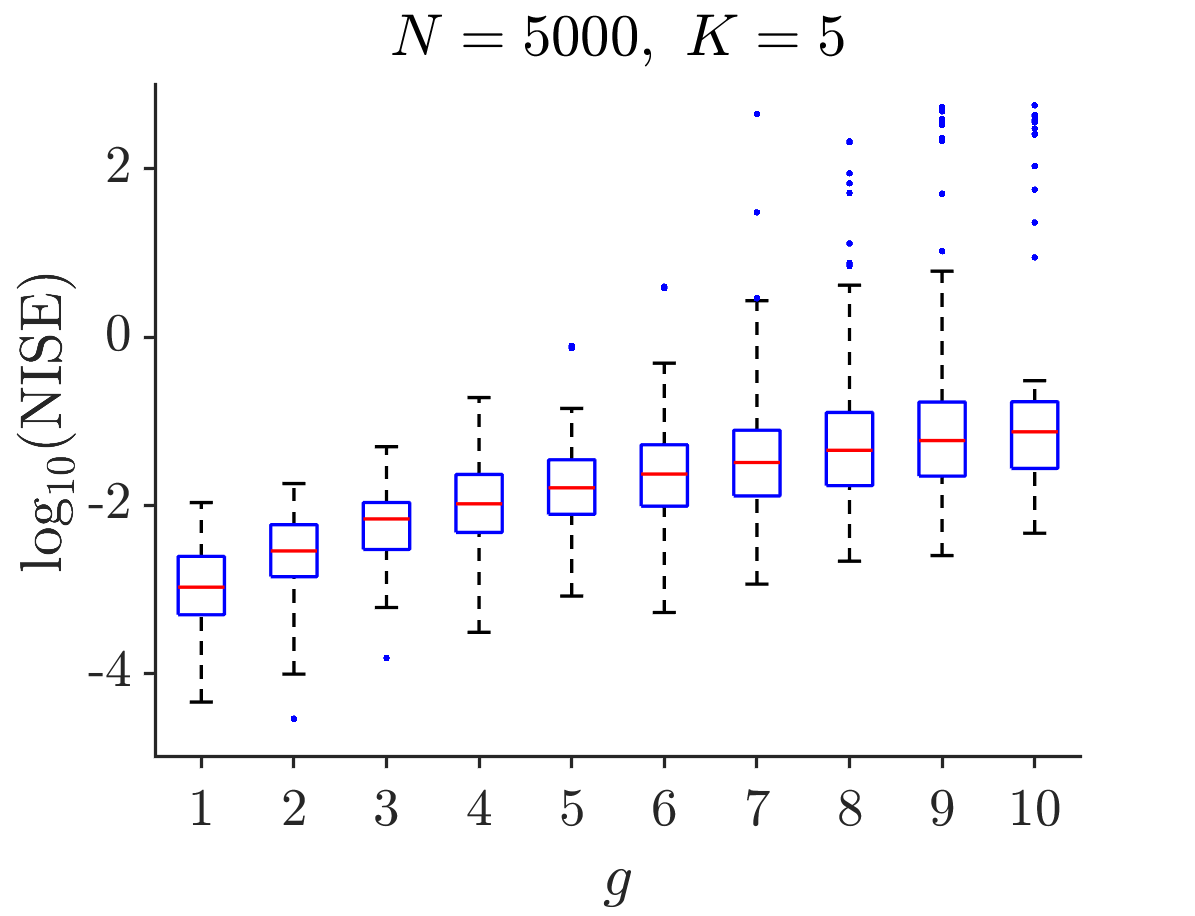}\qquad
    \includegraphics[width=0.4\textwidth]{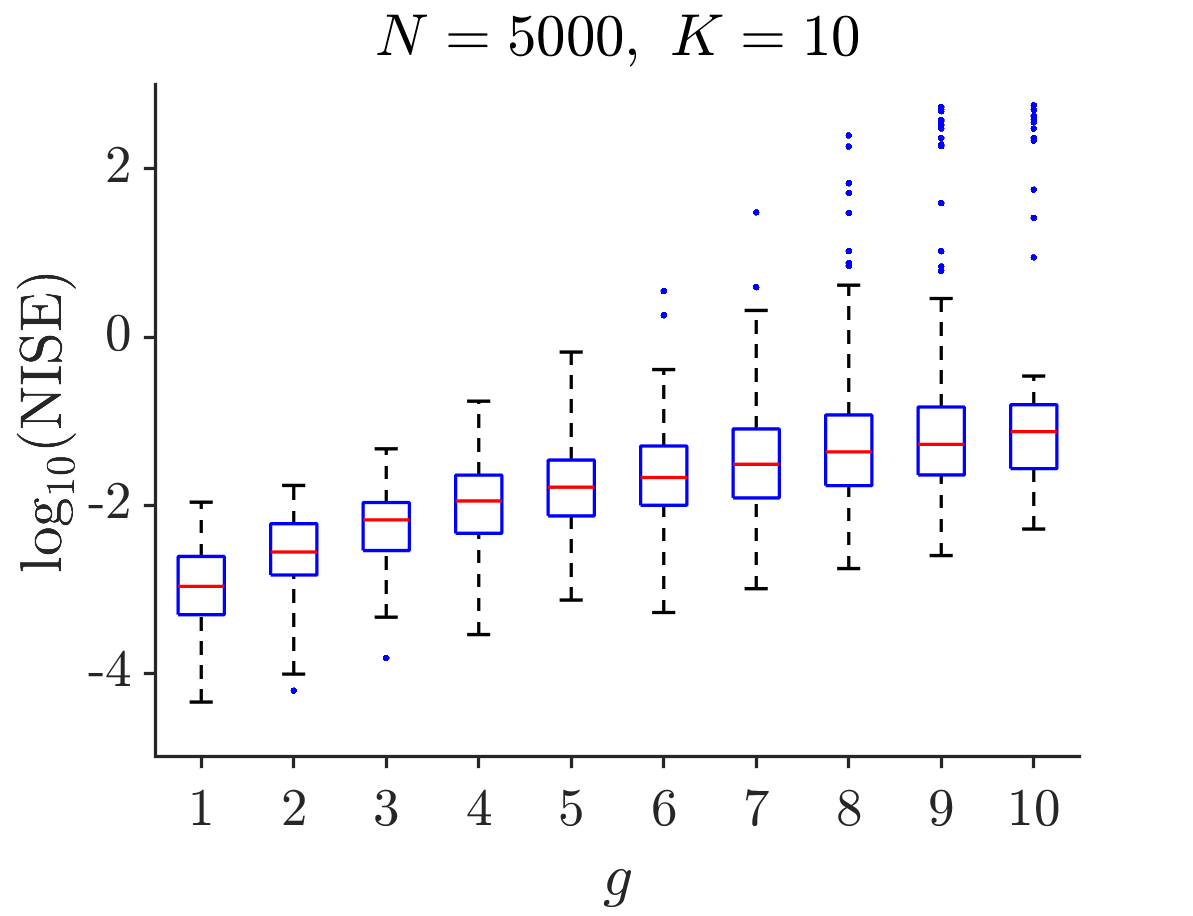}
    \caption{Boxplots of NISE values for the penalized likelihood estimates of $f_1(x)$ for different numbers of inner knots $g$ and the considered combinations of the sample size $N$ and the number of cross-validation folds $K$.}
    \label{NISE_f1}
\end{figure}

The NISE values exhibit different patterns across the three considered densities. For $f_1(x)$, the median NISE increases  monotonically with the number of inner knots, and the smallest values are obtained for $g=1$, see Figure~\ref{NISE_f1}. For $f_2(x)$, the dependence of the median NISE on $g$ is non-monotonic, with the lowest values attained for $g=7$, see Figure~\ref{NISE_f2}. For $f_3(x)$, the median NISE remains relatively high for small values of $g$, decreases substantially as $g$ increases up to $g=8$, where the minimum is attained, and remains higher for $g=9$ and $g=10$, see Figure~\ref{NISE_f3}. In all three cases, the observed patterns are highly similar across the considered sample sizes and numbers of cross-validation folds.

Occasional extreme NISE values are observed mainly for $f_1(x)$ and $f_2(x)$, particularly for larger values of $g$. These values may be associated with estimation inaccuracies in sparsely sampled tail regions. Since NISE is evaluated after the clr transformation, even small differences in areas of low density values may be amplified on the logarithmic scale and lead to large errors.

For comparison, the densities were also estimated using Gaussian kernel density estimators with bandwidths selected by Silverman's rule of thumb, Scott's rule, and the Sheather-Jones method. The resulting NISE distributions are shown in Figure \ref{NISE_kde}.

\begin{figure}[h]
    \centering
    \includegraphics[width=0.3\textwidth]{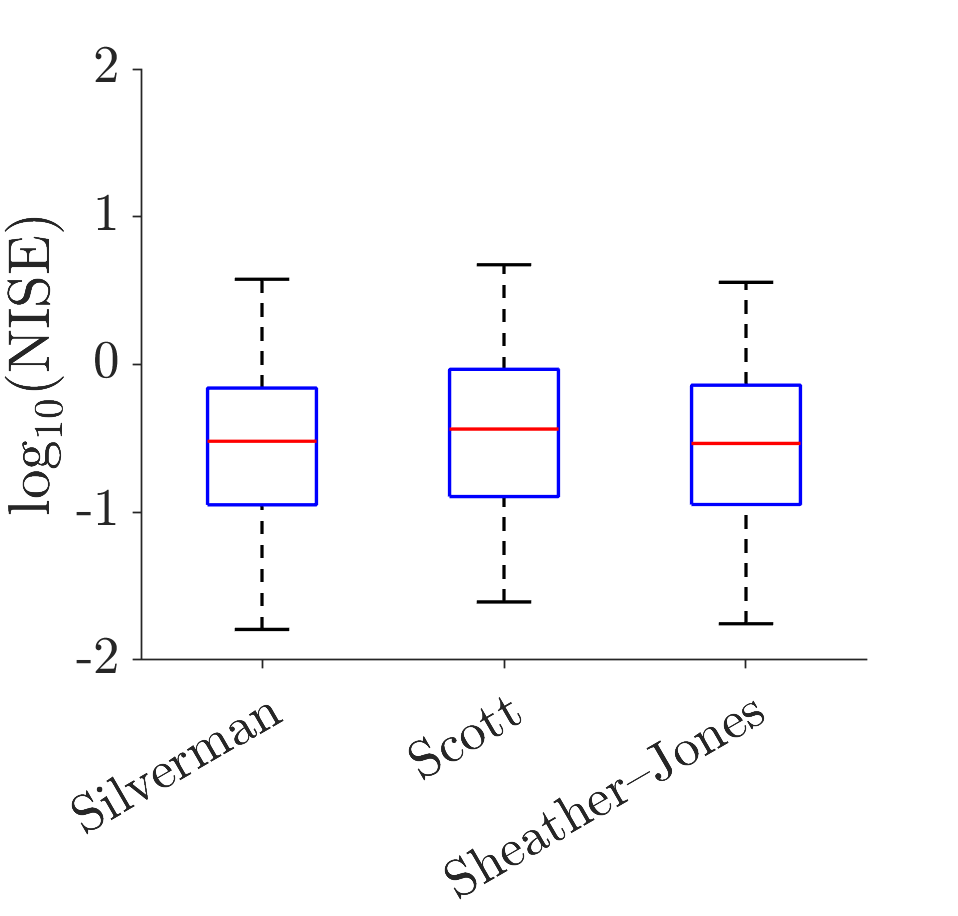}
    \includegraphics[width=0.3\textwidth]{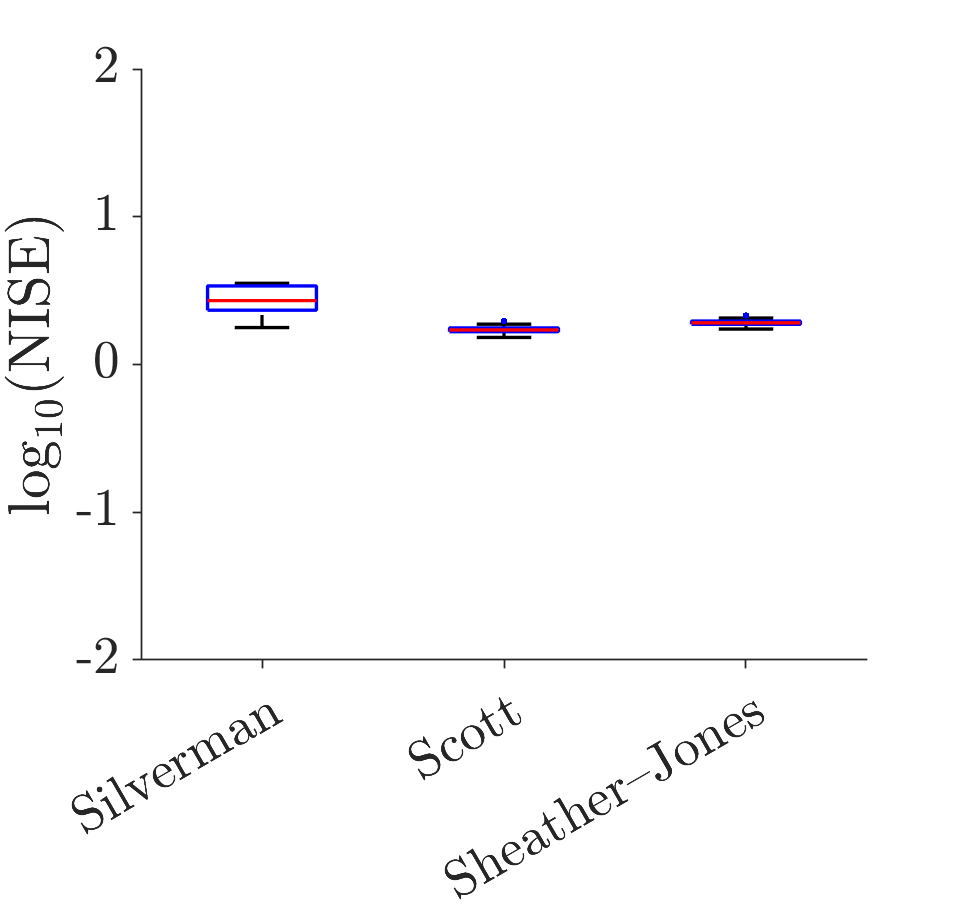}
    \includegraphics[width=0.3\textwidth]{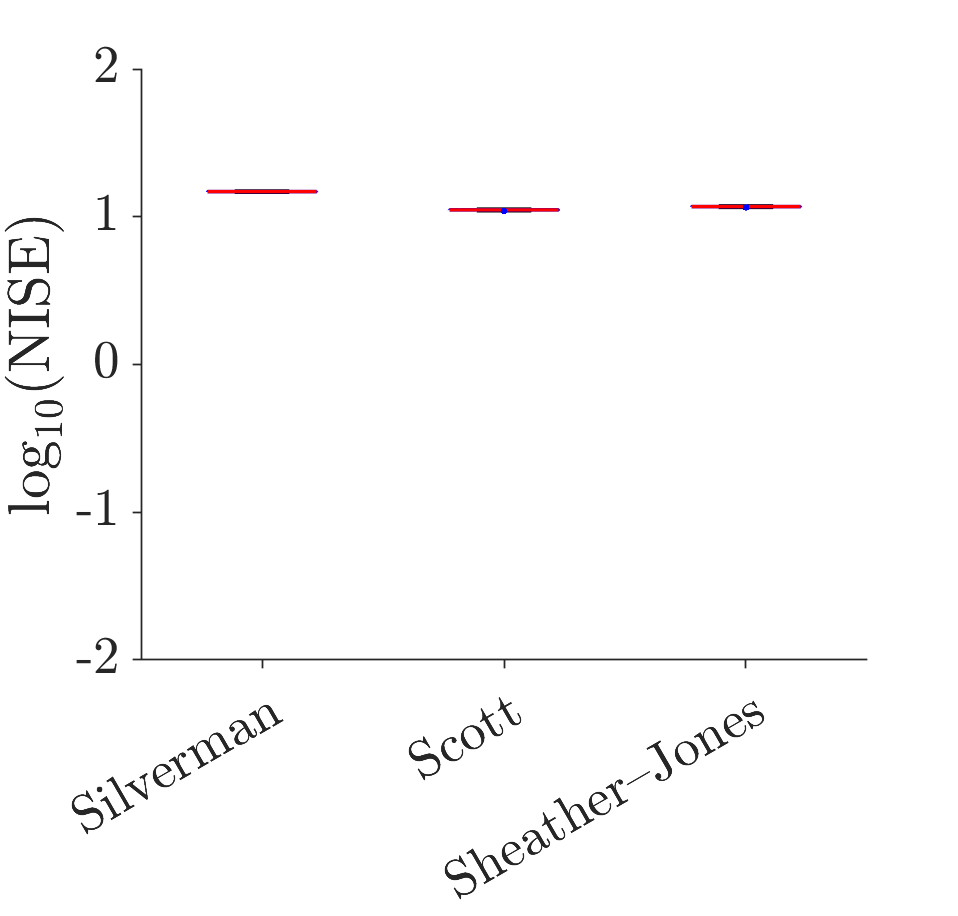}
    \caption{Boxplots of NISE values for kernel density estimates of $f_1(x)$, $f_2(x)$, and $f_3(x)$ (from left to right) using Silverman's rule of thumb, Scott's rule, and the Sheather--Jones method for bandwidth selection, with $N=3000$.}
    \label{NISE_kde}
\end{figure}

The boxplots in Figure~\ref{NISE_kde} indicate that the kernel density estimates exhibit lower variability and no comparably extreme NISE values. However, their median NISE values are higher than the minimum median NISE values obtained by the penalized likelihood  estimator for all three densities, with particularly pronounced differences for the multimodal densities $f_2(x)$ and $f_3(x)$. For a direct numerical comparison, the median NISE values for the penalized likelihood and kernel density estimates are reported in Tables~\ref{median_NISE_f1}, \ref{median_NISE_f2}, \ref{median_NISE_f3} and \ref{median_NISE_kde}.

\begin{table}[h]
\centering
\small
\setlength{\tabcolsep}{5pt}
\begin{tabular}{cc|cccccccccc}
\toprule
 &  & \multicolumn{10}{c}{$g$} \\
\cmidrule(lr){3-12}
$N$ & $K$ & 1 & 2 & 3 & 4 & 5 & 6 & 7 & 8 & 9 & 10 \\
\midrule
3000 & 5  & \textbf{0.011}& 0.029& 0.072& 0.091& 0.116& 0.183& 0.245& 0.321& 0.358& 0.427 \\
3000 & 10 & \textbf{0.011}& 0.03& 0.072& 0.089& 0.116& 0.179& 0.242& 0.32& 0.39& 0.398\\
5000 & 5 &\textbf{0.006}& 0.017& 0.041& 0.062& 0.096& 0.14& 0.192& 0.269& 0.351& 0.446\\
5000 & 10 & \textbf{0.006}& 0.017& 0.04& 0.067& 0.098& 0.128& 0.183& 0.257& 0.317& 0.449 \\
\bottomrule
\end{tabular}
\caption{Median NISE values of $f_1(x)$ using different numbers of sampled values ($N_1 = 3000$, $N_2 = 5000$) and different number of folds used in \eqref{CV_1D} ($K_1 = 5$, $K_2 = 10$) (rows) for different numbers of inner knots (columns) with the highlighted minimal NISE for each setting.}
\label{median_NISE_f1}
\end{table}

\begin{table}[h]
\centering
\small
\setlength{\tabcolsep}{6pt}
\begin{tabular}{c|ccc|ccc}
\toprule
 & \multicolumn{3}{c}{$N=3000$} & \multicolumn{3}{c}{$N=5000$} \\
\cmidrule(lr){2-4} \cmidrule(lr){5-7}
Density & Silverman & Scott & Sheather-Jones & Silverman & Scott & Sheather-Jones \\
\midrule
$f_1(x)$   & 0.302 & 0.366 & \textbf{0.292} & \textbf{0.206} & 0.259 & 0.213 \\
$f_2(x)$  & 2.705 & \textbf{1.714} & 1.914 & 2.368 & \textbf{1.391} & 1.571\\
$f_3(x)$ & 14.81 & \textbf{11.133} & 11.722 & 14.023 & \textbf{10.052} & 10.668 \\
\bottomrule
\end{tabular}
\caption{Median NISE values for kernel density estimators with different bandwidth selection rules with highlighted minimum for each setting.}
\label{median_NISE_kde}
\end{table}

Based on the median NISE values, the best-performing numbers of inner knots are $g=1$ for $f_1(x)$, $g=7$ for $f_2(x)$, and $g=8$ for $f_3(x)$. These choices are consistent across all considered combinations of the sample size $N$ and the number of cross-validation folds $K$. Thus, within the considered simulation settings, the multimodal densities require richer spline representations than the unimodal density $f_1(x)$. Examples of the resulting estimates of $f_1(x)$, $f_2(x)$, and $f_3(x)$ are shown in Figures~\ref{final_comp_f1}, \ref{final_comp_f2} and \ref{final_comp_f3}, respectively. Figures~\ref{final_comp_f2} and~\ref{final_comp_f3} also illustrate the substantially larger NISE values obtained by kernel density estimation for the multimodal densities, particularly for $f_3(x)$. On the original density scale, the kernel estimates smooth the local modal structure and the low-density regions between the modes. These discrepancies become considerably more pronounced after the clr transformation, as shown in the right panels of the figures. Since NISE is computed in $L_0^2(I)$ as the squared norm of the difference between the clr-transformed true and estimated densities, the resulting discrepancies on the logarithmic scale lead to large NISE values. This effect is especially pronounced for $f_3(x)$, which consists of three narrow and well-separated components.

\begin{figure}[h]
    \centering
    \includegraphics[width=0.3\textwidth]{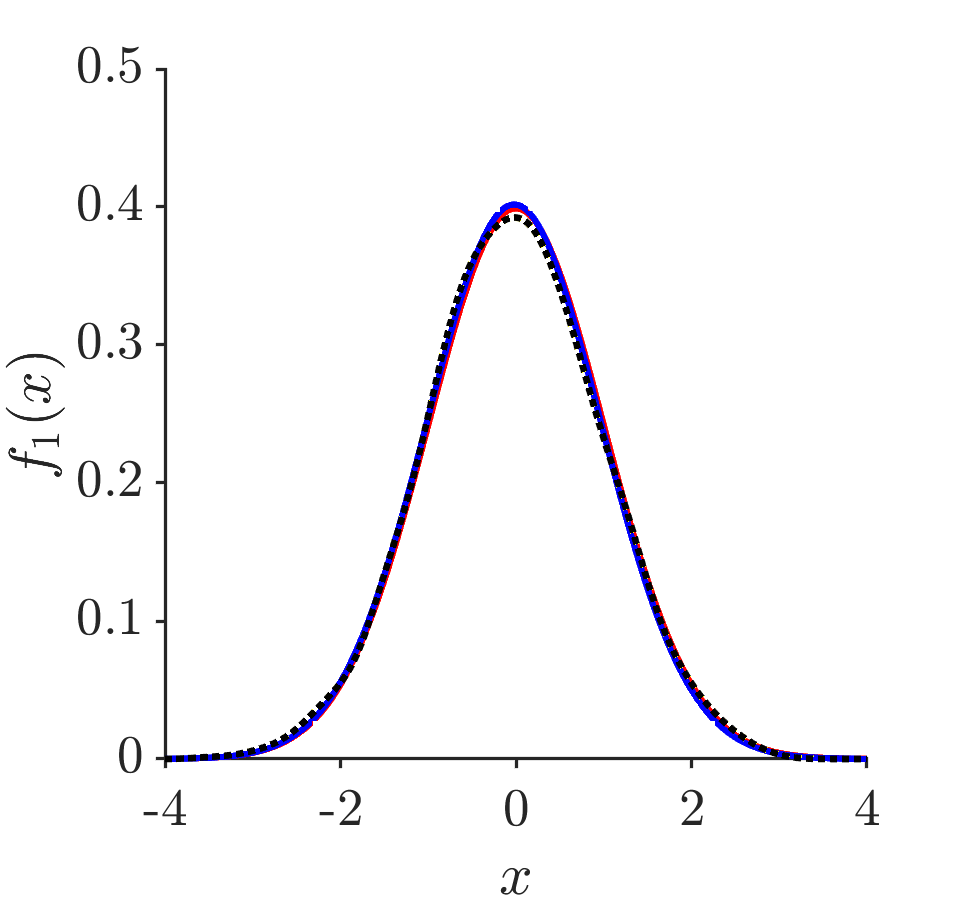}
    \includegraphics[width=0.3\textwidth]{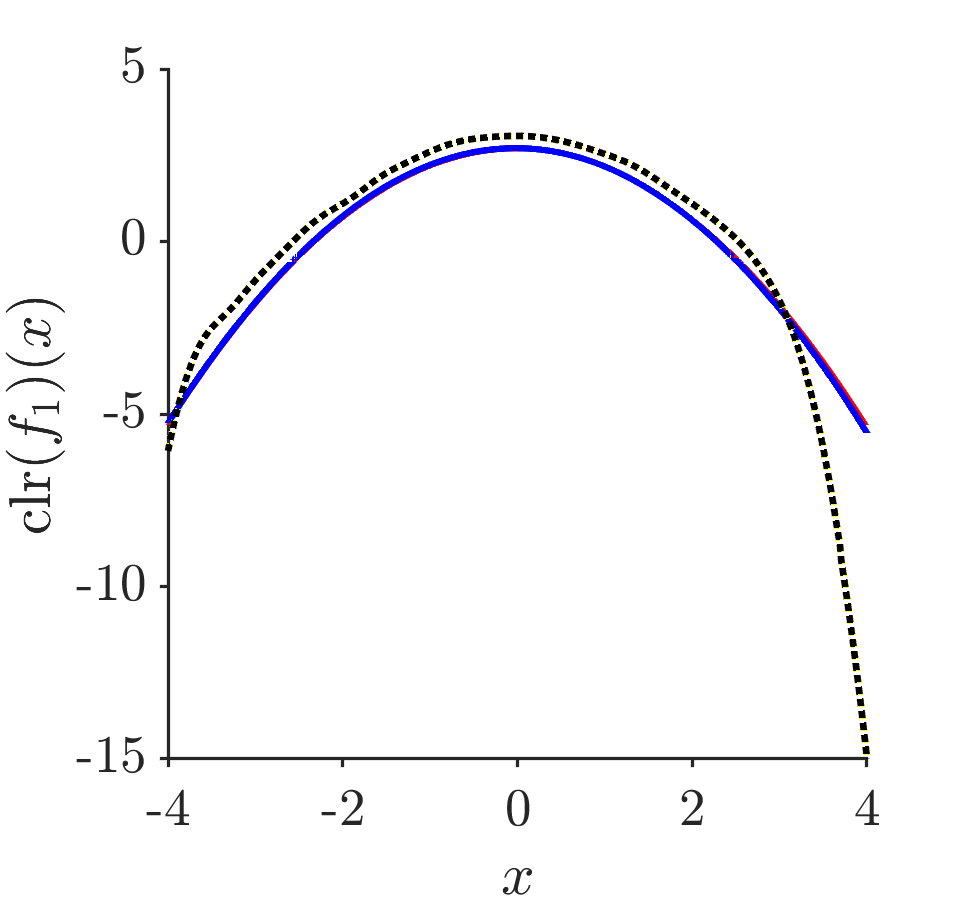}
    \caption{Comparison of the true density $f_1(x)$ (red), the penalized likelihood estimate constructed with $N=3000$ using $g=1$ and $\rho = 0.036$ (blue), and the kernel density estimates with $N=3000$ obtained using Silverman's rule of thumb (dotted, green), Scott's rule (dotted, yellow), and the Sheather--Jones method (dotted, black). The three kernel density estimates coincide. The unit-integral representatives in $\mathcal{B}^2(I)$ are shown in the left panel and their clr representations in $L_0^2(I)$ are shown in the right panel.}
    \label{final_comp_f1}
\end{figure}

Overall, the univariate simulation study shows that the proposed penalized likelihood estimator can accurately represent densities with varying structural complexity, provided that the spline complexity and the penalization parameter are selected appropriately. For the best-performing number of inner knots, the proposed estimator achieves lower median NISE values than all considered kernel density estimators in each of the three simulation scenarios. Moreover, the estimated densities are represented directly by their spline coefficients, which can subsequently be used in functional data analysis methods.

\subsection{Bivariate scenario}\label{subsec:bivariate_sim}
In the bivariate setting, we first considered the standard normal density 
$$
f_1(x,y) = \frac{1}{2\pi} \exp\left(-\frac{x^2+y^2}{2}\right)
$$
with the domain $\Omega = [-3,3]\times[-3,3]$. For the second scenario, we considered a bimodal bivariate Gaussian mixture having density
$$
f_2(x,y) = \frac{1}{2}\mathcal{N}_2(\boldsymbol{\mu}_1,\boldsymbol{\Sigma}_1) + \frac{1}{2}\mathcal{N}_2(\boldsymbol{\mu}_2,\boldsymbol{\Sigma}_2),
$$
where $\boldsymbol{\mu}_1=(-1,\ -1)^{\top}$, $\boldsymbol{\mu}_2=(1,\ 1)^{\top}$ and 
$$
\boldsymbol{\Sigma}_1 =
\begin{pmatrix}
1.6 & 0.2\\
0.2 & 0.5
\end{pmatrix},
\quad
\boldsymbol{\Sigma}_2 =
\begin{pmatrix}
1.6 & -0.2\\
-0.2 & 0.5
\end{pmatrix}.
$$
The domain was set to $\Omega = [-4,4]\times[-3,3]$. The two theoretical densities are pictured in Figure \ref{theoretical_dens_2D}.

\begin{figure}[h]
    \centering
    \includegraphics[width=0.32\textwidth]{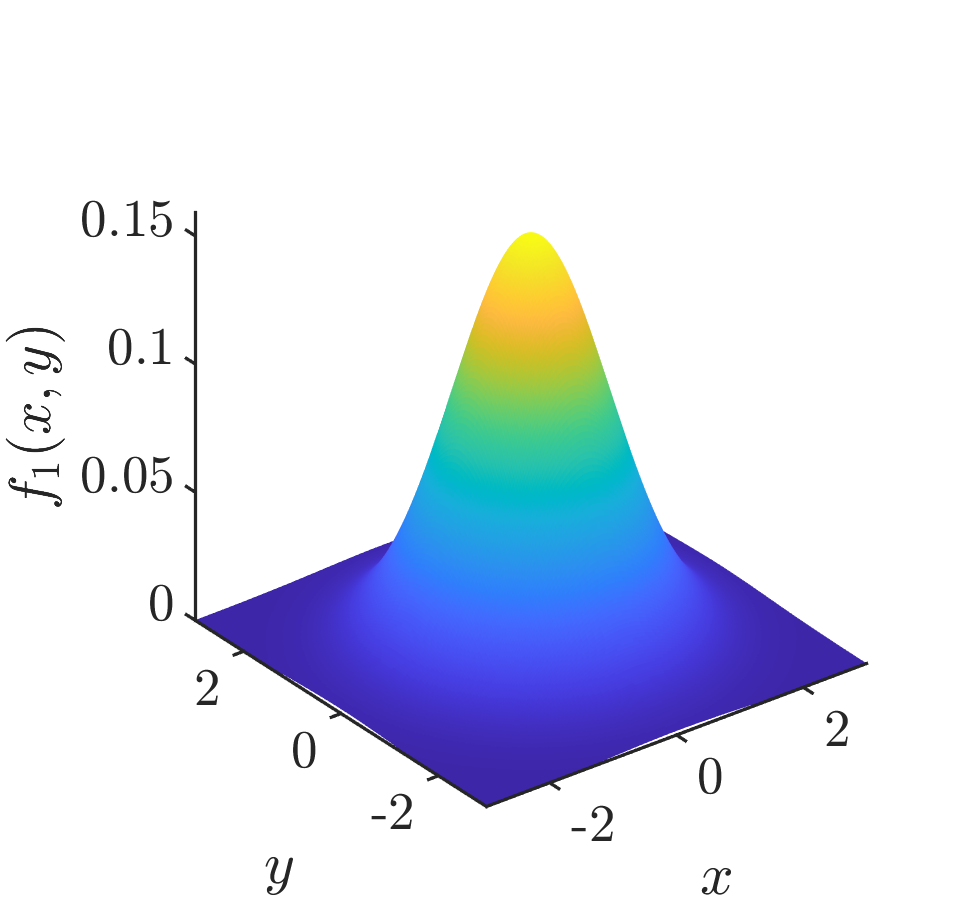}
    \includegraphics[width=0.32\textwidth]{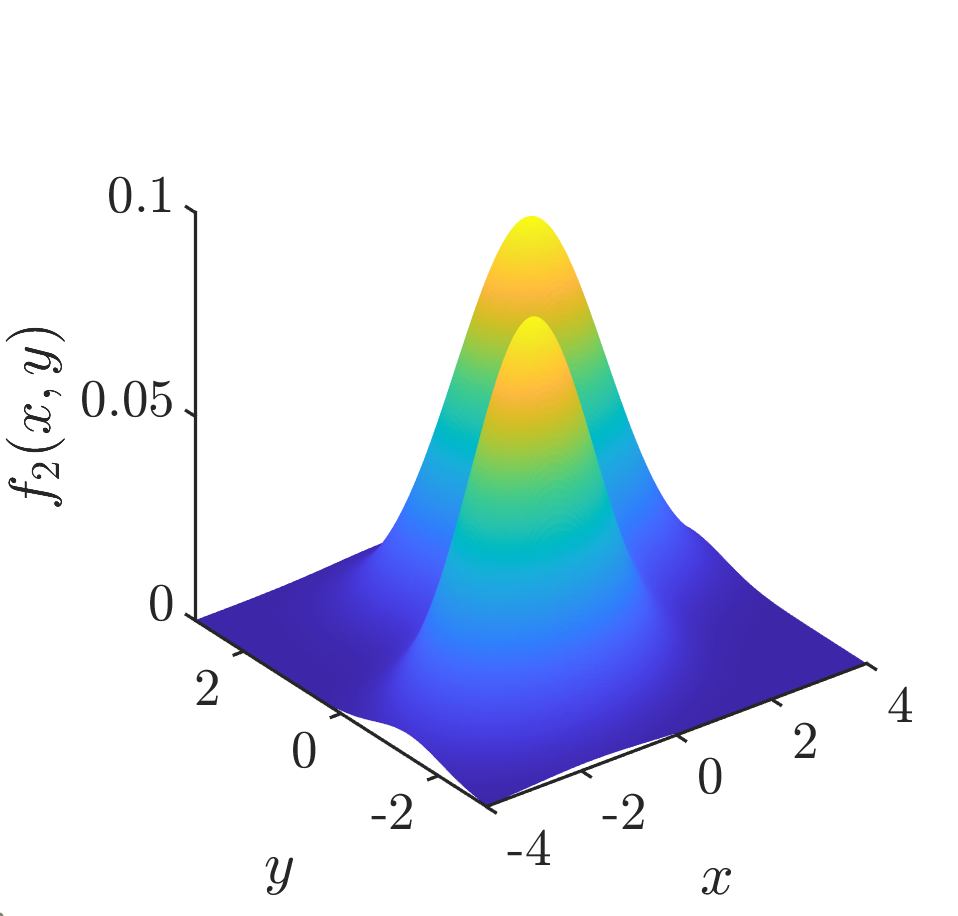}
    \caption{Theoretical densities $f_1(x,y)$ (left) and $f_2(x,y)$ (right).}
    \label{theoretical_dens_2D}
\end{figure}

Analogously to the univariate simulation study, for each bivariate density, we generated samples of sizes $N_1=3000$ and $N_2=5000$ and repeated each simulation 100 times. The densities were estimated using the penalized maximum likelihood method described in Section~\ref{PML2D}. We employed biquadratic splines, i.e., $k=l=2$, and first-order differences, i.e., $\mathrm{d}_x=\mathrm{d}_y=1$. Equidistant knot sequences with the same number of inner knots in both coordinate directions were considered, with $g=h=1,2,\ldots,10$. For each simulation run and each value of $g=h$, the optimal pair of penalization parameters $(\rho_x,\rho_y)$ was selected by $K$-fold cross-validation using $K_1=5$ and $K_2=10$ folds. 
Tables~\ref{opt_rho_f1_2D} and \ref{opt_rho_f2_2D} report the mean optimal values over the 100 simulation runs for the densities $f_1(x,y)$ and $f_2(x,y)$, respectively.

\begin{table}[h]
\centering
\small
\setlength{\tabcolsep}{5pt}
\begin{tabular}{cc|c|cccccccccc}
\toprule
 &  &  & \multicolumn{10}{c}{$g,\ h$} \\
\cmidrule(lr){4-13}
$N$ & $K$ & $\rho$ & 1 & 2 & 3 & 4 & 5 & 6 & 7 & 8 & 9 & 10 \\
\midrule
3000 & 5 & $\rho_x$ & 0.512 & 0.753 & 1.331 & 1.926 & 2.733 & 3.707 & 4.966 & 6.371 & 7.654 & 9.047 \\
3000 & 5 & $\rho_y$ & 0.503 & 0.665 & 1.274 & 1.856 & 2.477 & 3.410 & 4.466 & 5.748 & 7.235 & 8.410 \\
\bottomrule
3000 & 10 & $\rho_x$ & 0.576 & 0.816 & 1.446 & 1.983 & 2.883 & 4.033 & 5.208 & 6.694 & 8.000 & 9.329 \\
3000 & 10 & $\rho_y$ & 0.540 & 0.729 & 1.381 & 2.015 & 2.775 & 3.671 & 4.813 & 6.209 & 7.625 & 8.909 \\
\bottomrule
5000 & 5 & $\rho_x$ & 0.635 & 0.884 & 1.639 & 2.244 & 3.284 & 4.272 & 5.740 & 7.383 & 8.649 & 9.577 \\
5000 & 5 & $\rho_y$ & 0.612 & 0.950 & 1.842 & 2.476 & 3.246 & 4.447 & 5.758 & 7.077 & 8.857 & 9.771 \\
\bottomrule
5000 & 10 & $\rho_x$ & 0.669 & 0.928 & 1.774 & 2.338 & 3.386 & 4.525 & 5.806 & 7.853 & 8.939 & 9.730 \\
5000 & 10 & $\rho_y$ & 0.653 & 1.079 & 1.943 & 2.779 & 3.480 & 4.603 & 6.032 & 7.467 & 8.971 & 9.924 \\
\bottomrule
\end{tabular}
\caption{Mean values of optimal penalization parameters $\rho_x$ and $\rho_y$ for estimating the bivariate density $f_1(x,y)$ using different numbers of sampled values ($N_1=3000$, $N_2=5000$) and different number of folds $K$ used in \eqref{CV_2D} ($K_1=5$, $K_2=10$) (rows) for different number of inner knots $g=h$ (columns).}
\label{opt_rho_f1_2D}
\end{table}

For both bivariate densities, the mean optimal values of $\rho_x$ and $\rho_y$ generally increase with the number of inner knots $g=h$, see Tables~\ref{opt_rho_f1_2D} and \ref{opt_rho_f2_2D}. The same qualitative pattern is observed for both sample sizes and both numbers of cross-validation folds.

A clear difference is observed between the two simulation scenarios. For the unimodal density $f_1(x,y)$, the selected penalization parameters are larger than those obtained for the bimodal density $f_2(x,y)$. This is consistent with the greater local flexibility required to represent the two modes of $f_2(x,y)$, for which cross-validation selects weaker penalization. Thus, the selected penalization parameters depend not only on the dimension of the spline representation but also on the structure of the estimated density.

Since the true bivariate densities are known, the estimation accuracy was assessed using the normalized integrated squared error, analogously to the univariate case. For each simulation run $r$ and each fixed number of inner knots $g=h$, let $s^*_{kl,r,g}(x,y)$ denote the
estimated spline obtained by the penalized maximum likelihood procedure. The NISE was computed in $L_0^2(\Omega)$ as 
\[
\operatorname{NISE}_{r,g} = \frac{1}{(b-a)(d-c)}
\left\| \operatorname{clr}(f)(x,y)-s^*_{kl,r,g}(x,y)
\right\|_{L_0^2(\Omega)}^2, \qquad r=1,\ldots,100.
\]
The resulting NISE values for $f_1(x,y)$ and $f_2(x,y)$ are shown as boxplots in Figures~\ref{NISE_f1_2D} and \ref{NISE_f2_2D}, respectively.

\begin{figure}[h]
    \centering
    \includegraphics[width=0.4\textwidth]{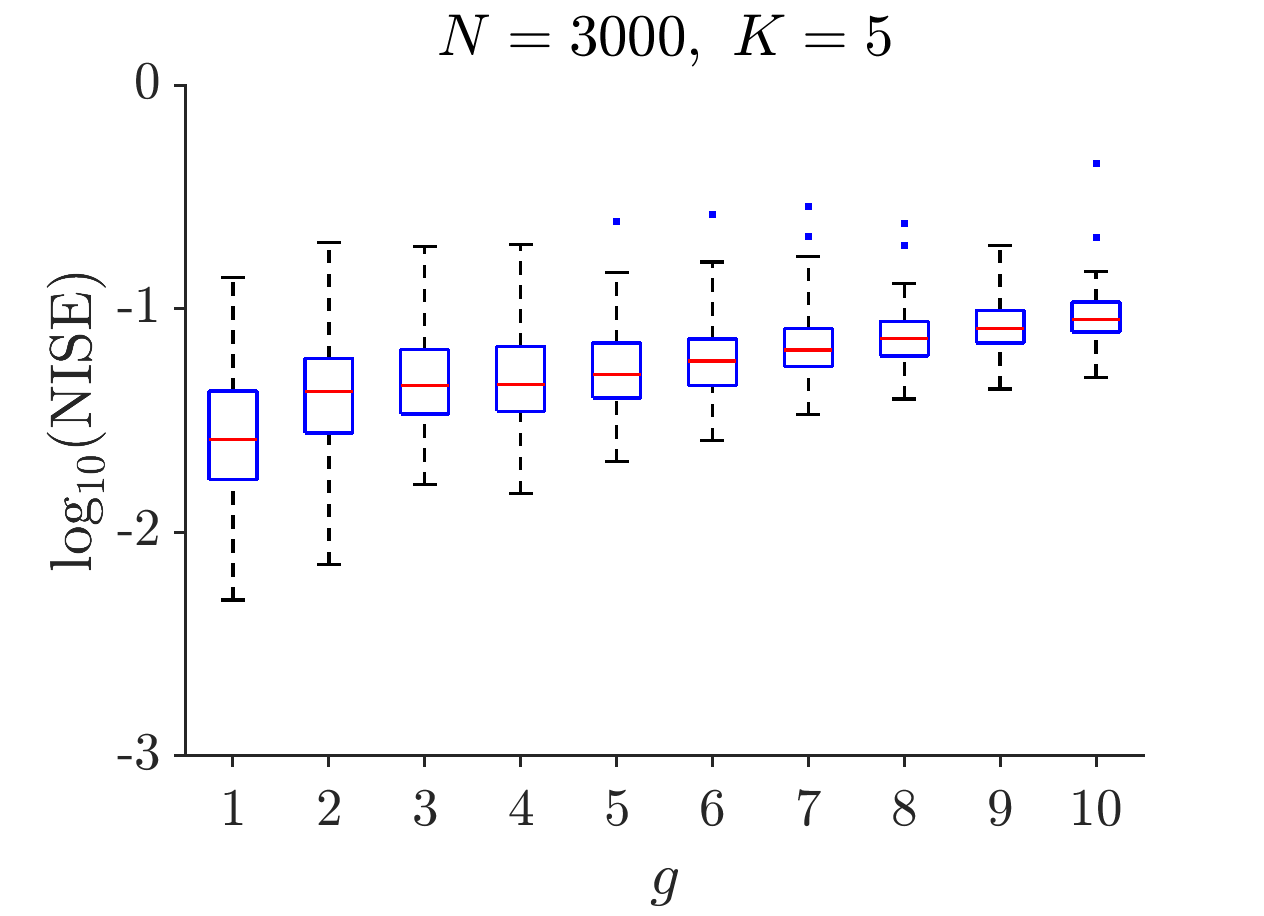}\qquad
    \includegraphics[width=0.4\textwidth]{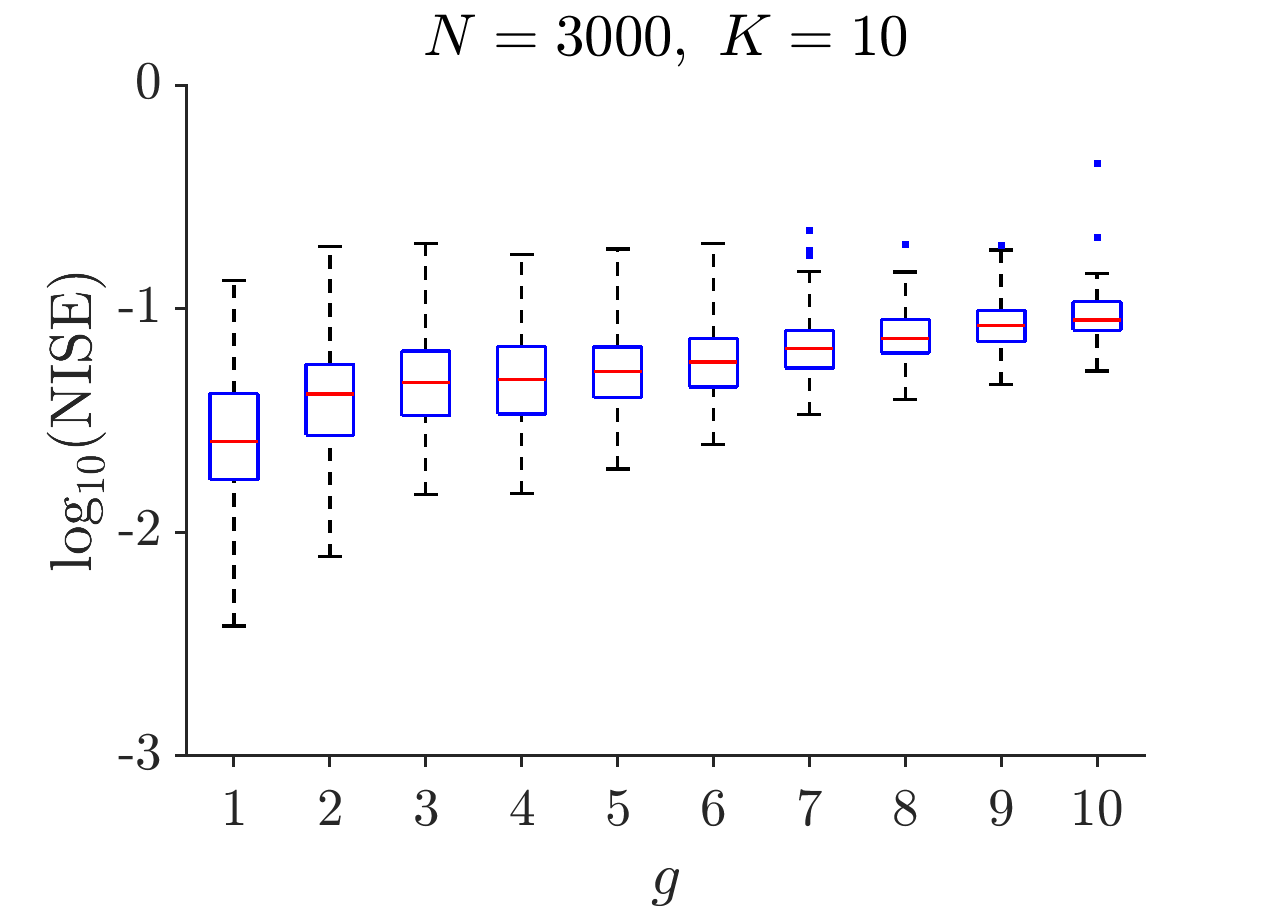}\\[0.5cm]
    \includegraphics[width=0.4\textwidth]{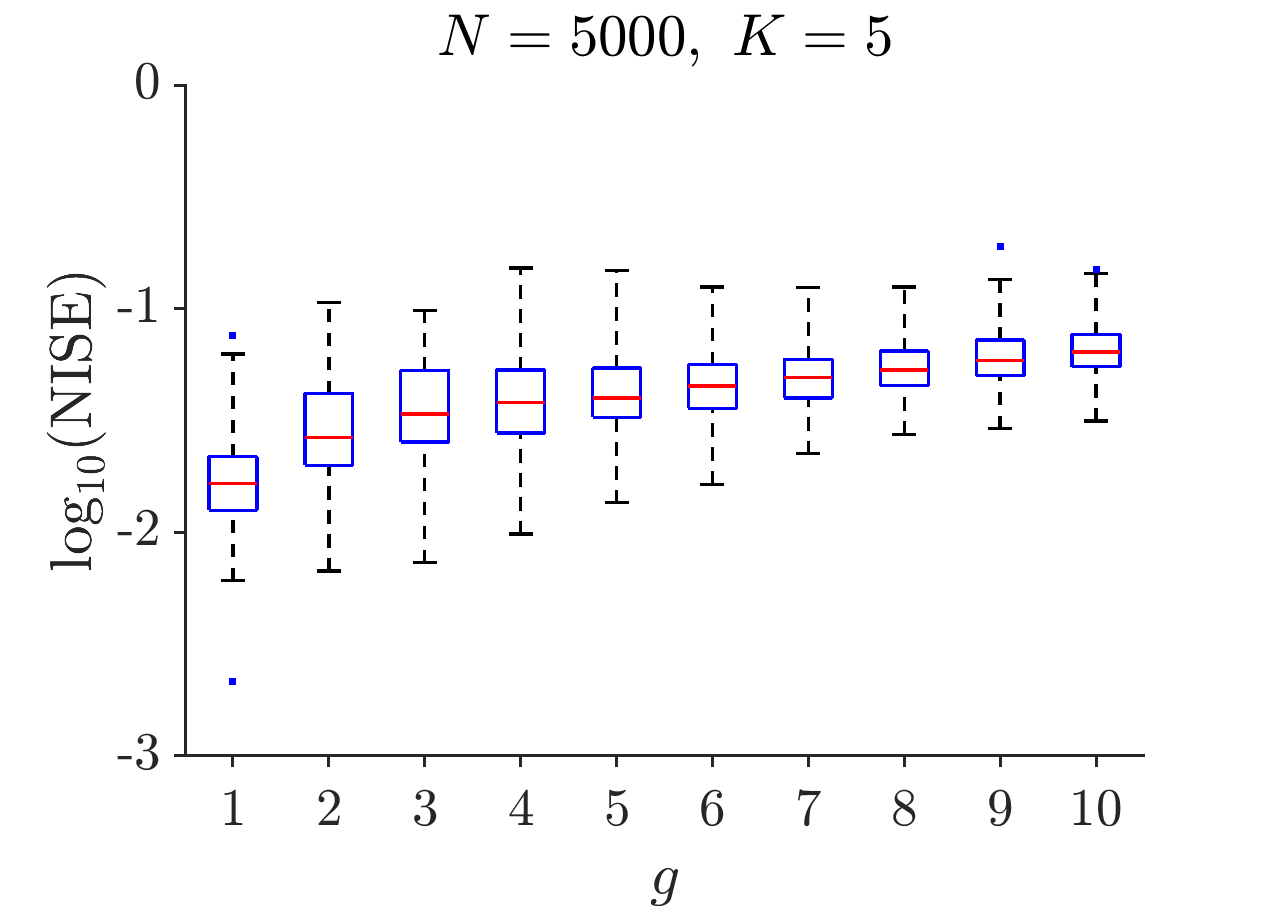}\qquad
    \includegraphics[width=0.4\textwidth]{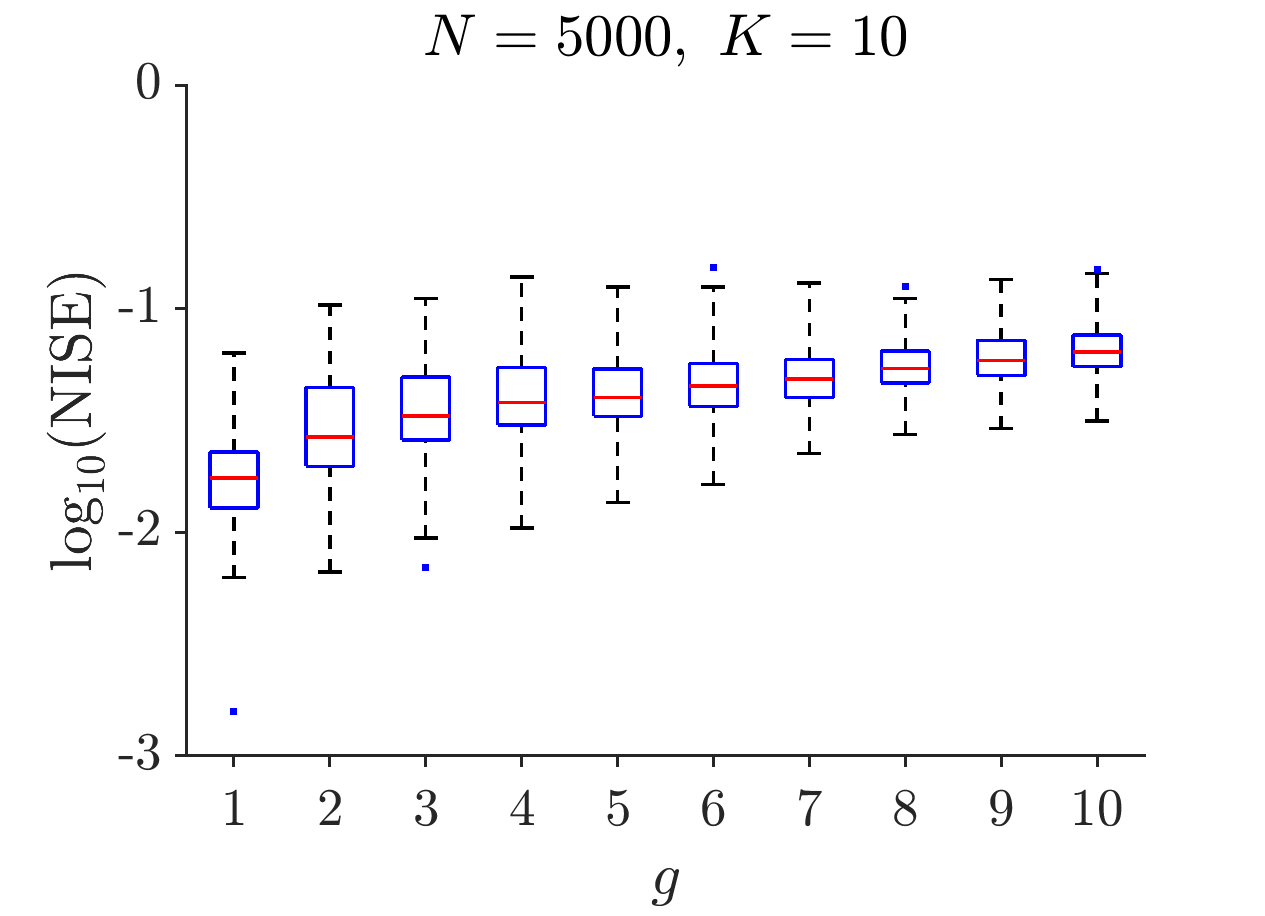}
    \caption{Boxplots of NISE values for the penalized likelihood estimates of $f_1(x,y)$ for different numbers of inner knots $g$ and the considered combinations of the sample size $N$ and the number of cross-validation folds $K$.}
    \label{NISE_f1_2D}
\end{figure}

In the unimodal case, the smallest median NISE for $f_1(x,y)$ is consistently obtained for $g=h=1$ across all considered combinations of $N$ and $K$, see Figure~\ref{NISE_f1_2D}. Increasing the number of inner knots does not improve estimation accuracy and generally leads to larger median NISE values, with occasional extreme values for some of the more flexible spline representations.

For the bimodal density $f_2(x,y)$, the smallest median NISE is obtained for $g=h=3$ in all considered settings, see Figure~\ref{NISE_f2_2D}. Thus, compared with the unimodal density, a moderately richer spline space is needed to represent the bimodal structure. Further increases in the number of inner knots do not improve the accuracy and are associated with increasing median NISE values. The qualitative patterns are similar for both sample sizes and both numbers of cross-validation folds.

The two bivariate densities were also estimated using Gaussian kernel density estimators with bandwidths selected by Silverman's rule of thumb and Scott's rule. Unlike the univariate simulation study, we omitted Sheather-Jones method for bandwidth selection, since it does not provide a unified bivariate counterpart. The resulting NISE values for $N=3000$ are shown as boxplots in Figure~\ref{NISE_kde_2D}.

\begin{figure}[h]
    \centering
    \includegraphics[width=0.3\textwidth]{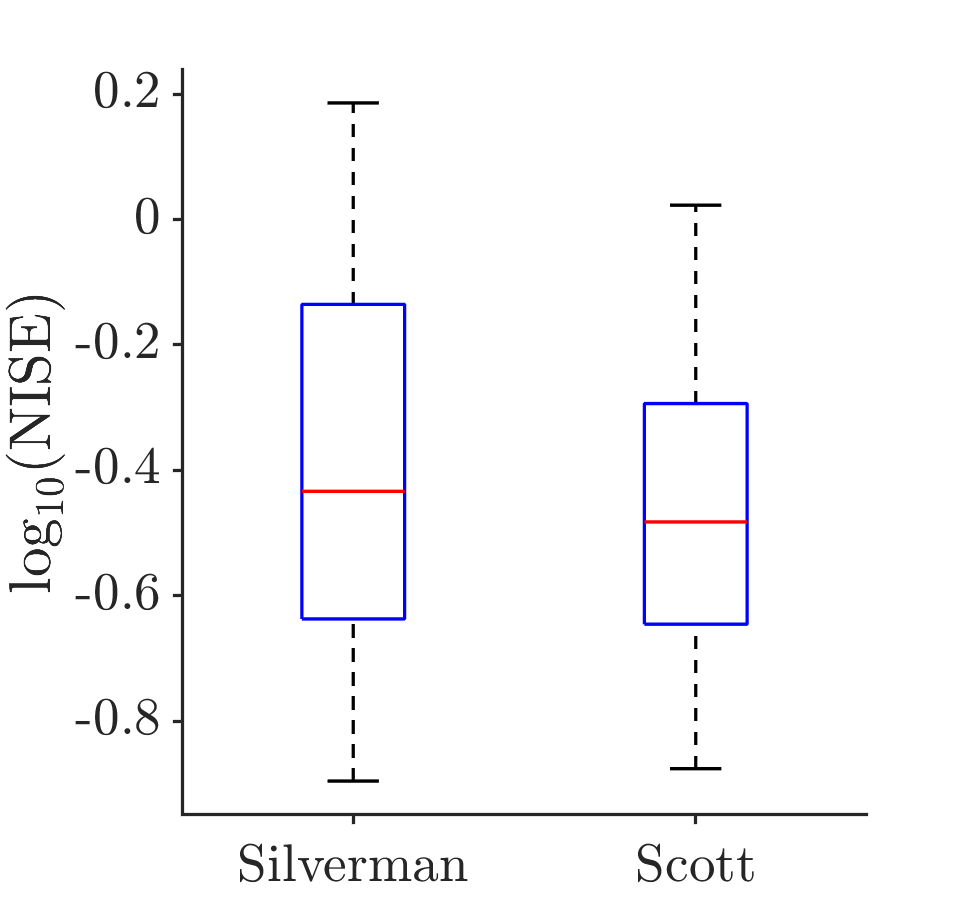}
    \includegraphics[width=0.3\textwidth]{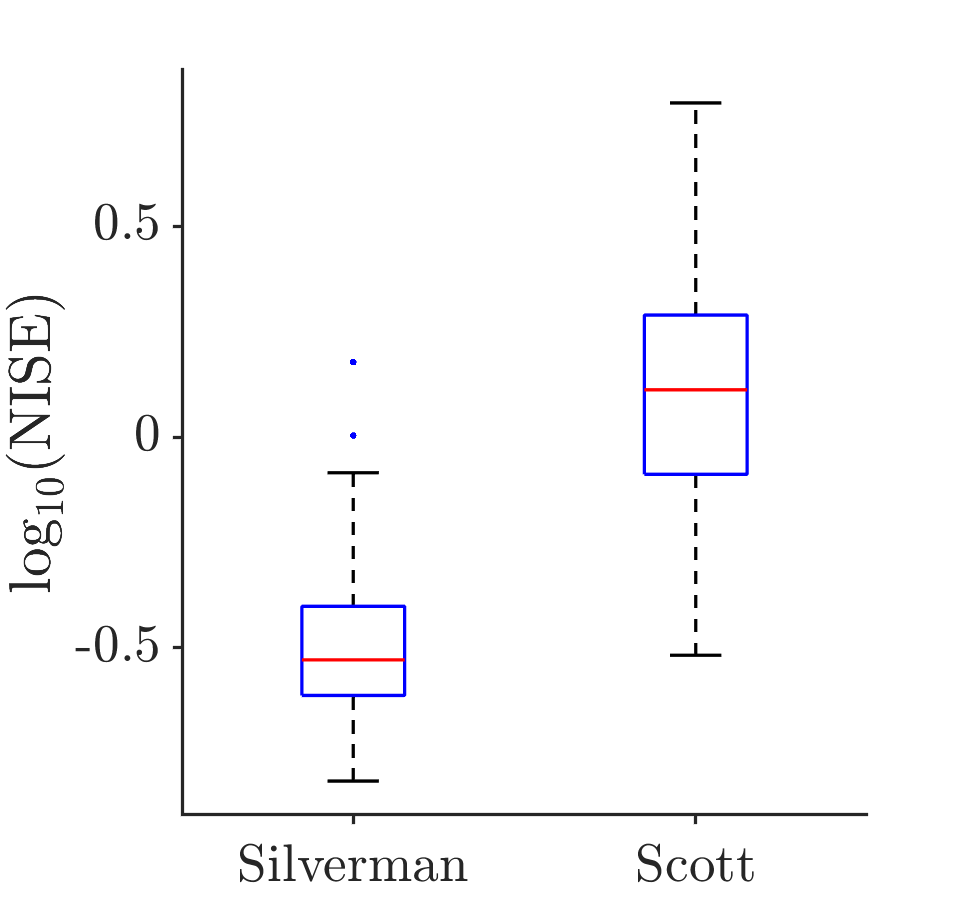}
    \caption{Boxplots of NISE values for Gaussian kernel density estimates of $f_1(x,y)$ and $f_2(x,y)$ (left and right, respectively) using Silverman's rule of thumb and Scott's rule, with $N=3000$.}
    \label{NISE_kde_2D}
\end{figure}

For the unimodal density $f_1(x,y)$, Scott's rule achieves slightly lower NISE values than Silverman's rule of thumb, whereas for the bimodal density $f_2(x,y)$ the lowest values are obtained using Silverman's rule. This illustrates the well-known sensitivity of kernel density estimation to the choice of bandwidth selection rule, particularly when the underlying density has a more complex structure.

The comparison with the proposed penalized likelihood estimator reveals different results for the two bivariate densities. For the unimodal density $f_1(x,y)$, the proposed estimator achieves substantially lower median NISE values than both kernel density estimators. For the bimodal density $f_2(x,y)$, however, Silverman's rule of thumb yields lower median NISE values than the best-performing spline estimates, whereas Scott's rule produces slightly higher values. Thus, the relative performance of the two approaches depends on both the structure of the underlying density and the bandwidth selection rule. The corresponding median NISE values are reported in Tables~\ref{NISE_2D_f1}, \ref{NISE_2D_f2}, and \ref{tab:KDE_NISE_2D}.

\begin{table}[h]
\centering
\small
\setlength{\tabcolsep}{5pt}
\renewcommand{\arraystretch}{1.15}

\begin{tabular}{cc|cccccccccc}
\toprule
 &  & \multicolumn{10}{c}{$g=h$} \\
\cmidrule(lr){3-12}
$N$ & $K$
&1&2&3&4&5&6&7&8&9&10\\
\midrule
3000&5 & 
\textbf{0.035}& 0.057& 0.061& 0.061& 0.068& 0.078& 0.087& 0.098& 0.109& 0.119 \\
3000&10 &
\textbf{0.034}& 0.055& 0.062& 0.064& 0.070& 0.077& 0.089& 0.098& 0.112& 0.119 \\
5000&5 &
\textbf{0.022}& 0.035& 0.045& 0.051& 0.053& 0.060& 0.066& 0.071& 0.078& 0.085 \\
5000&10 &
\textbf{0.023}& 0.036& 0.044& 0.051& 0.053& 0.060& 0.065& 0.072& 0.078& 0.085 \\
\bottomrule
\end{tabular}
\caption{
Median NISE values of $f_1(x,y)$ using different number of sampled values ($N_1=3000$, $N_2=5000$) and different number of folds used in \eqref{CV_2D} ($K_1=5$, $K_2=10$) (rows) for different numbers of inner knots $g=h$ (columns) with highlighted minimal NISE for each setting.}
\label{NISE_2D_f1}
\end{table}

\begin{table}[h]
\centering
\small
\setlength{\tabcolsep}{7pt}
\renewcommand{\arraystretch}{1.15}
\begin{tabular}{c|cc|cc}
\toprule
& \multicolumn{2}{c|}{$N=3000$}
& \multicolumn{2}{c}{$N=5000$} \\
\cmidrule(lr){2-3}
\cmidrule(lr){4-5}
Density & Silverman & Scott & Silverman & Scott \\
\midrule
$f_1(x,y)$ & 0.369 & \textbf{0.329} & 0.283
& \textbf{0.263} \\
$f_2(x,y)$ & \textbf{0.206} & 1.295 & \textbf{0.244} & 1.303\\
\bottomrule
\end{tabular}
\caption{Median NISE values for bivariate Gaussian kernel density estimates using Silverman's rule of thumb and Scott's rule for the considered sample sizes. The smaller median NISE for each density and sample size is highlighted in bold.}
\label{tab:KDE_NISE_2D}
\end{table}

Finally, examples of the resulting penalized maximum estimates of $f_1(x,y)$ and $f_2(x,y)$ are depicted in Figures \ref{final_comp_2D_f1} and \ref{final_comp_2D_f2}, respectively, together with the corresponding kernel density estimates.

\begin{figure}[h]
    \centering
    \includegraphics[width=0.32\textwidth]{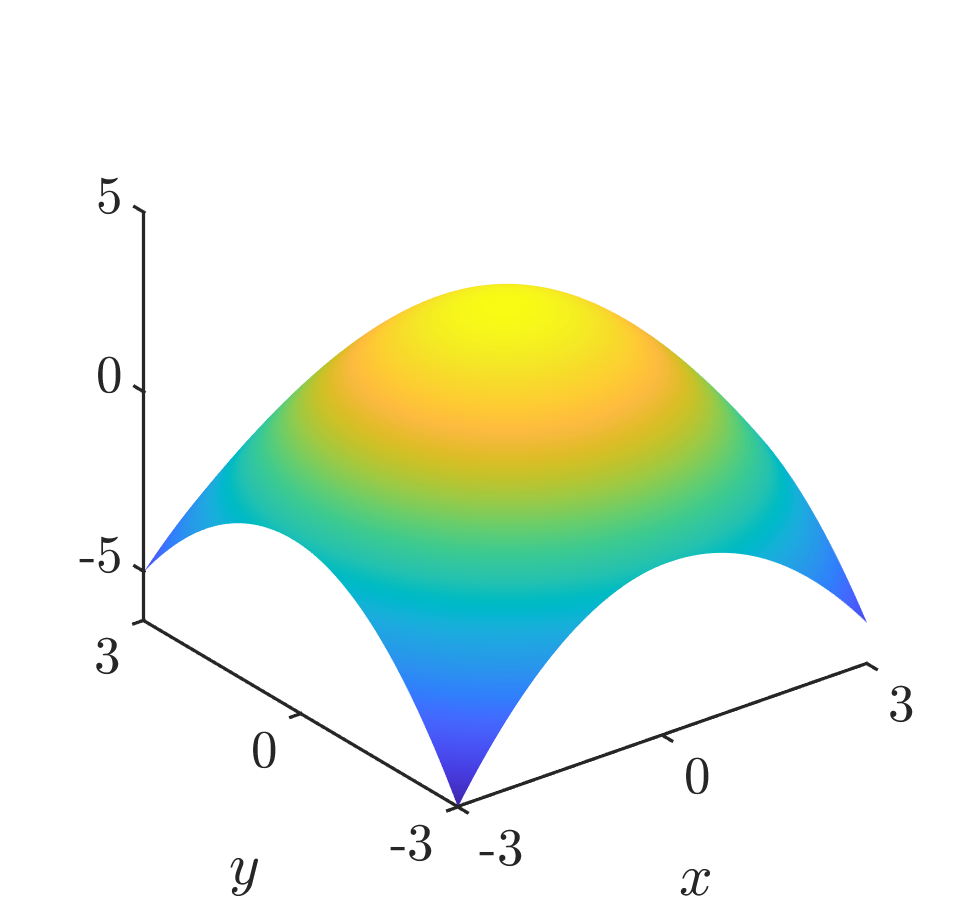}
    \includegraphics[width=0.32\textwidth]{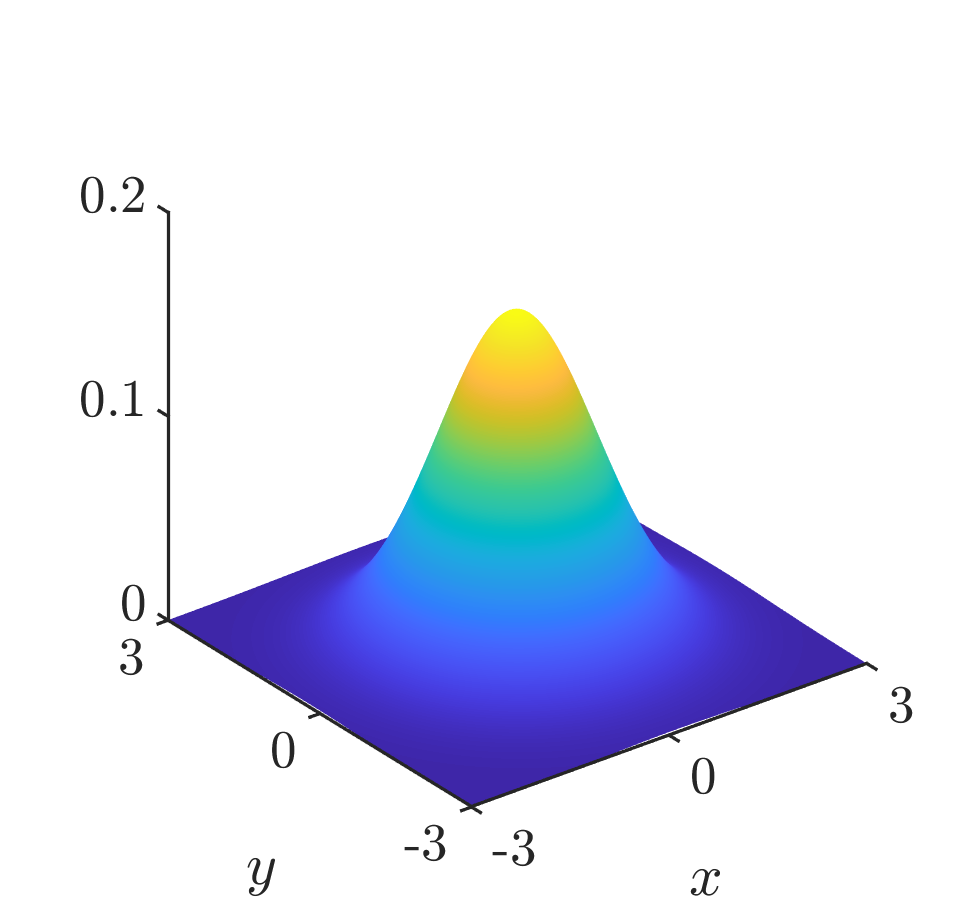}
    \includegraphics[width=0.32\textwidth]{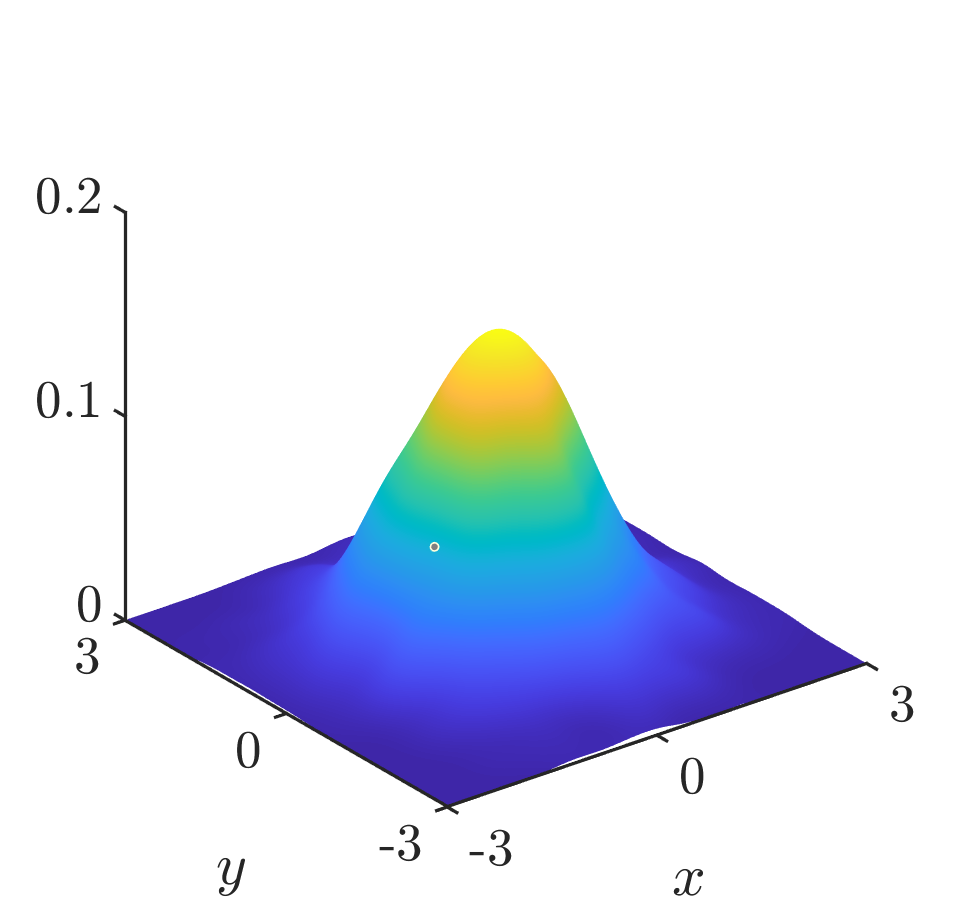}
    \caption{Penalized likelihood estimation of the density $f_1(x,y)$ constructed with $N=3000$ using $g=h=1$ and $\rho_x = 0.512$, $\rho_y=0.503$ in $L^2_0(\Omega)$ (left) and in $\mathcal{B}^2(\Omega)$ (middle), together with the kernel density estimate with $N=3000$ and bandwidth obtained with the Scott's rule in $\mathcal{B}^2(\Omega)$ (right).}
    \label{final_comp_2D_f1}
\end{figure}

Overall, the bivariate simulation study demonstrates that the proposed penalized maximum likelihood approach is capable of recovering both simple and more complex multimodal density functions. The results also emphasize the importance of balancing the flexibility of the spline representation with an appropriate level of penalization. The comparison with kernel density estimation further indicates neither approach uniformly dominates across different density structures. In particular, the proposed method performs particularly well for the unimodal density, whereas the relative performance in the more complex bimodal setting depends on the bandwidth selection rule. These findings support the ability of the proposed likelihood framework to provide flexible bivariate density estimates while retaining direct control over their smoothness via penalization.

\section{Application}\label{Application}
In this section, we illustrate the practical applicability of the proposed methodology using empirical geochemical data previously analyzed in \cite{Skorna2026, Grygar2024, Grygar2026}. The data were obtained from the Register of Contaminated Areas (Registr kontaminovaných ploch), collected by the Department of Agriculture of the Czech Republic \cite{eagri}. The dataset contains measurements of soil concentrations of several chemical elements in 77 Czech administrative districts, including the capital city of Prague, measured in $mg \cdot kg^{-1}$. We consider both a univariate application based on lead (Pb) concentrations and a bivariate application based on the joint distribution of lead (Pb) and zinc (Zn) concentrations.

\subsection{Lead distribution}
In the univariate application, we consider the distribution of log-transformed lead (Pb) concentrations on the common domain $I=[1.358,\ 5.716]$. For each district, the density was estimated using the penalized maximum likelihood approach described in Section \ref{PML1D}. We employed cubic splines, i.e., $k=3$, with first and second-order differences, i.e., $\mathrm{d}\in\{1,2\}$. Equidistant knot sequences with $g=1,2,\dots,10$ inner knots were used.
For each district and each combination of $\mathrm{d}$ and $g$, the penalization parameter $\rho$ was determined by $K$-fold cross-validation, as described in Section \ref{rhoselection}, using $K_1=5$ and $K_2=10$. For each combination of $\mathrm{d}$, $g$, and $K$, we recorded the cross-validated negative log-likelihood and the corresponding optimal penalization parameter for each district. Tables~\ref{application_cv_d=1} and \ref{application_cv_d=2} summarize these quantities across all $77$ districts by their mean and median values for $\mathrm{d}=1$ and $\mathrm{d}=2$, respectively.

\begin{table}[h]
\centering
\small
\setlength{\tabcolsep}{5pt}
\renewcommand{\arraystretch}{1.15}
\begin{tabular}{c|cccc|cccc}
\toprule
& \multicolumn{4}{c|}{$K=5$} 
& \multicolumn{4}{c}{$K=10$} \\
\cmidrule(lr){2-5}
\cmidrule(lr){6-9}
$g$
& Mean CV & Med CV & Mean $\rho$ & Med $\rho$
& Mean CV & Med CV & Mean $\rho$ & Med $\rho$ \\
\midrule
1  & 39.784 & 35.574 & 0.033 & 0.014 & 19.897 & 17.801 & 0.04  & 0.014 \\
2  & 39.401 & 35.004 & 0.037 & 0.014 & 19.68  & 17.35  & 0.045 & 0.018 \\
3  & 38.94  & 34.008 & 0.077 & 0.045 & 19.474 & 16.716 & 0.093 & 0.045 \\
4  & 38.881 & 34.324 & 0.114 & 0.045 & 19.457 & 16.851 & 0.135 & 0.091 \\
5  & 38.785 & 34.325 & 0.198 & 0.091 & 19.38  & 16.875 & 0.198 & 0.091 \\
6  & 38.661 & 33.843 & 0.273 & 0.184 & 19.314 & \textbf{16.669} & 0.267 & 0.184 \\
7  & 38.61  & \textbf{33.821} & 0.366 & 0.233 & 19.265 & 16.912 & 0.335 & 0.184 \\
8  & 38.688 & 33.867 & 0.429 & 0.233 & 19.308 & 16.882 & 0.415 & 0.295 \\
9  & \textbf{38.574} & 34.027 & 0.52  & 0.295 & 19.244 & 16.819 & 0.5   & 0.295 \\
10 & 38.871 & 34.12  & 0.592 & 0.373 &\textbf{ 19.23}  & 16.872 & 0.554 & 0.295 \\
\bottomrule
\end{tabular}
\caption{
Cross-validated negative log-likelihood values and corresponding optimal penalization parameters for the proposed penalized likelihood estimator for $\mathrm{d}=1$ with different numbers of inner knots $g$ and different numbers of folds $K$ with highlighted minimal mean and median CV values.}
\label{application_cv_d=1}
\end{table}

For small numbers of inner knots, the mean and median cross-validation criteria are generally higher. As $g$ increases, the criteria decrease, although the pattern differs somewhat between first- and second-order differences. For $\mathrm{d}=1$, see Table \ref{application_cv_d=1}, the smallest summary cross-validation criteria are attained for $g$ between 6 and 10, depending on the number of folds and whether the mean or median criterion is considered. For $\mathrm{d}=2$, as reported in Table \ref{application_cv_d=2}, the minima are more concentrated at the largest spline spaces, namely $g=9$ or $g=10$. At the same time, the optimal penalization parameter generally increases with $g$, particularly beyond the smallest knot configurations, indicating that increasing spline flexibility is accompanied by stronger penalization. Although the absolute cross-validation values differ between $K=5$ and $K=10$, the overall patterns with respect to $g$ are similar for both choices of $K$.

For comparison, we also considered Gaussian kernel density estimation with bandwidths selected using Silverman's rule of thumb, Scott’s rule, and the  Sheather-Jones method, as in Section \ref{Simulation1D}. Since the true reference densities are unknown, the kernel density estimators were evaluated using the same cross-validation partitions as the proposed  penalized likelihood estimator, ensuring a consistent comparison between the two approaches. Tables \ref{application_cv_kde_d=1} and \ref{application_cv_kde_d=2} summarize the resulting cross-validated negative log-likelihood values across all 77 districts by their mean and median values.

\begin{table}[h]
\centering
\small
\setlength{\tabcolsep}{7pt}
\renewcommand{\arraystretch}{1.15}
\begin{tabular}{l|cc|cc}
\toprule
& \multicolumn{2}{c|}{$K=5$} 
& \multicolumn{2}{c}{$K=10$} \\
\cmidrule(lr){2-3}
\cmidrule(lr){4-5}
Method 
& Mean CV & Median CV
& Mean CV & Median CV \\
\midrule
Silverman & 55.578 & 43.149 & 25.717 & 19.712 \\
Scott     & \textbf{48.853} & \textbf{39.243} & 24.029 & \textbf{19.261} \\
Sheather-Jones & 50.332 & 40.422 & \textbf{23.619} & 19.387 \\
\bottomrule
\end{tabular}
\caption{
Cross-validated negative log-likelihood values of kernel density estimators with different bandwidth selection rules, with highlighted minimal mean and median CV values.}
\label{application_cv_kde_d=1}
\end{table}

Among the three bandwidth selection rules, Silverman’s rule of thumb yields the highest cross-validation criteria in all settings. For $K=5$, Scott’s rule gives the smallest mean and median values, whereas for $K=10$, the smallest mean value is obtained with the Sheather–Jones method and the smallest median value with Scott’s rule. More importantly, for both $\mathrm{d}=1$ and $\mathrm{d}=2$, the minimum mean and median cross-validation criteria attained by the proposed spline-based estimator across the considered values of $g$ are lower than the corresponding values obtained by all three kernel density estimators. Thus, for the considered data and cross-validation settings, the proposed penalized likelihood estimator attains lower cross-validated negative log-likelihood values than the examined kernel density estimators.

An example of the resulting density estimate is shown in Figure~\ref{chomutov} for the Chomutov district, together with the corresponding kernel density estimate in $\mathcal{B}^2(I)$ and $L_0^2(I)$.

\begin{figure}[h]
    \centering
    \includegraphics[width=0.3\textwidth]{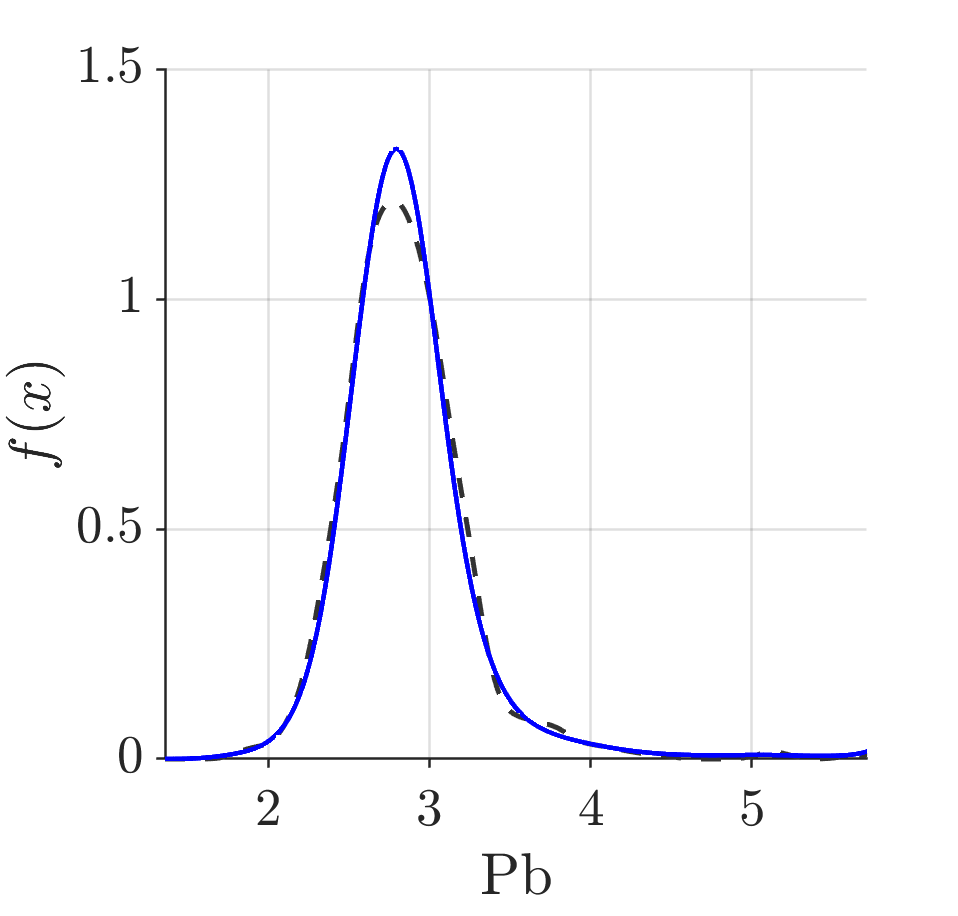}
    \includegraphics[width=0.3\textwidth]{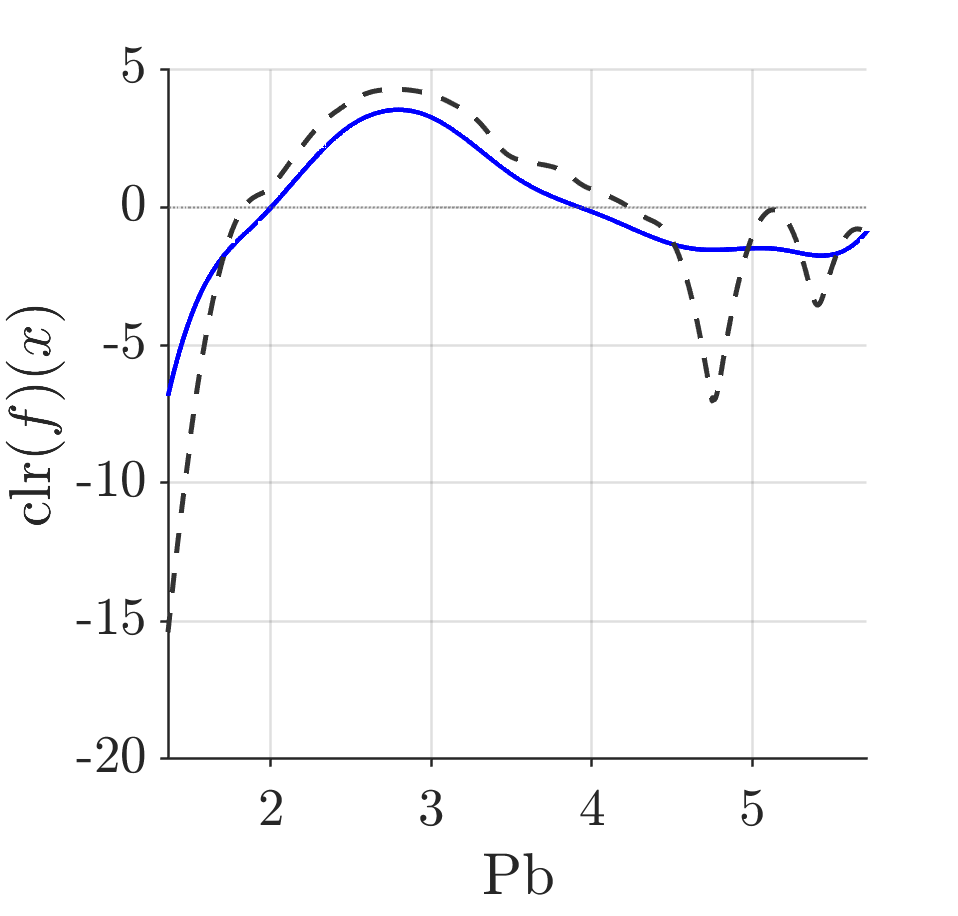}
    \caption{Penalized likelihood estimate of the density function for the Chomutov district with $\mathrm{d}=1$, $g=7$, $k=3$, and $\rho=0.366$ (blue), together with the kernel density estimate (dashed) with bandwidth obtained with the Scott's rule, in $\mathcal{B}^2(I)$ (left) and $L_0^2(I)$ (right).}
    \label{chomutov}
\end{figure}

\subsection{Lead and Zinc distribution}
In the bivariate application, we focus on the joint distribution of log-transformed lead (Pb) and zinc (Zn) concentrations. The common domain was set to $\Omega=[1.358,5.716]\times[1.459,5.663]$. For each district, the density was estimated using the penalized maximum likelihood approach described in Section~\ref{PML2D}. We employed bicubic splines, i.e., $k=l=3$, with equal difference orders in both coordinate directions, $\mathrm{d}_x=\mathrm{d}_y\in\{1,2\}$. Equidistant knot sequences with the same number of inner knots in both directions were used, with $g=h=1,2,\ldots,10$. For each district and each combination of $\mathrm{d}_x=\mathrm{d}_y$ and $g=h$, the pair of penalization parameters $(\rho_x,\rho_y)$ was determined by $K$-fold cross-validation, as described in Section~\ref{rhoselection}, using $K_1=5$ and $K_2=10$ folds.

For each combination of $\mathrm{d}_x=\mathrm{d}_y$, $g=h$, and $K$, we recorded the cross-validated negative log-likelihood and the corresponding optimal pair of penalization parameters $(\rho_x,\rho_y)$ for each district. Tables~\ref{application_cv_2D_d1_K5} and \ref{application_cv_2D_d1_K10} summarize these quantities across all $77$ districts by their mean and median values for $\mathrm{d}_x=\mathrm{d}_y=1$ using $K=5$ and $K=10$, respectively. The corresponding results for $\mathrm{d}_x=\mathrm{d}_y=2$ are reported in Tables~\ref{application_cv_2D_d2_K5} and \ref{application_cv_2D_d2_K10}.

\begin{table}[h]
\centering
\small
\setlength{\tabcolsep}{4pt}
\renewcommand{\arraystretch}{1.15}
\begin{tabular}{c|cccccc}
\toprule
& \multicolumn{6}{c}{$\mathrm{d}_x=\mathrm{d}_y=1,\ K=5$} \\
\cmidrule(lr){2-7}
$g=h$
& Mean CV & Med CV
& Mean $\rho_x$ & Med $\rho_x$
& Mean $\rho_y$ & Med $\rho_y$ \\
\midrule
1 & 80.367 & 71.860 & 0.009 & 0.002 & 0.013 & 0.002 \\
2 & 79.224 & 72.289 & 0.011 & 0.002 & 0.021 & 0.007 \\
3 & 78.767 & 72.940 & 0.031 & 0.013 & 0.057 & 0.043 \\
4 & 78.272 & 72.347 & 0.071 & 0.023 & 0.120 & 0.078 \\
5 & 78.051 & 72.675 & 0.134 & 0.078 & 0.215 & 0.144 \\
6 & 77.992 & 71.251 & 0.224 & 0.144 & 0.370 & 0.264 \\
7 & \textbf{77.704} & \textbf{70.964} & 0.400 & 0.264 & 0.571 & 0.483 \\
8 & 77.992 & 72.682 & 0.520 & 0.264 & 0.863 & 0.483 \\
9 & 78.088 & 72.818 & 0.685 & 0.264 & 1.177 & 0.886 \\
10 & 78.307 & 73.362 & 0.940 & 0.886 & 1.431 & 0.886 \\
\bottomrule
\end{tabular}
\caption{
Cross-validated negative log-likelihood values together with the corresponding optimal penalization parameters $\rho_x$ and $\rho_y$ of the proposed penalized likelihood estimator for different numbers of inner knots $g=h$ using 5-fold cross-validation with $\mathrm{d}_x=\mathrm{d}_y=1$. The smallest mean and median cross-validation values are highlighted in bold.}
\label{application_cv_2D_d1_K5}
\end{table}

\begin{table}[h]
\centering
\small
\setlength{\tabcolsep}{4pt}
\renewcommand{\arraystretch}{1.15}
\begin{tabular}{c|cccccc}
\toprule
& \multicolumn{6}{c}{$\mathrm{d}_x=\mathrm{d}_y=1,\ K=10$} \\
\cmidrule(lr){2-7}
$g=h$
& Mean CV & Med CV
& Mean $\rho_x$ & Med $\rho_x$
& Mean $\rho_y$ & Med $\rho_y$ \\
\midrule
1 & 40.095 & 36.444 & 0.008 & 0.002 & 0.013 & 0.004 \\
2 & 39.507 & 36.629 & 0.011 & 0.002 & 0.017 & 0.007 \\
3 & 39.312 & 36.854 & 0.025 & 0.013 & 0.064 & 0.043 \\
4 & 38.998 & 35.362 & 0.061 & 0.013 & 0.129 & 0.078 \\
5 & 38.917 & 36.042 & 0.150 & 0.078 & 0.212 & 0.144 \\
6 & 38.845 & 34.930 & 0.240 & 0.144 & 0.334 & 0.264 \\
7 & \textbf{38.733} & \textbf{34.889} & 0.395 & 0.264 & 0.571 & 0.483 \\
8 & 38.849 & 35.436 & 0.541 & 0.264 & 0.827 & 0.483 \\
9 & 38.886 & 35.483 & 0.666 & 0.483 & 1.161 & 0.886 \\
10 & 38.949 & 35.609 & 0.879 & 0.483 & 1.495 & 0.886 \\
\bottomrule
\end{tabular}
\caption{
Cross-validated negative log-likelihood values together with the corresponding optimal penalization parameters $\rho_x$ and $\rho_y$ of the proposed penalized likelihood estimator for different numbers of inner knots $g=h$ using 10-fold cross-validation with $\mathrm{d}_x=\mathrm{d}_y=1$. The smallest mean and median cross-validation values are highlighted in bold.}
\label{application_cv_2D_d1_K10}
\end{table}

For $\mathrm{d}_x=\mathrm{d}_y=1$, both the mean and median cross-validation criteria attain their minimum at $g=h=7$ for both $K=5$ and $K=10$. For $\mathrm{d}_x=\mathrm{d}_y=2$, a substantial decrease in the cross-validation criteria is observed when moving from $g=h=1$ to $g=h=2$, while the subsequent changes are considerably smaller. The mean criterion attains its minimum at $g=h=10$ for both choices of $K$, whereas the median criterion reaches its minimum at $g=h=2$ for $K=5$ and at $g=h=9$ for $K=10$.

The optimal penalization parameters $\rho_x$ and $\rho_y$ generally increase with the number of inner knots. This pattern is particularly clear for $\mathrm{d}_x=\mathrm{d}_y=1$. For $\mathrm{d}_x=\mathrm{d}_y=2$, the values corresponding to $g=h=1$ deviate from this pattern, whereas from $g=h=2$ onward, both penalization parameters generally increase with increasing spline complexity.

A comparison of the two difference orders reveals different patterns for the two cross-validation schemes. For $K=5$, except for $g=h=1$, both the mean and median cross-validation criteria are lower for $\mathrm{d}_x=\mathrm{d}_y=2$ than for $\mathrm{d}_x=\mathrm{d}_y=1$ at the same number of inner knots. For $K=10$, the differences between the two difference orders are considerably smaller and depend on the number of inner knots and on whether the mean or median criterion is considered. In particular, for larger values of $g=h$, the mean criterion tends to be slightly lower for $\mathrm{d}_x=\mathrm{d}_y=2$, whereas no uniform pattern is observed for the median criterion. Thus, the effect of the difference order depends on both the spline complexity and the cross-validation setting.

We also evaluated Gaussian kernel density estimation with bandwidths selected using Silverman's rule of thumb and Scott's rule. Sheather-Jones bandwidth selection method was omitted as mentioned in Section~\ref{subsec:bivariate_sim}. The kernel density estimators were evaluated using the same cross-validation partitions as the proposed penalized likelihood estimator, ensuring a consistent comparison between the two approaches. Table~\ref{application_cv_kde_2D} summarizes the resulting cross-validated negative log-likelihood values across all $77$ districts by their mean and median values. For both $K=5$ and $K=10$, Scott's rule yields lower mean and median cross-validation criteria than Silverman's rule of thumb.

\begin{table}[h]
\centering
\small
\setlength{\tabcolsep}{7pt}
\renewcommand{\arraystretch}{1.15}
\begin{tabular}{l|cc|cc}
\toprule
& \multicolumn{2}{c|}{$K=5$} 
& \multicolumn{2}{c}{$K=10$} \\
\cmidrule(lr){2-3}
\cmidrule(lr){4-5}
Method 
& Mean CV & Median CV
& Mean CV & Median CV \\
\midrule
Silverman & 88.680 & 81.981 & 43.164 & 40.312 \\
Scott & \textbf{85.144} & \textbf{76.821} & \textbf{41.984} & \textbf{38.260} \\
\bottomrule
\end{tabular}
\caption{
Cross-validated negative log-likelihood values of kernel density estimators with different bandwidth selection rules with highlighted minimal mean and median CV values.}
\label{application_cv_kde_2D}
\end{table}

Comparing Tables~\ref{application_cv_2D_d1_K5}, \ref{application_cv_2D_d1_K10}, \ref{application_cv_2D_d2_K5} and \ref{application_cv_2D_d2_K10} with Table \ref{application_cv_kde_2D} shows that the proposed penalized likelihood estimator attains lower cross-validation criteria than both kernel density estimators for both difference orders. For $\mathrm{d}_x=\mathrm{d}_y=1$, the lowest mean and median criteria are $77.704$ and $70.964$, respectively, for $K=5$, and $38.733$ and $34.889$ for $K=10$; all these values are attained at $g=h=7$. For $\mathrm{d}_x=\mathrm{d}_y=2$, the lowest mean and median criteria are $73.943$ and $66.902$, respectively, for $K=5$, attained at $g=h=10$ and $g=h=2$, and $38.591$ and $34.941$ for $K=10$, attained at $g=h=10$ and $g=h=9$, respectively. In comparison, Scott's rule, which gives the lower kernel-based criteria, yields mean and median values of $85.144$ and $76.821$ for $K=5$, and $41.984$ and $38.260$ for $K=10$.

The results obtained for the considered geochemical data further support the practical applicability of the proposed penalized likelihood estimator in the bivariate setting. Moreover, the resulting densities are represented directly through their spline coefficients, which provide a natural basis for subsequent statistical analyses of density-valued data, such as functional principal component analysis, clustering, or regression; see, e.g.,~\cite{Skorna2025}.

An example of the resulting maximum likelihood estimate is shown in Figure \ref{chomutov_2D} for Chomutov district in $L^2_0(\Omega)$ and $\mathcal{B}^2(\Omega)$, together with the corresponding kernel density estimate in $\mathcal{B}^2(\Omega)$.

\begin{figure}[h]
    \centering
    \includegraphics[width=0.31\textwidth]{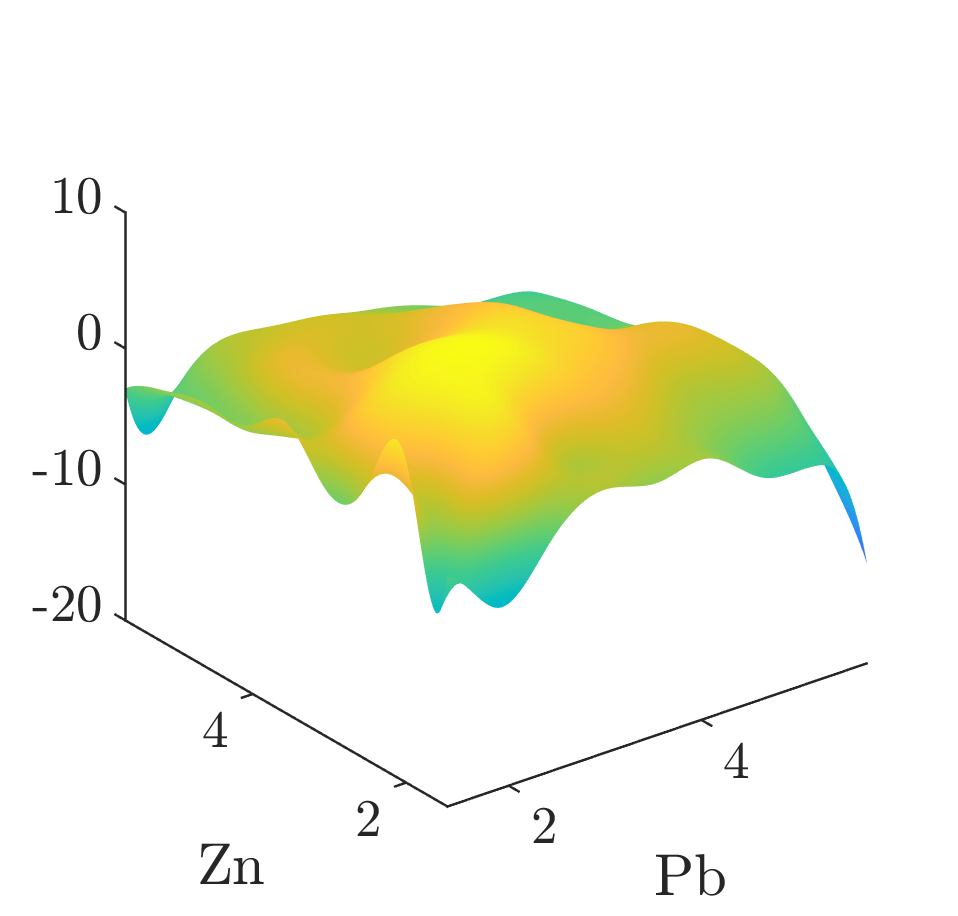}
    \includegraphics[width=0.32\textwidth]{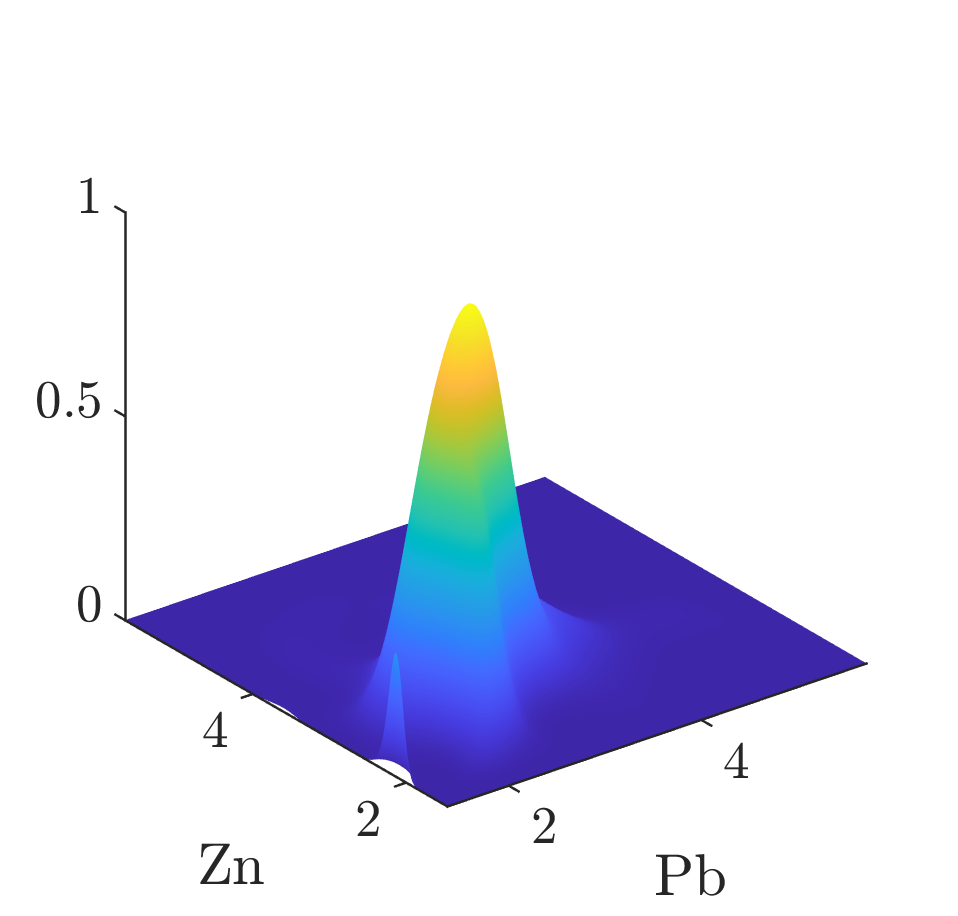}
    \includegraphics[width=0.32\textwidth]{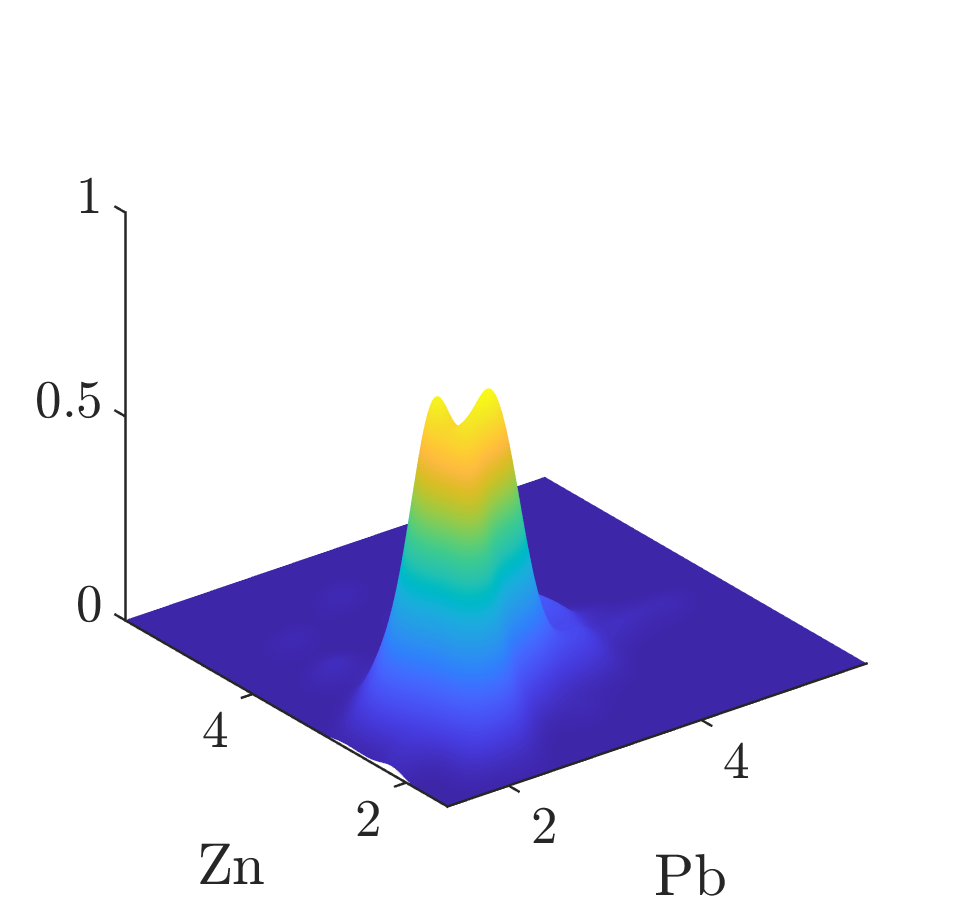}
    \caption{Penalized likelihood estimate of the density function for the Chomutov district with $\mathrm{d}_x=\mathrm{d}_y=1$, $g=h=7$, $k=l=3$, and $\rho_x=0.264$, $\rho_y=0.483$ in $L^2_0(\Omega)$ (left) and $\mathcal{B}^2(\Omega)$ (middle), together with the kernel density estimate with bandwidth obtained with the Scott's rule in $\mathcal{B}^2(\Omega)$ (right).}
    \label{chomutov_2D}
\end{figure}

Since the kernel density estimate consists of zero values on the edges of the common domain $\Omega$, the clr transformation cannot be employed to obtain the corresponding counterpart in $L^2_0(\Omega)$ without additional preprocessing steps, see \cite{Skorna2026}.

As the proposed penalized likelihood approach preserves the orthogonal decomposition, as described in Sections~\ref{Bayes} and \ref{bivariateZB}, the methodology directly provides estimates of independent and interactive parts, together with the estimate of the corresponding joint density. The decomposition parts are depicted in Figure~\ref{chomutov_decomposition}.

\begin{figure}[h]
    \centering
    \includegraphics[width=0.32\textwidth]{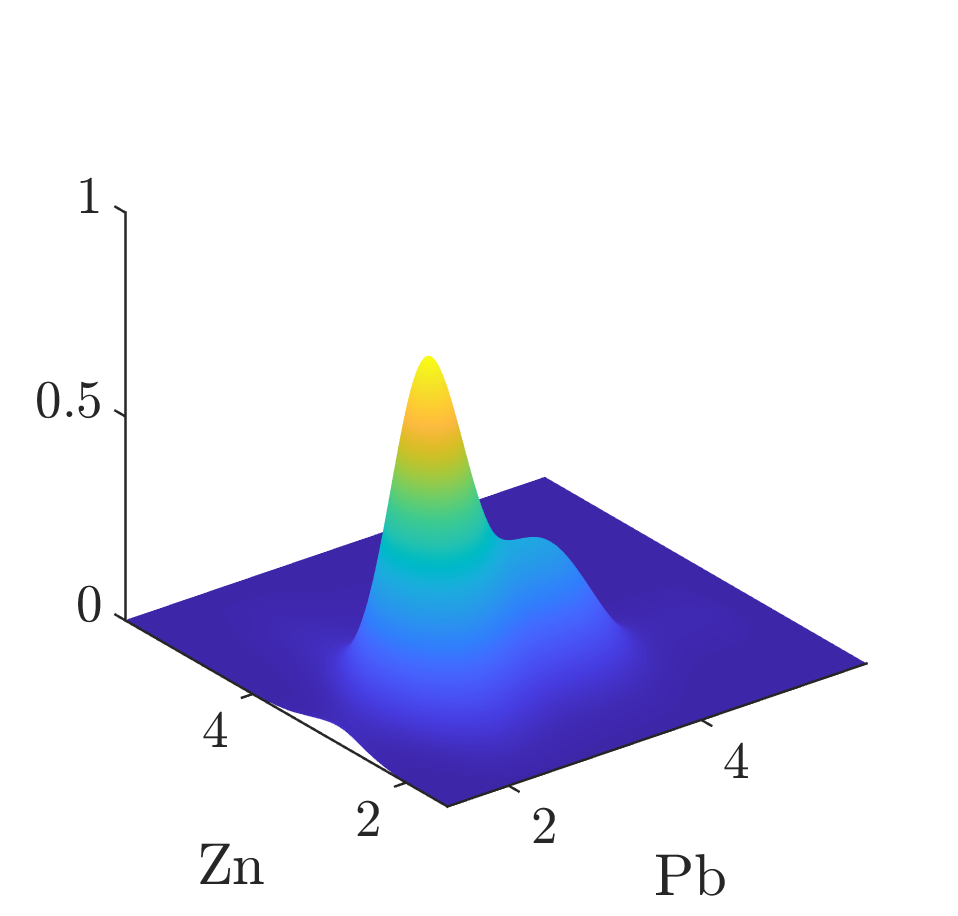}
    \includegraphics[width=0.32\textwidth]{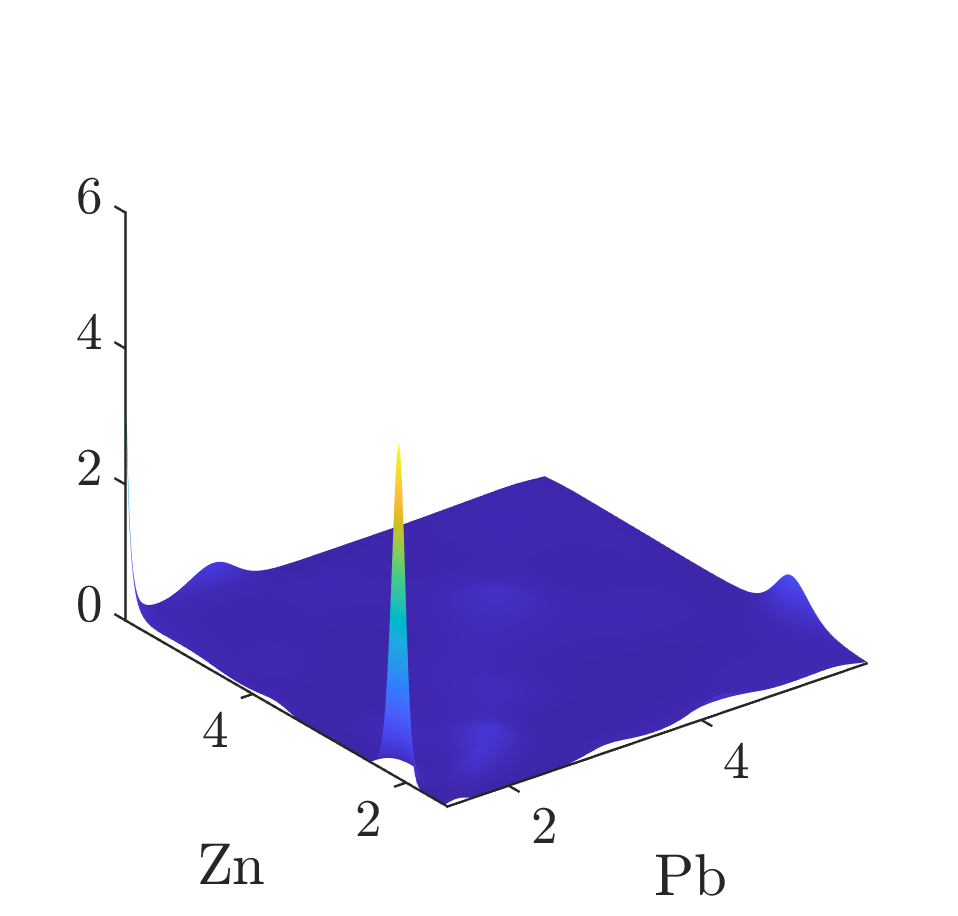}\\
    \includegraphics[width=0.32\textwidth]{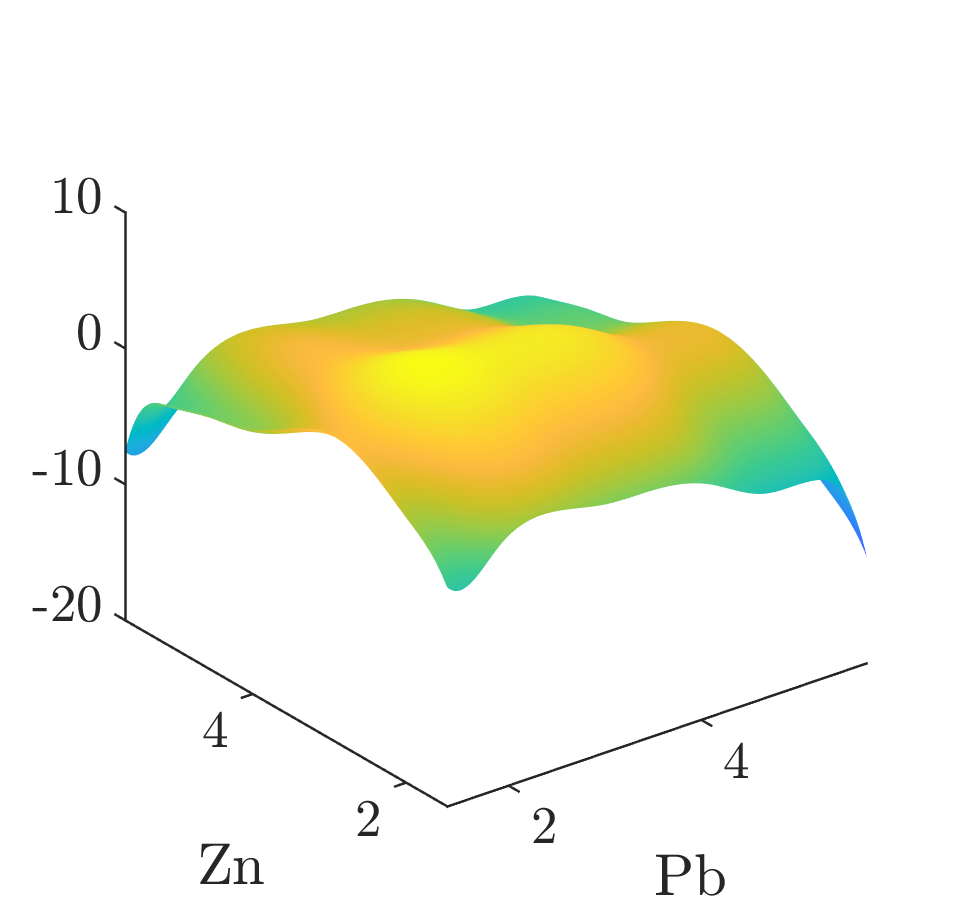}
    \includegraphics[width=0.32\textwidth]{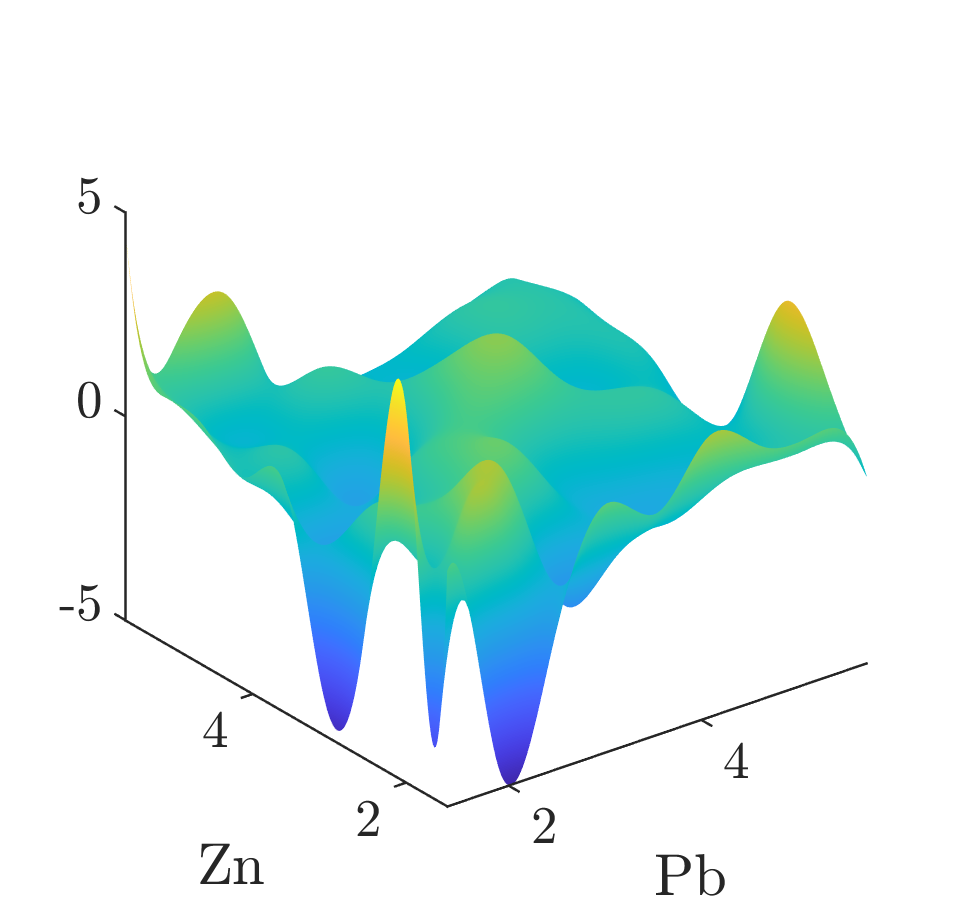}
    \caption{Penalized likelihood estimates of independent part (left) and interactive part (right) in $\mathcal{B}^2(\Omega)$ (top) and $L^2_0(\Omega)$ (bottom) for the Chomutov district.}
    \label{chomutov_decomposition}
\end{figure}

The independent part (left) describes the joint structure determined solely by the respective geometrical marginals of Pb and Zn. In contrast, the interactive part (right) captures the remaining joint structure that cannot be explained by the marginal distributions. The decomposition, therefore, allows an interesting insight, distinguishing effects on the geometric marginals for the distribution of Pb and Zn from those on their dependence structure (similar to the decomposition in copula methods, e.g., \cite{genest22}).

For the Chomutov district, the independent part exhibits a dominant smooth mode corresponding to the main concentration range of Pb and Zn, indicating that a substantial part of the joint distribution is determined by their geometrical marginals. However, the interactive part, particularly apparent in the $L^2_0(\Omega)$ representation (bottom right), reveals pronounced localized deviations from independence. In particular, it shows alternating positive and negative regions, while the interaction remains comparatively weak over a substantial part of the domain. This indicates that the dependence between Pb and Zn is predominantly localized, mostly at the edges of the given domain, rather than uniformly distributed.

\section{Conclusion}\label{Conclusion}

In this paper, we propose a penalized maximum likelihood framework for the direct estimation of probability density functions in Bayes Hilbert spaces. The key contribution of the work lies in avoiding the commonly used two-stage strategy, where densities are first approximated from raw observations and subsequently represented within a functional or Bayes space framework. Instead, the proposed method estimates the density directly from the observed sample through a likelihood-based formulation while preserving the intrinsic compositional geometry of probability density functions.

The methodology combines three key ingredients: the Bayes Hilbert space representation of densities, the clr transformation, and $Z\!B$-spline basis functions satisfying the zero-integral constraint by construction. Together, these components yield a unified likelihood-based framework in which the unknown clr-transformed density is represented directly by spline coefficients and estimated by minimizing a penalized negative log-likelihood. Smoothness is controlled through quadratic penalties imposed on the spline coefficients.

The proposed framework was developed for both univariate and bivariate densities. In the bivariate case, the method naturally incorporates the orthogonal decomposition of densities into independent and interactive components. This decomposition constitutes one of the key advantages of the proposed framework, since it provides an interpretable representation of both independent and interactive structures within a unified likelihood-based setting.

The simulation studies demonstrated that the performance of the proposed methodology depends on the structural complexity of the underlying density and on the selected spline representation. In the univariate settings, the best-performing spline configurations achieved lower median normalized integrated squared error than all considered kernel density estimators. In the bivariate setting, the relative performance depended on the underlying density and the bandwidth selection rule.

The empirical application to geochemical data further demonstrated the practical applicability of the methodology for real distributional data. In both the univariate and bivariate settings, the proposed estimator attained lower cross-validated negative log-likelihood values than the examined kernel density estimators.

Overall, the proposed approach provides a direct likelihood-based alternative to existing two-stage procedures for the estimation and functional representation of density data in Bayes spaces. Further directions for future research include extensions to higher-dimensional densities and subsequent functional data analysis procedures based directly on the resulting spline representations.

\section*{Acknowledgment}
We gratefully acknowledge the support of this research by following grants: IGA\_PrF\_2026\_018 Mathematical Models and the Czech Science Foundation grant 25-15447S.

\section*{Data availability}
Data together with codes are available at \href{https://github.com/skorst01/Penalized-Likelihood-Density-Estimation}{https://github.com/skorst01/Penalized-Likelihood-Density-Estimation}.

\section*{Conflict of interest}
We declare that there are no conflicts of interest regarding the publication of this paper.

\begin{appendices}

\section{Proofs}\label{app:proofs}
In this appendix, we provide the proofs of the theorems stated in Section \ref{PML}.

\bigskip
\noindent
\textbf{Proof of Theorem \ref{thm:existence_uniqueness_univariate}}
We divide the proof into several steps.

\smallskip
\noindent
\textit{Step 1: Continuity of the penalized negative log-likelihood}

\noindent
For fixed $\rho>0$ and $\mathrm{d}\in\{1,\ldots,g+k-1\}$, consider the penalized negative log-likelihood function
\begin{equation*}
l_{\rho}(\mathbf{z}) = -\mathbf{c}^{\top}\mathbf{z} + N \ln \int_I \exp\bigl(\mathbf{Z}_{k+1} (t)\mathbf{z}\bigr)\,\mathrm{d}t + \rho\,\mathbf{z}^{\top}\mathbf{P}_\mathrm{d}\, \mathbf{z}.
\end{equation*}
The first and the third terms are continuous functions of $\mathbf z$. It therefore remains to consider
$$
B(\mathbf z) := 
\int_I \exp\bigl(\mathbf Z_{k+1}(t)\mathbf z\bigr)\,\mbox{d}t.
$$
Since the $Z\!B$-spline basis functions are continuous on the compact interval $I$, the mapping $(t,\mathbf z)\mapsto \exp\bigl(\mathbf Z_{k+1}(t)\mathbf z\bigr)$ is continuous. Hence, $B(\mathbf z)$ is continuous with respect to $\mathbf z$. Moreover, $B(\mathbf z)>0$ for every $\mathbf z\in\mathbb R^{g+k}$. Therefore, $\ln B(\mathbf z)$ is continuous, and consequently $l_\rho$ is continuous on $\mathbb R^{g+k}$.

\medskip
\noindent
\textit{Step 2: Derivation of the gradient and Hessian matrix} 

\noindent
Differentiating $l_\rho(\mathbf z)$ with respect to $\mathbf z$ gives
\begin{equation}\label{eq:grad}
\nabla_{\mathbf z} l_\rho(\mathbf z) = -\mathbf{c} + N \frac{ \displaystyle\int_I \mathbf Z_{k+1}^{\top}(t) \exp\bigl(\mathbf Z_{k+1}(t)\mathbf z\bigr)\,\mbox{d}t }{
\displaystyle\int_I \exp\bigl(\mathbf Z_{k+1}(t)\mathbf z\bigr)\,\mbox{d}t} + 2\rho \mathbf{P}_\mathrm{d}\mathbf z.
\end{equation}
Let
$$ A(\mathbf z) := \int_I \mathbf Z_{k+1}^{\top}(t) \exp\bigl(\mathbf Z_{k+1}(t)\mathbf z\bigr)\,\mbox{d}t. 
$$
Using $B(\mathbf z)$ defined in Step~1, the gradient can be written as 
$$
\nabla_{\mathbf z}l_\rho(\mathbf z) = -\mathbf c + N\frac{A(\mathbf z)}{B(\mathbf z)} + 2\rho \mathbf{P}_\mathrm{d}\mathbf z.
$$
Since $\nabla_{\mathbf z}B(\mathbf z)=A(\mathbf z)$ and
$$
\nabla_{\mathbf{z}}A(\mathbf z) = \int_I \mathbf Z_{k+1}^{\top}(t)\mathbf Z_{k+1}(t) \exp\bigl(\mathbf Z_{k+1}(t)\mathbf z\bigr)\,\mbox{d}t,
$$
differentiating the gradient once more gives the Hessian matrix
$$
\nabla_{\mathbf{zz}}^2l_\rho(\mathbf z) = N\left( \frac{\nabla_{\mathbf z}A(\mathbf z)}{B(\mathbf z)} - \frac{A(\mathbf z)A^\top(\mathbf z)}{B^2(\mathbf z)} \right) + 2\rho \mathbf{P}_\mathrm{d}.
$$

\medskip
\noindent
\textit{Step 3: Positive definiteness of the Hessian matrix} 

\noindent
For any $\mathbf{a}\in\mathbb R^{g+k}$, $\mathbf{a}\neq\mathbf 0$, we have
$$
\begin{aligned}
\mathbf{a}^\top\nabla_{\mathbf{zz}}^2l_\rho(\mathbf z)\mathbf{a} ={}& N\left[ \frac{1}{B(\mathbf z)}
\int_I \bigl(\mathbf Z_{k+1}(t)\mathbf a\bigr)^2
\exp\bigl(\mathbf Z_{k+1}(t)\mathbf z\bigr)\,\mbox{d}t
\right.\\
&\left. -\frac{1}{B^2(\mathbf z)} \left( \int_I \mathbf Z_{k+1}(t)\mathbf a\, \exp\bigl(\mathbf Z_{k+1}(t)\mathbf z\bigr)\,\mbox{d}t \right)^2 \right] + 2\rho\,\mathbf a^\top \mathbf{P}_\mathrm{d}\mathbf a.
\end{aligned}
$$
Define
\[ 
p_{\mathbf z}(t) = \frac{\exp\bigl(\mathbf Z_{k+1}(t)\mathbf z\bigr)} {B(\mathbf z)},
\]
which is a strictly positive probability density function on $I$. Then
\[
\mathbf{a}^\top\nabla_{\mathbf{zz}}^2l_\rho(\mathbf{z})\mathbf{a} = N\operatorname{Var}_{p_{\mathbf z}} \bigl(\mathbf Z_{k+1}(t)\mathbf a\bigr) + 2\rho\|\textbf{D}_\mathrm{d}\mathbf a\|_2^2.
\]
The second term is nonnegative. Moreover, if $\operatorname{Var}_{p_{\mathbf z}}
\bigl(\mathbf Z_{k+1}(t)\mathbf a\bigr)=0,$ then $\mathbf Z_{k+1}(t)\mathbf a$ is constant almost everywhere on $I$. Since every linear combination of $Z\!B$-splines has a zero integral over $I$, this constant must be zero. The linear independence of the $Z\!B$-spline basis then implies $\mathbf a=\mathbf 0$, which contradicts the assumption $\mathbf a\neq\mathbf 0$. Hence $\operatorname{Var}_{p_{\mathbf z}} \bigl(\mathbf Z_{k+1}(t)\mathbf a\bigr)>0$
for every $\mathbf a\neq\mathbf 0$, and consequently
\[
\mathbf a^\top\nabla_{\mathbf{zz}}^2l_\rho(\mathbf z)\mathbf a>0.
\]
Therefore, $\nabla^2_{\mathbf z\mathbf z} l_\rho(\mathbf z)$ is positive definite for every $\mathbf z\in\mathbb R^{g+k}$, and consequently $l_\rho$ is strictly convex. The same argument with $\rho=0$ shows that the unpenalized negative log-likelihood $l_0$ is also strictly convex.

\smallskip
\noindent
\textit{Step 4: Existence of the minimizer}\\ 
Consider first the unpenalized negative log-likelihood
$$
l_0(\mathbf z) = -\mathbf c^\top \mathbf z + N\ln \int_I \exp\bigl(\mathbf Z_{k+1}(t)\mathbf z\bigr)\,\mathrm dt.
$$
The corresponding model forms a finite-dimensional exponential family generated by the $Z\!B$-spline basis. Since the interval $I$ is bounded and the $Z\!B$-spline basis functions are bounded on $I$, the normalizing integral is finite for every $\mathbf{z}\in\mathbb{R}^{g+k}$. Moreover, the family is minimal, because the $Z\!B$-spline basis admits no nontrivial linear combination that is constant on $I$: by the zero-integral property, such a constant must be zero, and the linear independence of the basis then implies that all coefficients vanish. Under the nondegeneracy assumption
$$
\frac{\mathbf{c}^\top}{N} \in \operatorname{int}\operatorname{conv} \left\{ \mathbf{Z}_{k+1}(t): t\in I \right\},
$$
the existence of the maximum-likelihood estimator follows from the standard theory of exponential families; see ~\cite[Theorem~9.13]{Barnd1978}. Consequently, there exists $\widehat{\mathbf{z}}_0\in\mathbb{R}^{g+k}$ such that
$$
l_0(\widehat{\mathbf{z}}_0) = \min_{\mathbf{z}} l_0(\mathbf{z}).
$$
As shown above, the Hessian matrix of $l_0$ is positive definite for every $\mathbf z\in\mathbb R^{g+k}$. Hence, $l_0$ is strictly convex and $\widehat{\mathbf z}_0$ is its unique minimizer. 

We now show that $l_0$ is coercive. Suppose, to the contrary, that $l_0$ is not coercive. Then there exists a sequence $\{\mathbf z_n\}\subset\mathbb R^{g+k}$ such that $\|\mathbf z_n\|_2\to\infty $, while the sequence $\{l_0(\mathbf z_n)\}$ is bounded from above. Thus, for some $C\in \mathbb R$, $l_0(\mathbf z_n)\leq C $ along a subsequence. Set
$$
r_n = \|\mathbf z_n-\widehat{\mathbf z}_0\|_2, \qquad \mathbf a_n = \frac{\mathbf z_n-\widehat{\mathbf z}_0}{r_n}.
$$
Then $r_n\to\infty$ and $\|\mathbf a_n\|_2=1$. Since the unit sphere in $\mathbb R^{g+k}$ is compact, there exists a subsequence, denoted again by $\{\mathbf a_n\}$, and a vector $\mathbf a$ with $\|\mathbf a\|_2=1$ such that
$$
\mathbf a_n\to\mathbf a.
$$
Let $r>0$ be arbitrary. For all sufficiently large $n$, $r<r_n$, and
$$
\widehat{\mathbf z}_0+r\mathbf a_n = \left(1-\frac{r}{r_n}\right)\widehat{\mathbf z}_0 + \frac{r}{r_n}\mathbf z_n. 
$$
By convexity of $l_0$,
$$
l_0(\widehat{\mathbf z}_0+r\mathbf a_n) \leq \left(1-\frac{r}{r_n}\right) l_0(\widehat{\mathbf z}_0) + \frac{r}{r_n}l_0(\mathbf z_n).
$$
Since $r_n\to\infty$ and $l_0(\mathbf z_n)\leq C$, passing to the limit and using continuity of $l_0$ yields 
$$
l_0(\widehat{\mathbf z}_0+r\mathbf a) \leq l_0(\widehat{\mathbf z}_0).
$$
On the other hand, $\widehat{\mathbf z}_0$ is a global minimizer of $l_0$, and therefore
$$
l_0(\widehat{\mathbf z}_0+r\mathbf a) = l_0(\widehat{\mathbf z}_0).
$$
Since $r>0$ and $\|\mathbf a\|_2=1$, the point $\widehat{\mathbf z}_0+r\mathbf a$ differs from $\widehat{\mathbf z}_0$, which contradicts the uniqueness of the minimizer. Hence,
$$
l_0(\mathbf z)\to+\infty \qquad\text{as}\qquad \|\mathbf z\|_2\to\infty,
$$
and therefore $l_0$ is coercive. Finally, since
$\mathbf{P}_\mathrm{d}=\mathbf{D}_\mathrm{d}^\top \mathbf{D}_\mathrm{d},$ we have $\mathbf z^\top \mathbf{P}_\mathrm{d}\mathbf z = \|\mathbf{D}_\mathrm{d}\mathbf z\|_2^2 \geq 0.$ Consequently,
$$
l_\rho(\mathbf z) = l_0(\mathbf z) + \rho\,\mathbf z^\top \mathbf{P}_\mathrm{d}\mathbf z \geq l_0(\mathbf z).
$$
The coercivity of $l_0$ therefore implies 
$$
l_\rho(\mathbf z)\to+\infty \qquad\text{as}\qquad \|\mathbf z\|_2\to\infty.
$$
Thus, $l_\rho$ is coercive. Since $l_\rho$ is continuous on $\mathbb R^{g+k}$, it attains its global minimum. Hence, there exists $\mathbf z_\rho^*\in\mathbb R^{g+k}$ such that $l_\rho(\mathbf z_\rho^*) = \min_{\mathbf{z}}l_\rho(\mathbf z).$

\smallskip
\noindent
\textit{Step 5: Uniqueness of the minimizer}\\
By Step 3, the penalized negative log-likelihood $l_\rho$ is strictly convex on $\mathbb R^{g+k}$. By Step~4, $l_\rho$ attains its global minimum at some $\mathbf z_\rho^* \in \mathbb R^{g+k}$. Since a strictly convex function can have at most one global minimizer, $\mathbf z_\rho^*$ is unique. Therefore, there exists a unique vector $\mathbf z_\rho^* \in \mathbb R^{g+k}$ such that $l_\rho(\mathbf z_\rho^*) = \min_{\mathbf{z}} l_\rho(\mathbf z).$
\qed

\bigskip
\noindent
\textbf{Proof of Theorem \ref{thm:existence_uniqueness_bivariate}}
The proof follows the same general arguments as the proof of Theorem~\ref{thm:existence_uniqueness_univariate}. We therefore provide the explicit expressions for the gradient and Hessian matrix and indicate the modifications required for the bivariate setting.

For fixed $\rho_x>0$, $\rho_y>0$ and $\mathrm{d}_x\in\{1,\ldots,g+k-1\}$, $\mathrm{d}_y\in\{1,\ldots,h+l-1\}$, consider 
\[
l_{\rho_x,\rho_y}(\boldsymbol{\theta}) = -\mathbf{h}^{\top}\boldsymbol{\theta}
+ N\ln \int\!\!\!\int_{\Omega} \exp\bigl(\boldsymbol{\Phi}(t,s)\boldsymbol{\theta}\bigr)
\,\mathrm{d}t\,\mathrm{d}s + \boldsymbol{\theta}^{\top}
\mathbb{P}_{\rho_x,\rho_y} \boldsymbol{\theta}.
\]
Let
\[
B(\boldsymbol{\theta}) = \int\!\!\!\int_{\Omega} \exp\bigl(\boldsymbol{\Phi}(t,s)\boldsymbol{\theta}\bigr) \,\mathrm{d}t\,\mathrm{d}s
\]
and
\[
A(\boldsymbol{\theta}) = \int\!\!\!\int_{\Omega} \boldsymbol{\Phi}^{\top}(t,s)
\exp\bigl(\boldsymbol{\Phi}(t,s)\boldsymbol{\theta}\bigr) \,\mathrm{d}t\,\mathrm{d}s.
\]
Then
\begin{equation}\label{eq:grad2D}
\nabla_{\boldsymbol{\theta}} l_{\rho_x,\rho_y}(\boldsymbol{\theta}) =
-\mathbf{h} + N\frac{A(\boldsymbol{\theta})} {B(\boldsymbol{\theta})}
+ 2\mathbb{P}_{\rho_x,\rho_y}\boldsymbol{\theta}.
\end{equation}
Since $\nabla_{\boldsymbol{\theta}} B(\boldsymbol{\theta}) = \mathbf{A}(\boldsymbol{\theta})$ and
\[
\nabla_{\boldsymbol{\theta}}A(\boldsymbol{\theta}) = \int\!\!\!\int_{\Omega} \boldsymbol{\Phi}^{\top}(t,s)\boldsymbol{\Phi}(t,s) \exp\bigl(\boldsymbol{\Phi}(t,s)\boldsymbol{\theta}\bigr) \,\mathrm{d}t\,\mathrm{d}s,
\]
differentiating the gradient once more gives
\[
\nabla_{\boldsymbol{\theta}\boldsymbol{\theta}}^2 l_{\rho_x,\rho_y}(\boldsymbol{\theta}) = N\left(
\frac{\nabla_{\boldsymbol{\theta}}A(\boldsymbol{\theta})}{ B(\boldsymbol{\theta})}
- \frac{ A(\boldsymbol{\theta})A^{\top}(\boldsymbol{\theta})}{B^2(\boldsymbol{\theta})}
\right)+2\mathbb{P}_{\rho_x,\rho_y}.
\]
Define
\[
p_{\boldsymbol{\theta}}(t,s) = \frac{ \exp\bigl(\boldsymbol{\Phi}(t,s)\boldsymbol{\theta}\bigr)}{
B(\boldsymbol{\theta})}
\]
which is a strictly positive probability density on $\Omega$. Then for any nonzero $\mathbf{a}\in\mathbb{R}^{(g+k+1)(h+l+1)-1}$, we obtain
\[
\mathbf{a}^{\top} \nabla_{\boldsymbol{\theta}\boldsymbol{\theta}}^2 l_{\rho_x,\rho_y}(\boldsymbol{\theta}) \mathbf{a} = N\,\operatorname{Var}_{p_{\boldsymbol{\theta}}}
\bigl( \boldsymbol{\Phi}(t,s)\mathbf{a} \bigr) + 2\mathbf{a}^{\top} \mathbb{P}_{\rho_x,\rho_y}
\mathbf{a}.
\]
The second term is nonnegative since $\mathbb{P}_{\rho_x,\rho_y}$ is positive semidefinite. Moreover, if $\operatorname{Var}_{p_{\boldsymbol{\theta}}}\bigl(\boldsymbol{\Phi}(t,s)\mathbf{a}\bigr)=0$, then $\boldsymbol{\Phi}(t,s)\mathbf{a}$ is constant almost everywhere on $\Omega$. Since every linear combination of the basis functions constituting $\boldsymbol{\Phi}(t,s)$ has zero integral over $\Omega$, this constant must be zero. By linear independence of the bivariate $Z\!B$-spline basis, this implies $\mathbf{a}=\mathbf{0}$, which contradicts the assumption  $\mathbf{a}\neq\mathbf{0}$. Hence, $\operatorname{Var}_{p_{\boldsymbol{\theta}}}\bigl(\boldsymbol{\Phi}(t,s)\mathbf{a} \bigr)>0$ for every nonzero $\mathbf{a}$, and consequently
\[
\mathbf{a}^{\top} \nabla_{\boldsymbol{\theta}\boldsymbol{\theta}}^2
l_{\rho_x,\rho_y}(\boldsymbol{\theta}) \mathbf{a}>0.
\]
Therefore, $l_{\rho_x,\rho_y}$ is strictly convex. The same argument with $\rho_x=\rho_y=0$ shows that the unpenalized negative log-likelihood 
\[
l_{0,0}(\boldsymbol{\theta}) = -\mathbf{h}^{\top}\boldsymbol{\theta} +
N\ln \int\!\!\!\int_{\Omega} \exp\bigl(\boldsymbol{\Phi}(t,s)\boldsymbol{\theta}\bigr)
\,\mathrm{d}t\,\mathrm{d}s
\]
is also strictly convex. It remains to establish the existence of the minimizer. The argument is analogous to that used in the proof of Theorem~\ref{thm:existence_uniqueness_univariate}. Consider first the unpenalized negative log-likelihood $l_{0,0}$. The corresponding model is a finite-dimensional exponential family generated by the components of $\boldsymbol{\Phi}(t,s)$. Since $\Omega$ is compact and all basis functions constituting $\boldsymbol{\Phi}(t,s)$ are bounded on $\Omega$, the normalizing integral is finite for every $\boldsymbol{\theta}$. Moreover, the family is minimal, since no nontrivial linear combination of the bivariate $Z\!B$-spline basis functions can be constant on $\Omega$: by the zero-integral property, such a constant must be zero, and linear independence then implies that all coefficients vanish. Under the nondegeneracy assumption
\[
\frac{\mathbf{h}^{\top}}{N} \in \operatorname{int}\operatorname{conv} \left\{
\boldsymbol{\Phi}(t,s):(t,s)\in\Omega \right\},
\]
the existence of the maximum-likelihood estimator follows from the standard theory of exponential families; see \cite[Theorem~9.13]{Barnd1978}. Hence, $l_{0,0}$ has a minimizer. Since $l_{0,0}$ is strictly convex, as shown above, this minimizer is unique. As in the proof of Theorem~\ref{thm:existence_uniqueness_univariate}, $l_{0,0}$ is continuous. Since it is strictly convex and has a unique global minimizer, exactly the same convexity argument as in the proof of Theorem~\ref{thm:existence_uniqueness_univariate} shows that $l_{0,0}$ is coercive, i.e.,
\[
l_{0,0}(\boldsymbol{\theta}) \longrightarrow +\infty \qquad\text{as}\qquad \|\boldsymbol{\theta}\|_2\longrightarrow\infty.
\]
Furthermore, $\mathbb{P}_{\rho_x,\rho_y}$ is positive semidefinite by construction. Hence, $\boldsymbol{\theta}^{\top} \mathbb{P}_{\rho_x,\rho_y} \boldsymbol{\theta}\geq 0,$ and therefore
\[
l_{\rho_x,\rho_y}(\boldsymbol{\theta}) = l_{0,0}(\boldsymbol{\theta}) +
\boldsymbol{\theta}^{\top} \mathbb{P}_{\rho_x,\rho_y} \boldsymbol{\theta}
\geq l_{0,0}(\boldsymbol{\theta}).
\]
Consequently,
\[
l_{\rho_x,\rho_y}(\boldsymbol{\theta}) \longrightarrow +\infty
\qquad\text{as}\qquad \|\boldsymbol{\theta}\|_2\longrightarrow\infty,
\]
so that $l_{\rho_x,\rho_y}$ is coercive.

As in the proof of Theorem~\ref{thm:existence_uniqueness_univariate}, the function $l_{\rho_x,\rho_y}$ is continuous. Hence, by continuity and coercivity, it attains a global minimum. Since $l_{\rho_x,\rho_y}$ is strictly convex, as shown above, this minimum is unique. Therefore, there exists a unique $\boldsymbol{\theta}^{*}_{\rho_x,\rho_y}$ such that
\[
l_{\rho_x,\rho_y} \bigl(\boldsymbol{\theta}^{*}_{\rho_x,\rho_y}\bigr)
= \min_{\boldsymbol{\theta}} l_{\rho_x,\rho_y}(\boldsymbol{\theta}).
\]
\qed

\section{Supplementary materials for Simulation study}
\label{app:simulation}
In this section, we include supplementary outputs for Section \ref{Simulation}.

For the univariate simulation study, Tables \ref{opt_rho_f2} and \ref{opt_rho_f3} report the mean values of optimal penalization parameter $\rho$ for maximum likelihood estimation of densities $f_2(x)$ and $f_3(x)$, respectively, regarding the considered simulation scenario, i.e., sample sizes $N$ and the number of folds $K$. Figures \ref{NISE_f2} and \ref{NISE_f3} show boxplots of measured NISE values for the likelihood estimates of $f_2(x)$ and $f_3(x)$ given the considered number of inner knots $g$ in each simulation scenario. Subsequently, Tables \ref{median_NISE_f2} and \ref{median_NISE_f3} state the corresponding median NISE values. Finally, Figures \ref{final_comp_f2} and \ref{final_comp_f3} depicture the true densities $f_2(x)$ and $f_3(x)$, respectively, together with their corresponding estimates via penalized maximum likelihood and kernel smoothing.

For the bivariate case, Table \ref{opt_rho_f2_2D} gives information about the mean optimal values of penalization parameters $(\rho_x, \rho_y)$ for maximum likelihood estimation of $f_2(x,y)$ with respect to the considered simulation scenario, similarly to the univariate case. Figure \ref{NISE_f2_2D} displays boxplots of measured NISE values for the likelihood estimation of $f_2(x,y)$ and Table \ref{NISE_2D_f2} reports the corresponding median NISE values. Finally, Figure \ref{final_comp_2D_f2} depictures penalized maximum likelihood estimate of $f_2(x,y)$ together with the kernel density estimate.

\begin{table}[h]
\centering
\small
\setlength{\tabcolsep}{5pt}
\begin{tabular}{cc|cccccccccc}
\toprule
 &  & \multicolumn{10}{c}{$g$} \\
\cmidrule(lr){3-12}
$N$ & $K$ & 1 & 2 & 3 & 4 & 5 & 6 & 7 & 8 & 9 & 10 \\
\midrule
3000 & 5  & 0.461& 0.01& 0.044& 0.061& 0.087& 0.152& 0.15& 0.216& 0.133& 0.205 \\
3000 & 10 & 0.469& 0.01& 0.043& 0.06& 0.087& 0.158& 0.16& 0.226& 0.137& 0.218 \\
5000 & 5 & 0.401& 0.009& 0.042& 0.056& 0.08& 0.144& 0.166& 0.233& 0.132& 0.204 \\
5000 & 10 & 0.381& 0.01& 0.044& 0.057& 0.08& 0.151& 0.173& 0.245& 0.138& 0.21 \\
\bottomrule
\end{tabular}
\caption{Mean values of optimal penalization parameters $\rho$ for estimating the density $f_2(x)$ using different numbers of sampled values ($N_1 = 3000$, $N_2 = 5000$) and different number of folds used in \eqref{CV_1D} ($K_1 = 5$, $K_2 = 10$) (rows) for different number of inner knots $g$ (columns).}
\label{opt_rho_f2}
\end{table}

\begin{table}[h]
\centering
\small
\setlength{\tabcolsep}{5pt}
\begin{tabular}{cc|cccccccccc}
\toprule
 &  & \multicolumn{10}{c}{$g$} \\
\cmidrule(lr){3-12}
$N$ & $K$ & 1 & 2 & 3 & 4 & 5 & 6 & 7 & 8 & 9 & 10 \\
\midrule
3000 & 5  & 0.033& 0.205& 0.025& 0.02& 0.024& 0.007& 0.032& 0.01& 0.046& 0.041 \\
3000 & 10 & 0.031& 0.189& 0.024& 0.02& 0.025& 0.007& 0.033& 0.011& 0.046& 0.04\\
5000 & 5 & 0.031& 0.204& 0.025& 0.02& 0.026& 0.007& 0.033& 0.011& 0.049& 0.037 \\
5000 & 10 & 0.031& 0.193& 0.024& 0.02& 0.025& 0.006& 0.033& 0.01& 0.051& 0.036 \\
\bottomrule
\end{tabular}
\caption{Mean values of optimal penalization parameters $\rho$ for estimating the density $f_3(x)$ using different numbers of sampled values ($N_1 = 3000$, $N_2 = 5000$) and different number of folds used in \eqref{CV_1D} ($K_1 = 5$, $K_2 = 10$) (rows) for different number of inner knots $g$ (columns).}
\label{opt_rho_f3}
\end{table}

\begin{figure}[h]
    \centering
    \includegraphics[width=0.4\textwidth]{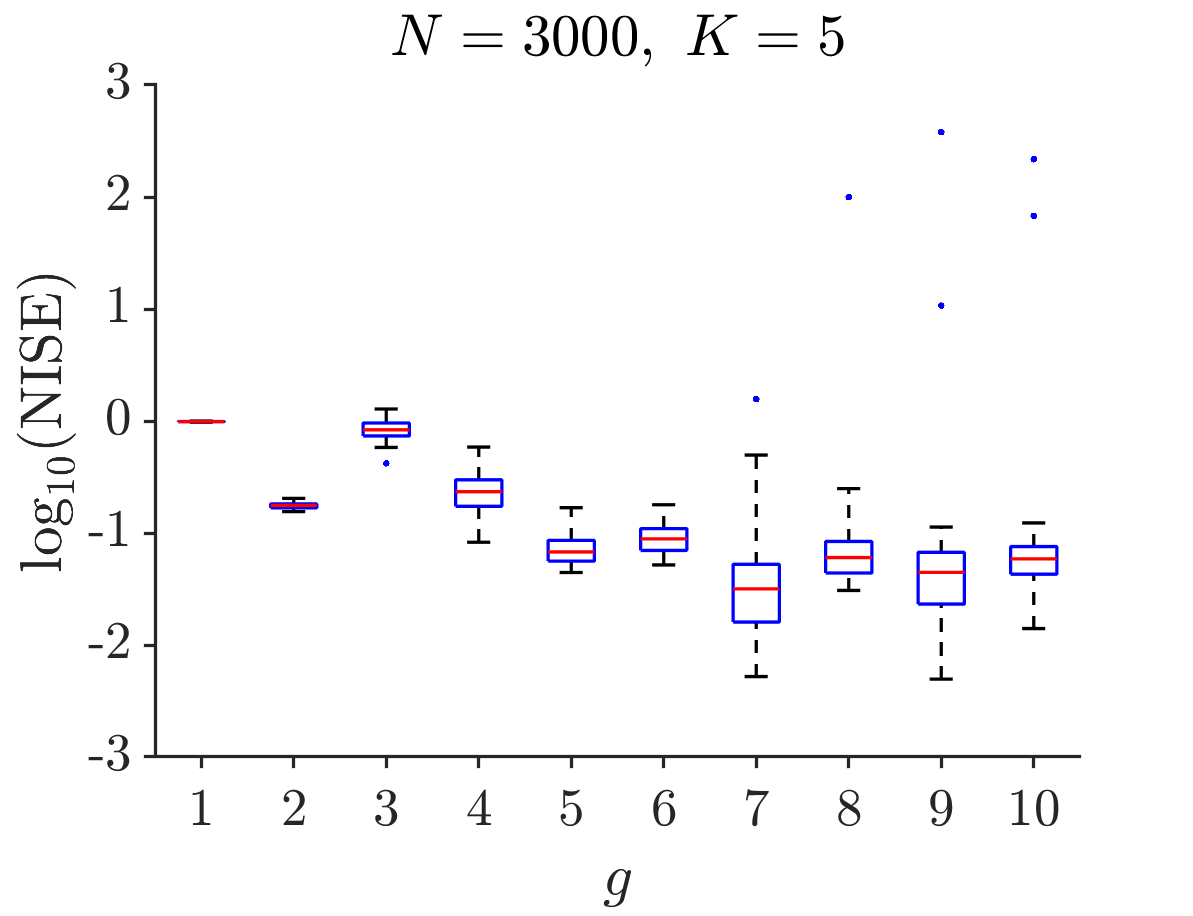}\qquad
    \includegraphics[width=0.4\textwidth]{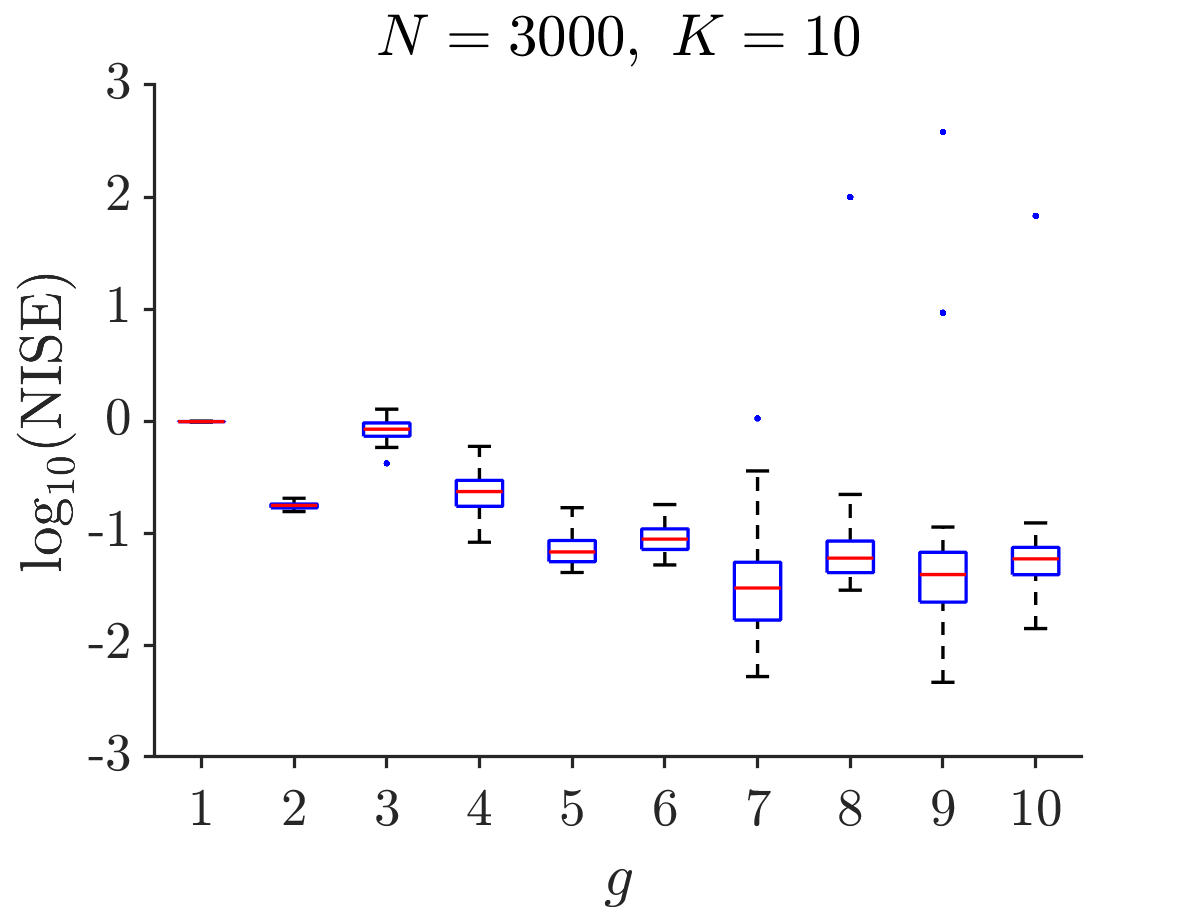}\\[0.5cm]
    \includegraphics[width=0.4\textwidth]{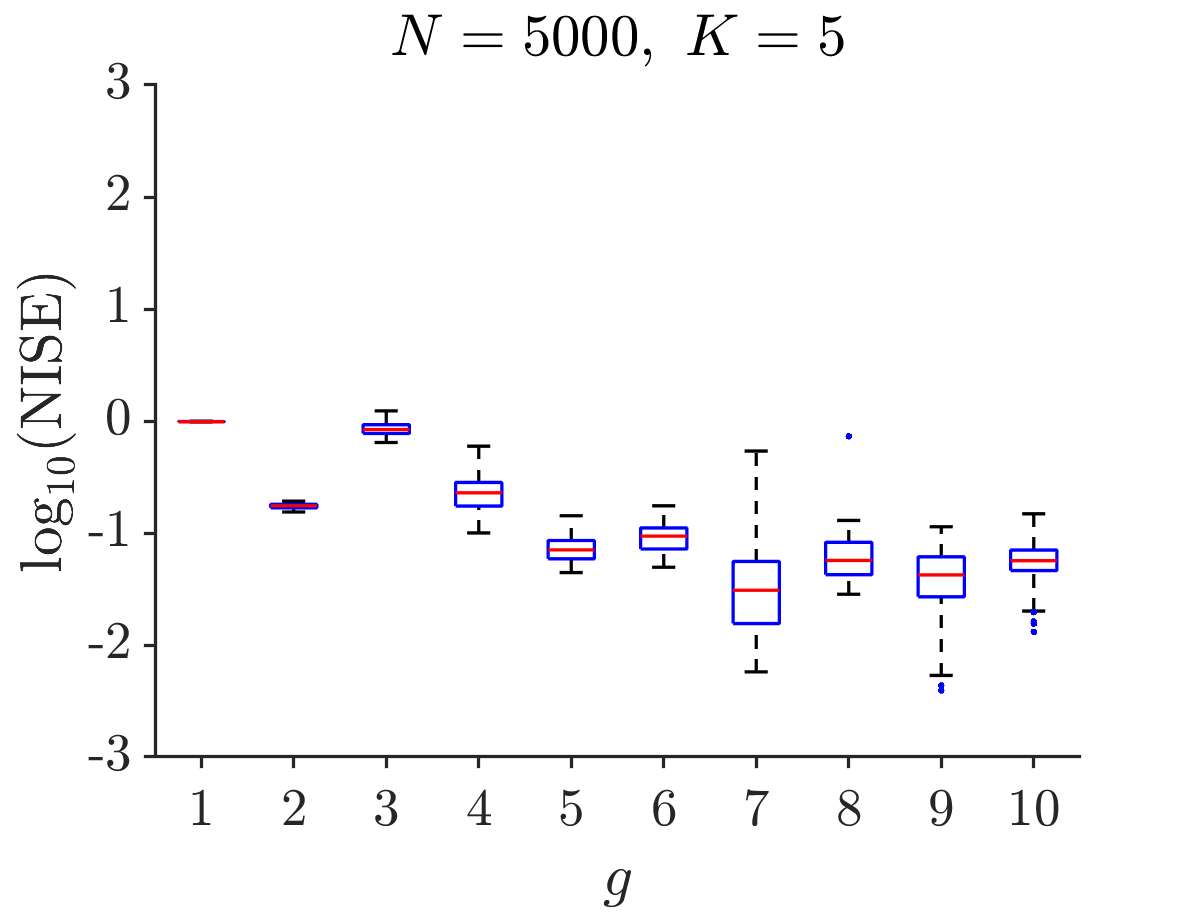}\qquad
    \includegraphics[width=0.4\textwidth]{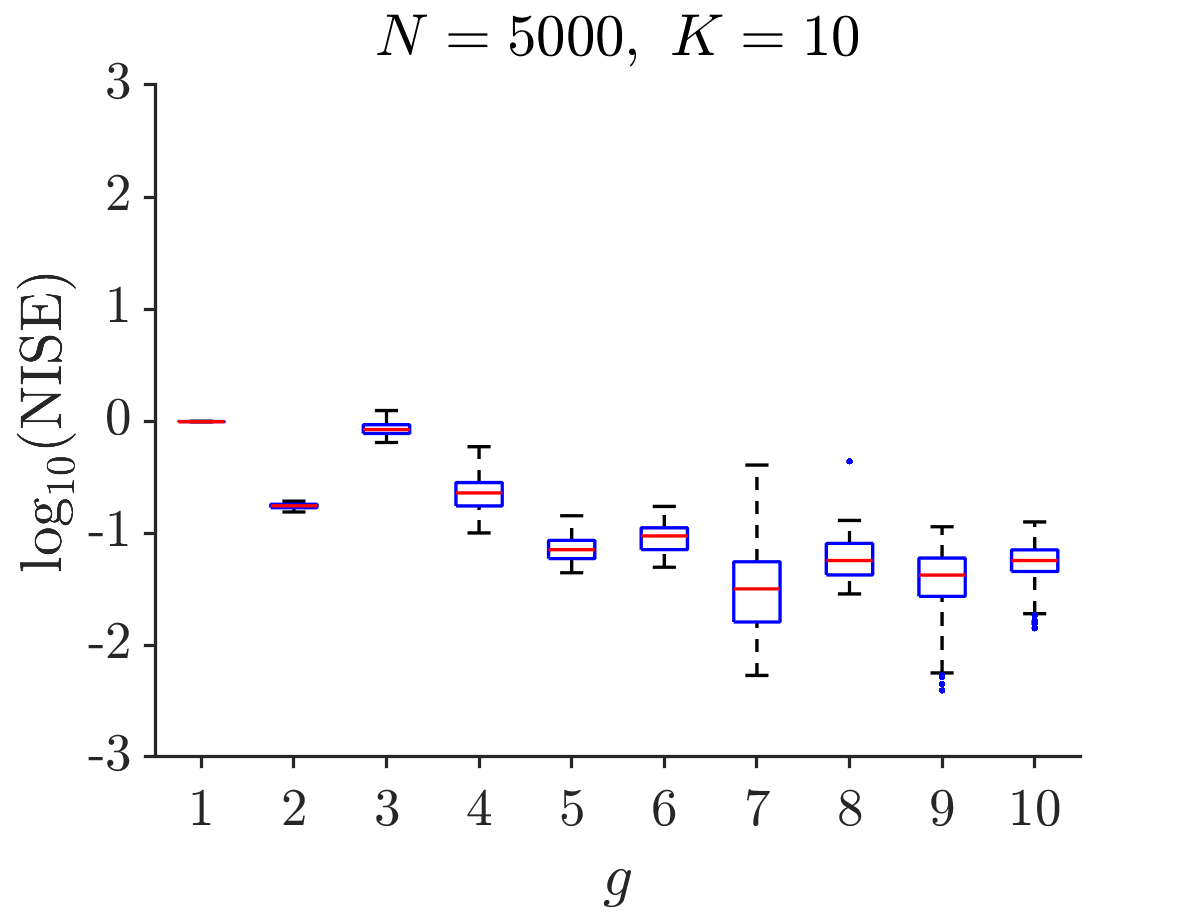}
    \caption{Boxplots of NISE values for the penalized likelihood estimates of $f_2(x)$ for different numbers of inner knots $g$ and the considered combinations of the sample size $N$ and the number of cross-validation folds $K$.}
    \label{NISE_f2}
\end{figure}

\begin{figure}[h]
    \centering
    \includegraphics[width=0.4\textwidth]{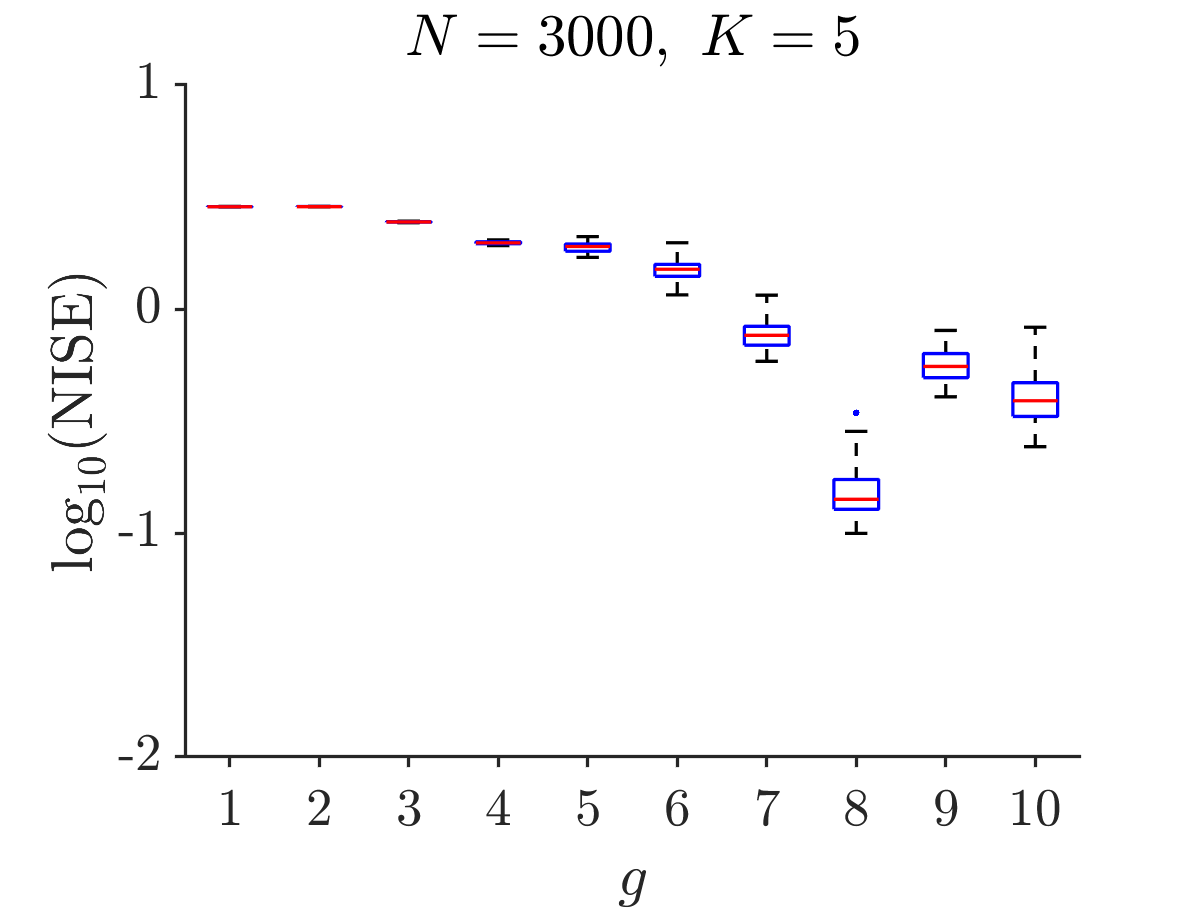}\qquad
    \includegraphics[width=0.4\textwidth]{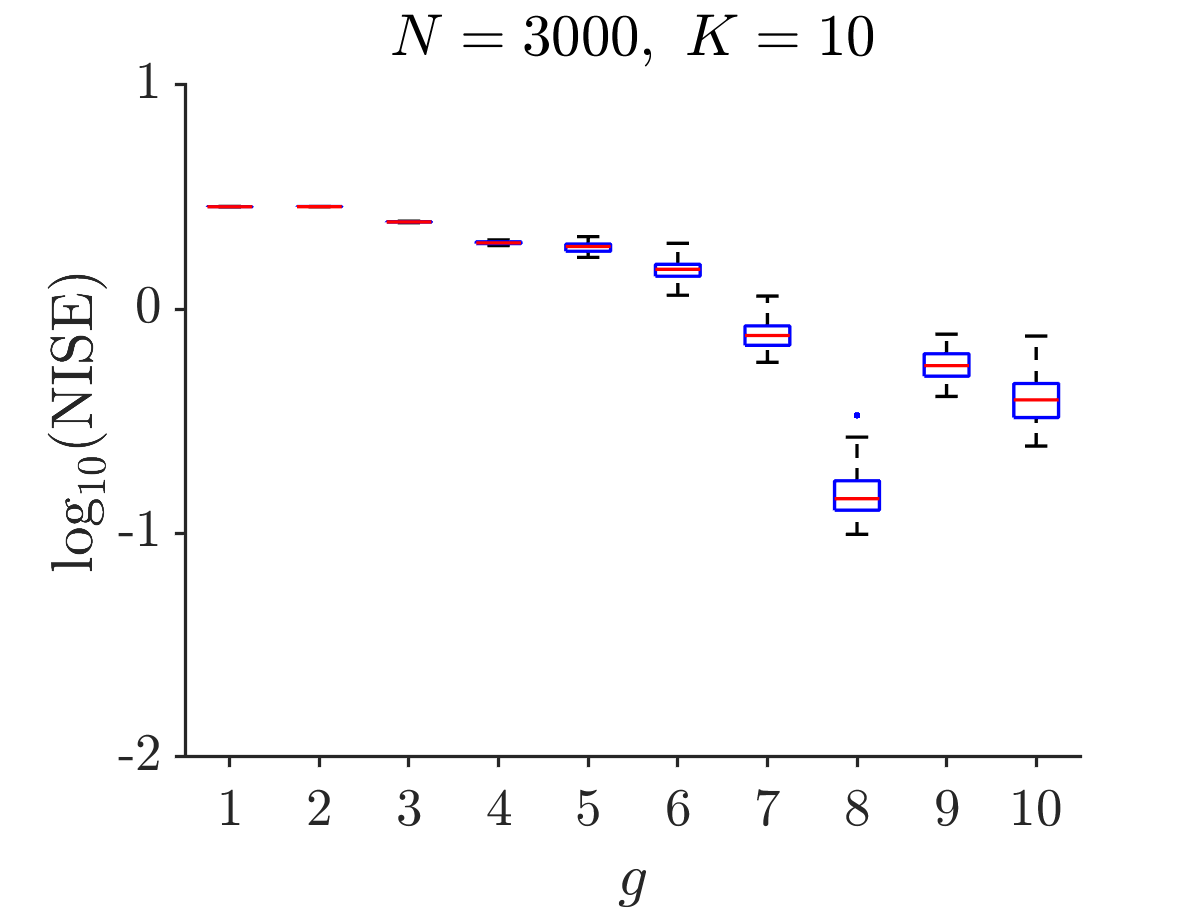}\\[0.5cm]
    \includegraphics[width=0.4\textwidth]{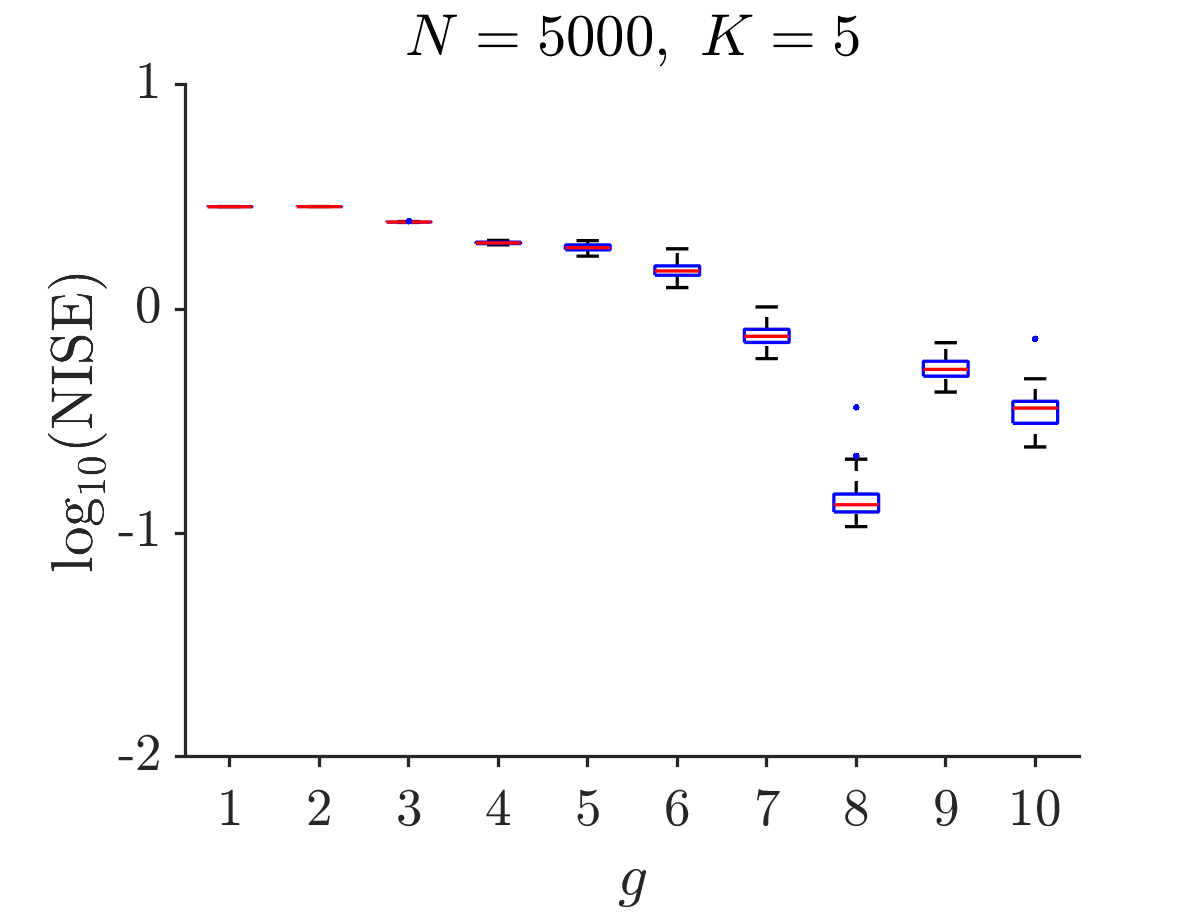}\qquad
    \includegraphics[width=0.4\textwidth]{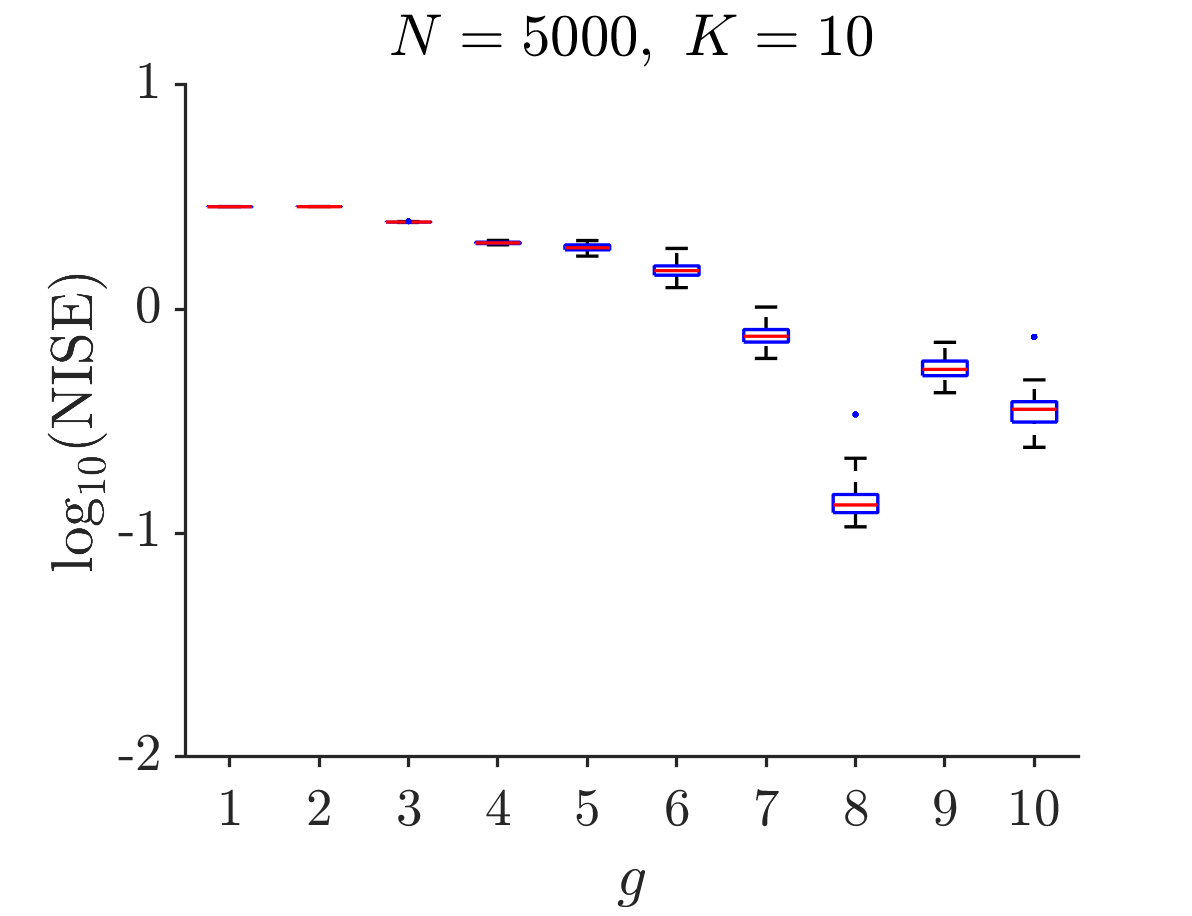}
    \caption{Boxplots of NISE values for the penalized likelihood estimates of $f_3(x)$ for different numbers of inner knots $g$ and the considered combinations of the sample size $N$ and the number of cross-validation folds $K$.}
    \label{NISE_f3}
\end{figure}

\begin{table}[h]
\centering
\small
\setlength{\tabcolsep}{5pt}
\begin{tabular}{cc|cccccccccc}
\toprule
 &  & \multicolumn{10}{c}{$g$} \\
\cmidrule(lr){3-12}
$N$ & $K$ & 1 & 2 & 3 & 4 & 5 & 6 & 7 & 8 & 9 & 10 \\
\midrule
3000 & 5  & 5.921& 1.058& 5.01& 1.402& 0.406& 0.532& \textbf{0.19}& 0.362& 0.267& 0.352 \\
3000 & 10 & 5.922& 1.058& 5.068& 1.407& 0.406& 0.529& \textbf{0.194}& 0.358& 0.255& 0.353 \\
5000 & 5 & 5.916& 1.05& 5.032& 1.371& 0.424& 0.563& \textbf{0.185}& 0.341& 0.253& 0.34 \\
5000 & 10 & 5.916& 1.05& 5.028& 1.368& 0.426& 0.565& \textbf{0.19}& 0.341& 0.253& 0.34 \\
\bottomrule
\end{tabular}
\caption{Median NISE values of $f_2(x)$ using different numbers of sampled values ($N_1 = 3000$, $N_2 = 5000$) and different number of folds used in \eqref{CV_1D} ($K_1 = 5$, $K_2 = 10$) (rows) for different numbers of inner knots (columns) with the highlighted minimal NISE for each setting.}
\label{median_NISE_f2}
\end{table}

\begin{table}[h]
\centering
\small
\setlength{\tabcolsep}{4pt}
\begin{tabular}{cc|cccccccccc}
\toprule
 &  & \multicolumn{10}{c}{$g$} \\
\cmidrule(lr){3-12}
$N$ & $K$ & 1 & 2 & 3 & 4 & 5 & 6 & 7 & 8 & 9 & 10 \\
\midrule
3000 & 5 & 17.165& 17.18& 14.688& 11.842& 11.407& 9.026& 4.576& \textbf{0.848} & 3.326& 2.337\\
3000 & 10 & 17.165& 17.178& 14.687& 11.842& 11.406& 9.023& 4.568& \textbf{0.854}& 3.352& 2.357 \\
5000 & 5 & 17.167& 17.177& 14.672& 11.831& 11.26& 8.876& 4.53& \textbf{0.802}& 3.225& 2.166\\
5000 & 10 & 17.167& 17.177& 14.671& 11.832& 11.26& 8.902& 4.535& \textbf{0.801}& 3.225& 2.139 \\
\bottomrule
\end{tabular}
\caption{Median NISE values of $f_3(x)$ using different numbers of sampled values ($N_1 = 3000$, $N_2 = 5000$) and different number of folds used in \eqref{CV_1D} ($K_1 = 5$, $K_2 = 10$) (rows) for different numbers of inner knots (columns) with the highlighted minimal NISE for each setting.}
\label{median_NISE_f3}
\end{table}

\begin{figure}[h]
    \centering
    \includegraphics[width=0.3\textwidth]{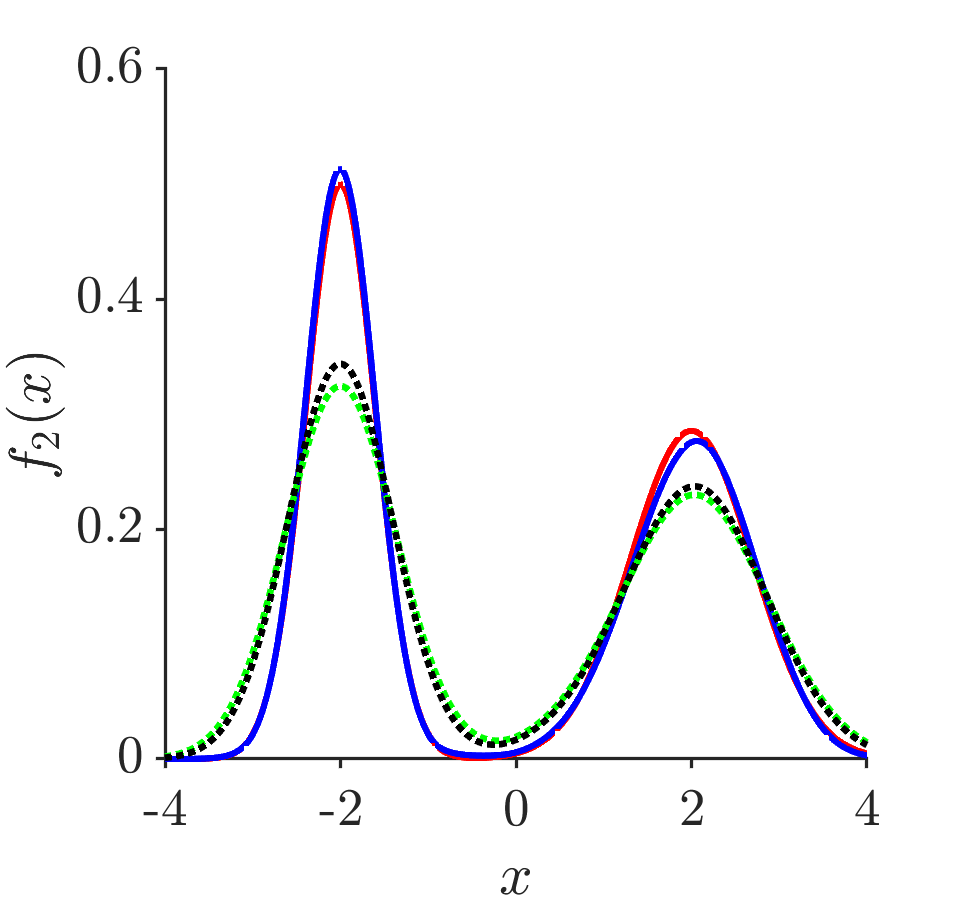}
    \includegraphics[width=0.3\textwidth]{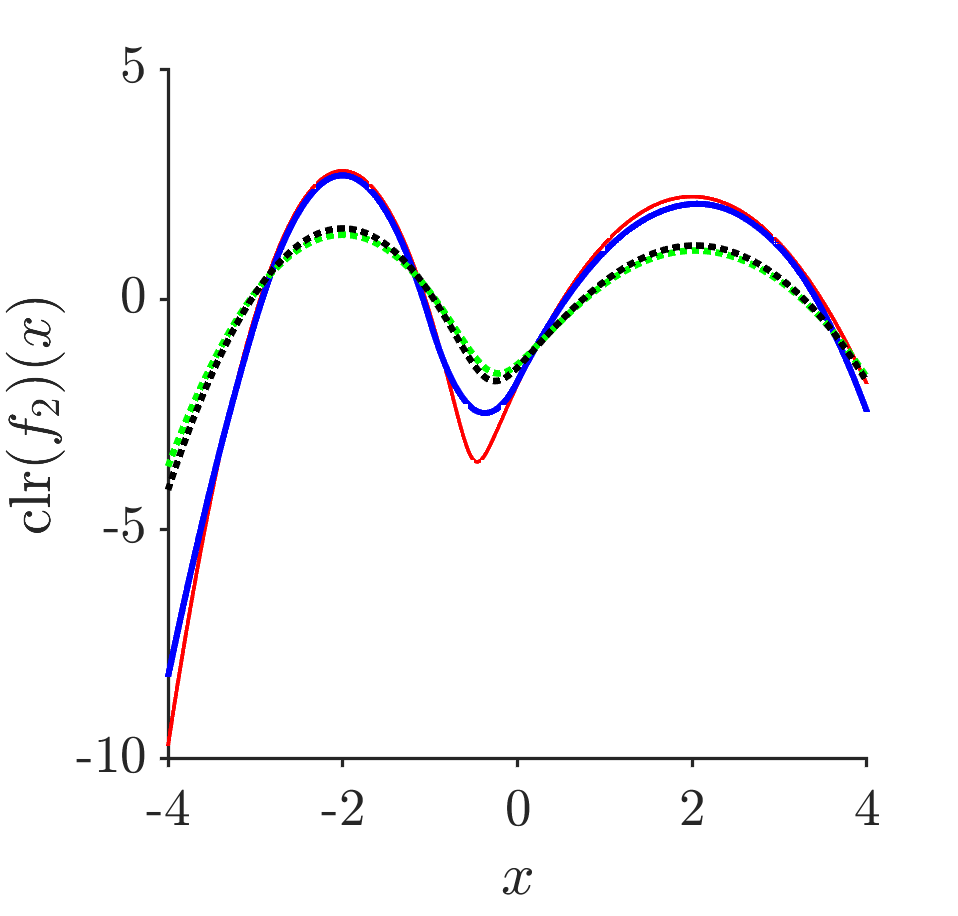}
    \caption{Comparison of the original density $f_2(x)$ (red), the penalized likelihood estimate constructed with $N=3000$ using $g=7$ and $\rho = 0.15$(blue), and the kernel density estimates with $N=3000$ based on Silverman's rule (dotted, green), Scott's rule (dotted, yellow), and the Sheather--Jones method (dotted, black). The estimates based on Scott's rule and the Sheather--Jones method coincide. The unit-integral representatives in $B^2(I)$ are shown in the left panel, and their clr representations in $L_0^2(I)$ are shown in the right panel.}
    \label{final_comp_f2}
\end{figure}

\begin{figure}[h]
    \centering
    \includegraphics[width=0.3\textwidth]{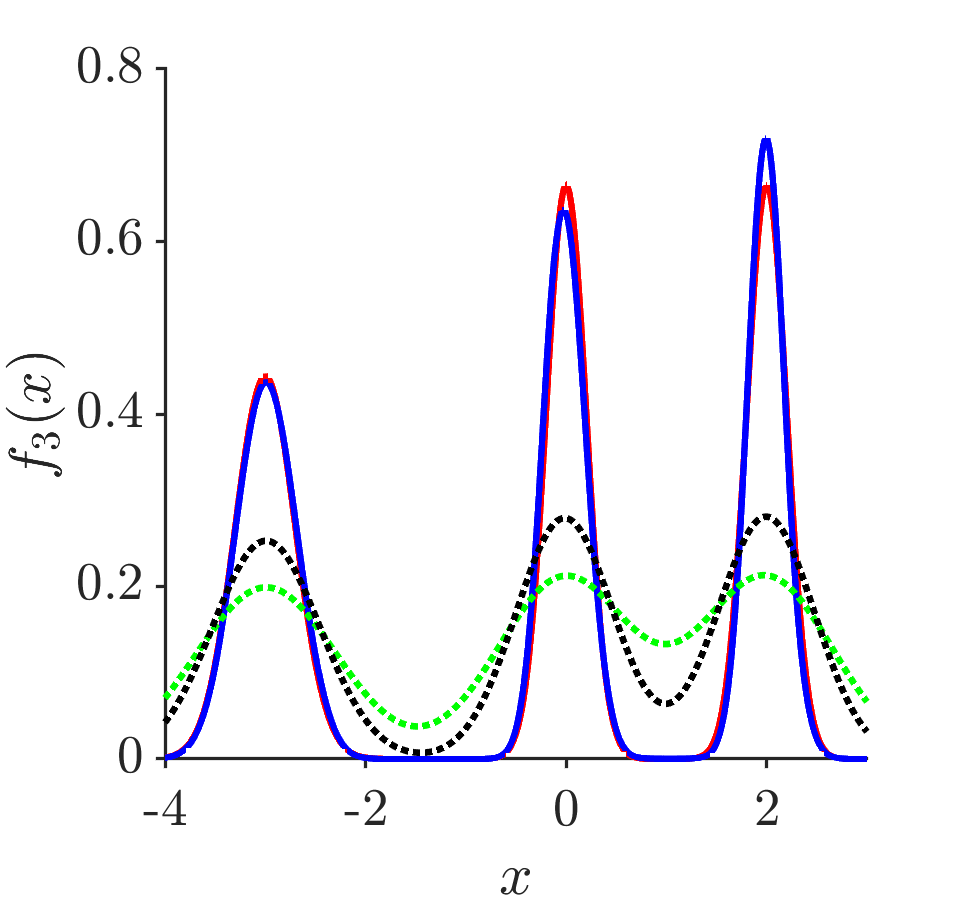}
    \includegraphics[width=0.3\textwidth]{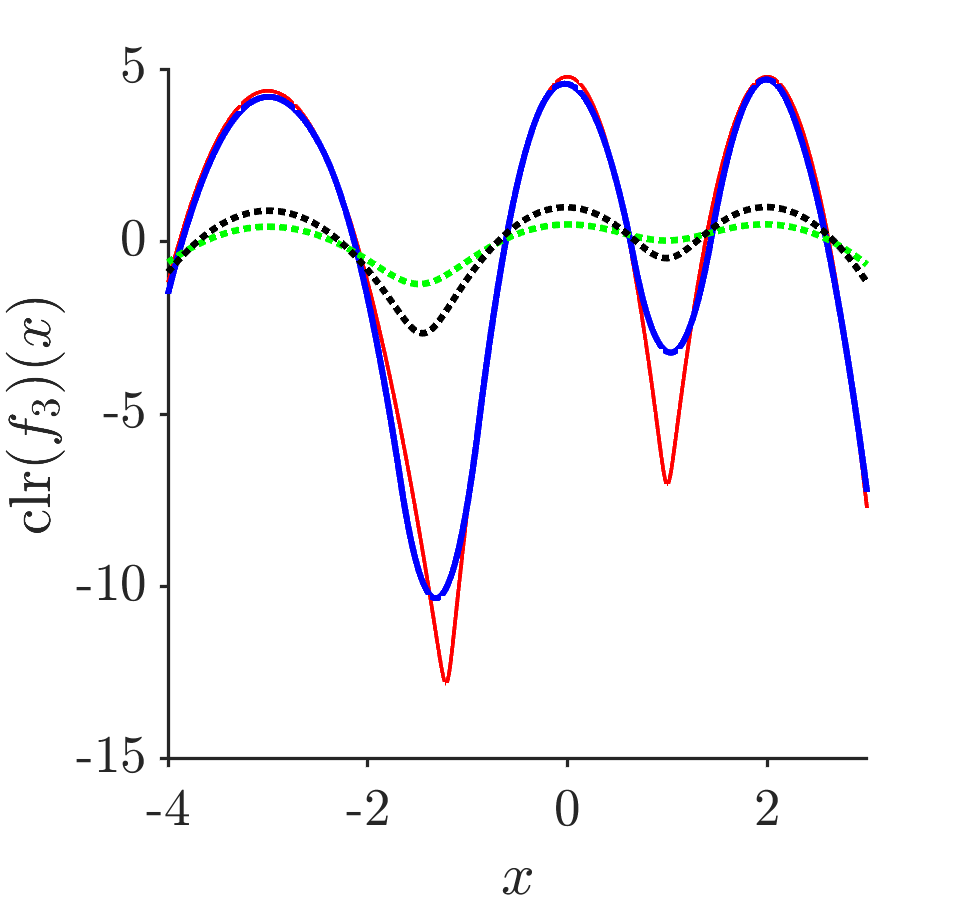}
    \caption{Comparison of the true density $f_3(x)$ (red), the penalized likelihood estimate constructed with $N=3000$ using $g=8$ and $\rho = 0.01$(blue), and the kernel density estimates with $N=3000$ based on Silverman's rule (dotted, green), Scott's rule (dotted, yellow), and the Sheather--Jones method (dotted, black). The estimates based on Scott's rule and the Sheather--Jones method coincide. The unit-integral representatives in $B^2(I)$ are shown in the left panel, and their clr representations in $L_0^2(I)$ are shown in the right panel.}
    \label{final_comp_f3}
\end{figure}

\begin{table}[h]
\centering
\small
\setlength{\tabcolsep}{5pt}
\begin{tabular}{cc|c|cccccccccc}
\toprule
 &  &  & \multicolumn{10}{c}{$g,\ h$} \\
\cmidrule(lr){4-13}
$N$ & $K$ & $\rho$ & 1 & 2 & 3 & 4 & 5 & 6 & 7 & 8 & 9 & 10 \\
\midrule
3000 & 5 & $\rho_x$ & 0.004 & 0.019 & 0.030 & 0.034 & 0.039 & 0.048 & 0.058 & 0.066 & 0.083 & 0.102 \\
3000 & 5 & $\rho_y$ & 0.002 & 0.006 & 0.015 & 0.021 & 0.029 & 0.034 & 0.040 & 0.049 & 0.056 & 0.066 \\
\bottomrule
3000 & 10 & $\rho_x$ & 0.004 & 0.019 & 0.031 & 0.034 & 0.039 & 0.045 & 0.054 & 0.066 & 0.078 & 0.099 \\
3000 & 10 & $\rho_y$ & 0.002 & 0.006 & 0.015 & 0.021 & 0.030 & 0.035 & 0.039 & 0.049 & 0.056 & 0.063 \\
\bottomrule
5000 & 5 & $\rho_x$ & 0.004 & 0.020 & 0.032 & 0.040 & 0.045 & 0.049 & 0.053 & 0.064 & 0.072 & 0.076 \\
5000 & 5 & $\rho_y$ & 0.002 & 0.006 & 0.015 & 0.021 & 0.029 & 0.031 & 0.038 & 0.043 & 0.054 & 0.065 \\
\bottomrule
5000 & 10 & $\rho_x$ & 0.004 & 0.020 & 0.033 & 0.041 & 0.044 & 0.047 & 0.054 & 0.063 & 0.072 & 0.077 \\
5000 & 10 & $\rho_y$ & 0.002 & 0.007 & 0.015 & 0.022 & 0.029 & 0.034 & 0.039 & 0.044 & 0.054 & 0.064 \\
\bottomrule
\end{tabular}
\caption{Mean values of optimal penalization parameters $\rho_x$ and $\rho_y$ for estimating the  density $f_2(x,y)$ using different numbers of sampled values ($N_1=3000$, $N_2=5000$) and different number of folds $K$ used in \eqref{CV_2D} ($K_1=5$, $K_2=10$) (rows) for different number of inner knots $g=h$ (colums).}
\label{opt_rho_f2_2D}
\end{table}

\begin{figure}[h]
    \centering
    \includegraphics[width=0.4\textwidth]{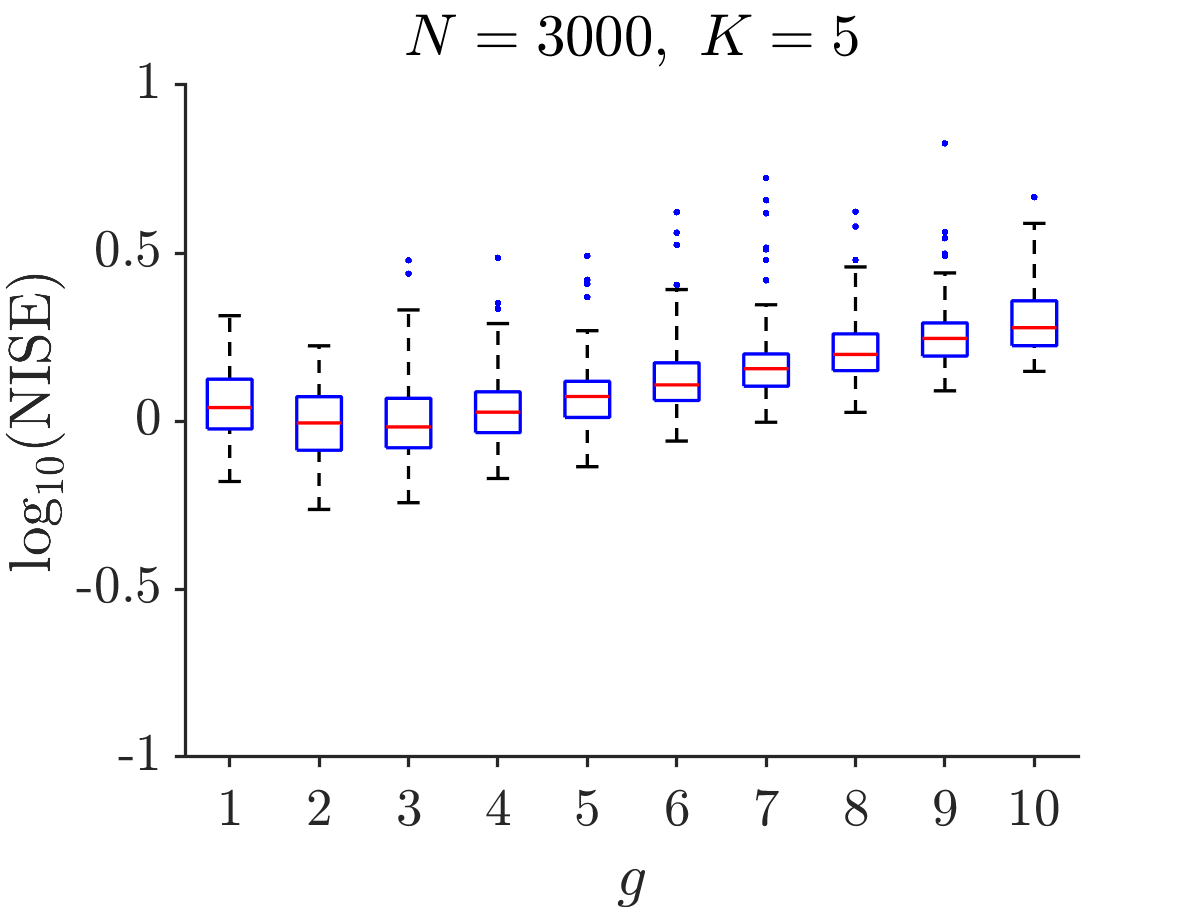}\qquad
    \includegraphics[width=0.4\textwidth]{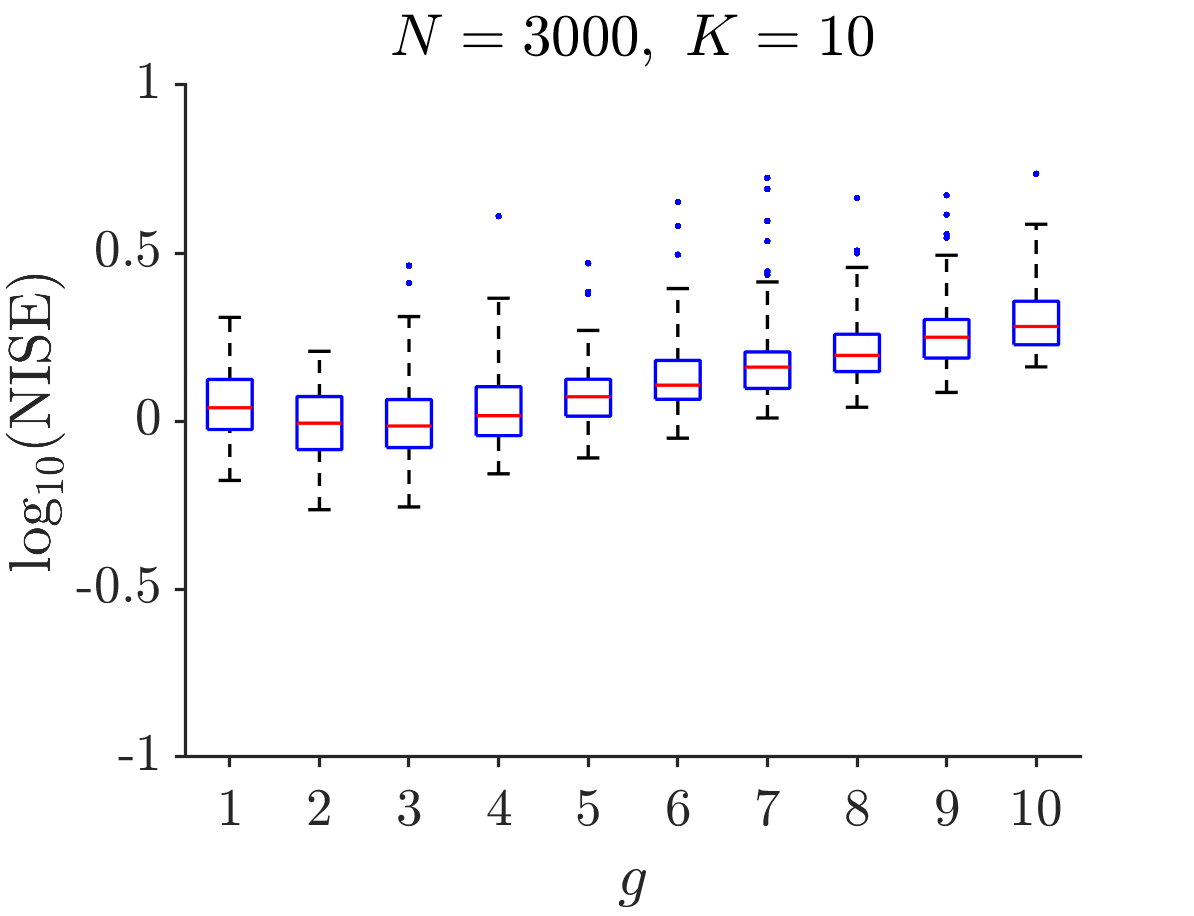}\\[0.5cm]
    \includegraphics[width=0.4\textwidth]{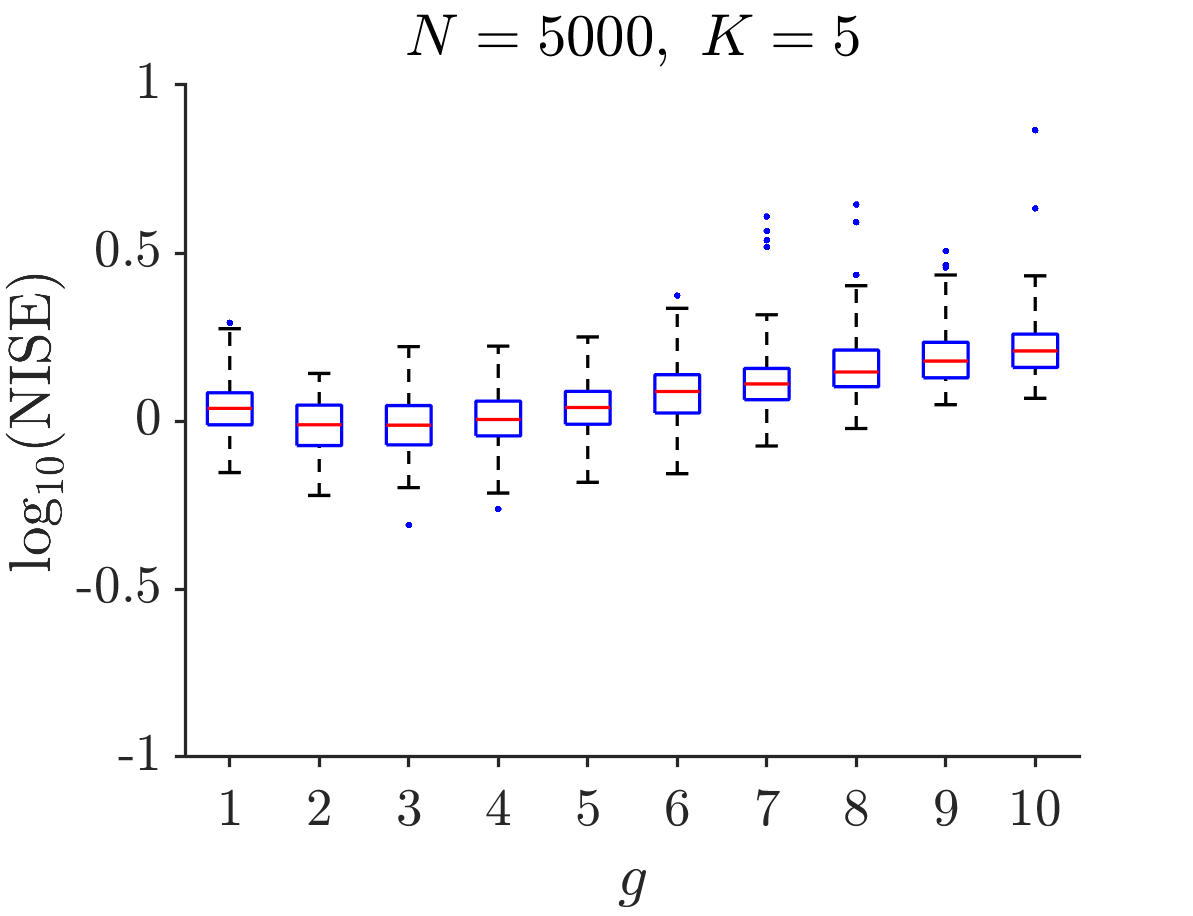}\qquad
    \includegraphics[width=0.4\textwidth]{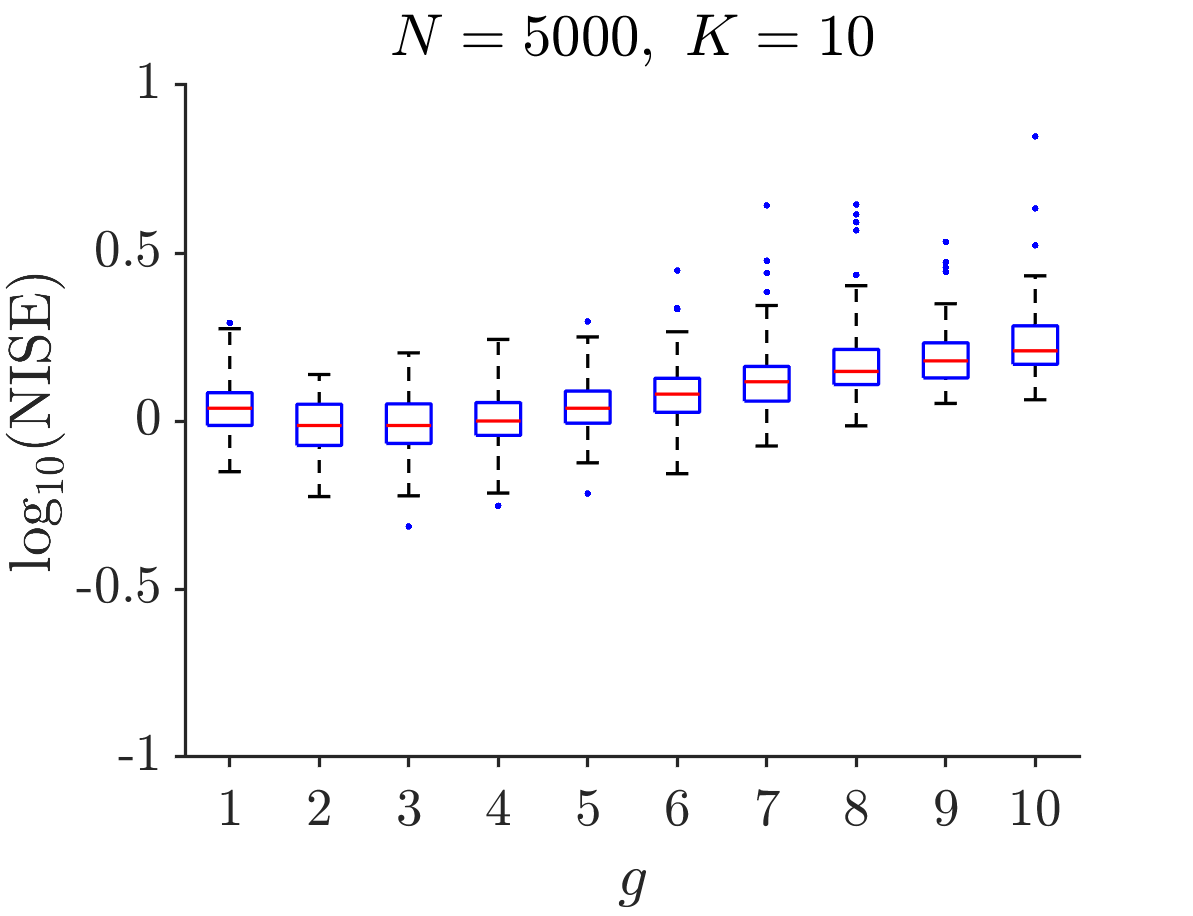}
    \caption{Boxplots of NISE values for the penalized likelihood estimates of $f_2(x,y)$ for different numbers of inner knots $g=h$ and the considered combinations of the sample size $N$ and the number of cross-validation folds $K$.}
    \label{NISE_f2_2D}
\end{figure}

\begin{table}[h]
\centering
\small
\setlength{\tabcolsep}{5pt}
\renewcommand{\arraystretch}{1.15}
\begin{tabular}{cc|cccccccccc}
\toprule
 &  & \multicolumn{10}{c}{$g=h$} \\
\cmidrule(lr){3-12}
$N$ & $K$ &1&2&3&4&5&6&7&8&9&10\\
\midrule
3000&5 
& 1.097& 0.987& \textbf{0.960}& 1.062& 1.183& 1.281& 1.432& 1.578& 1.760& 1.896 \\
3000&10 
& 1.095& 0.985& \textbf{0.966}& 1.037& 1.181& 1.278& 1.447& 1.567& 1.775& 1.911 \\
5000&5 
& 1.090& 0.975& \textbf{0.972}& 1.010& 1.097& 1.224& 1.289& 1.399& 1.508& 1.617 \\
5000&10
& 1.091& 0.9697& \textbf{0.9696}& 1.000&  1.091& 1.202& 1.309& 1.404& 1.510& 1.618 \\
\bottomrule
\end{tabular}
\caption{Median NISE values of $f_2(x,y)$ using different numbers of sampled values ($N_1 = 3000$, $N_2 = 5000$) and different number of folds used in \eqref{CV_1D} ($K_1 = 5$, $K_2 = 10$) (rows) for different numbers of inner knots (columns) with the highlighted minimal NISE for each setting.}
\label{NISE_2D_f2}
\end{table}

\begin{figure}[h]
    \centering
    \includegraphics[width=0.32\textwidth]{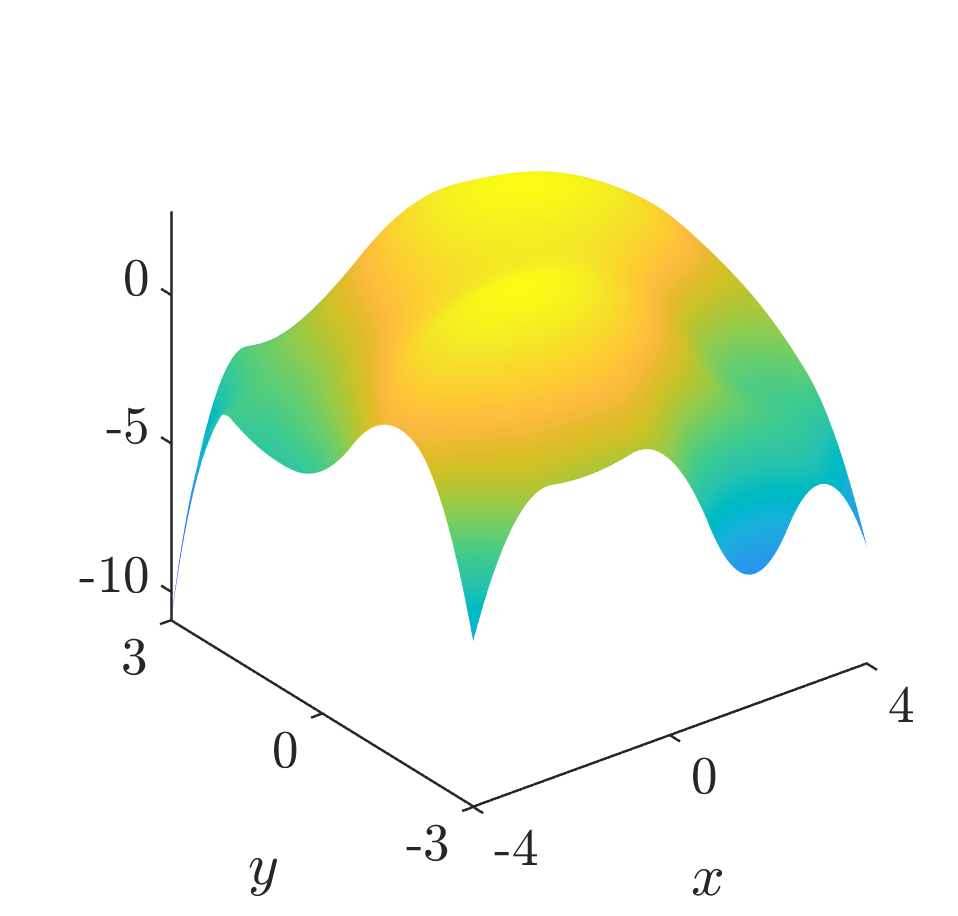}
    \includegraphics[width=0.32\textwidth]{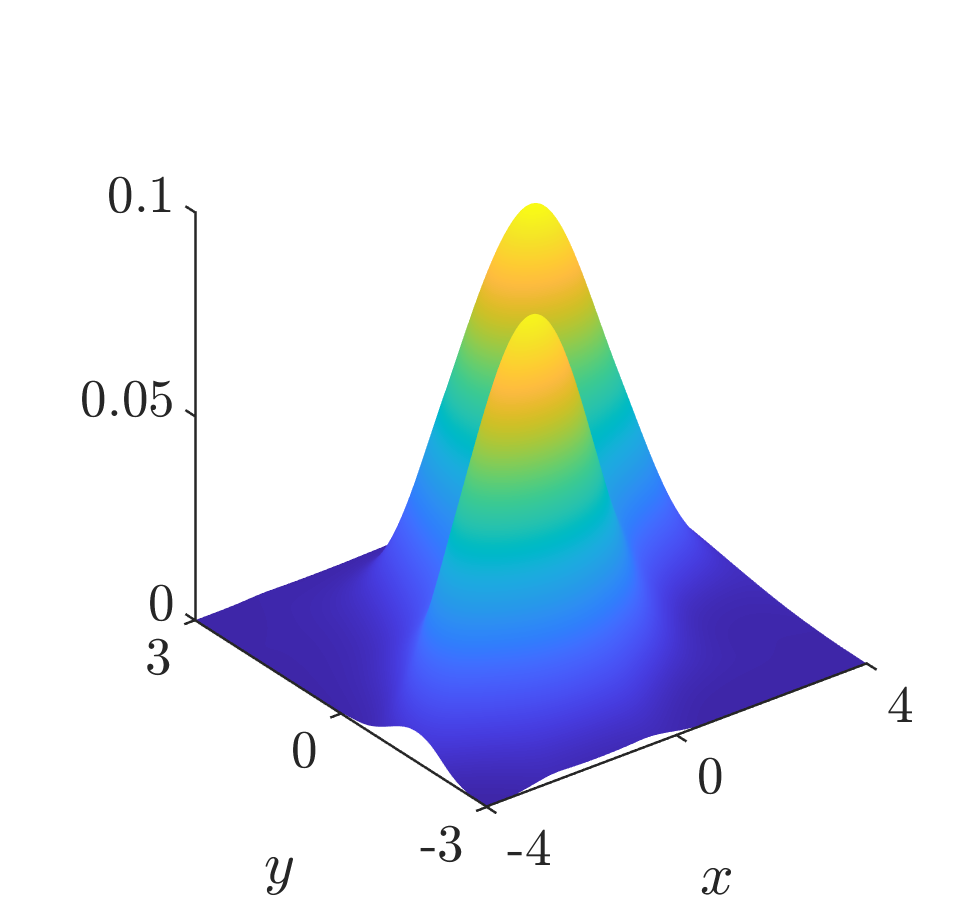}
    \includegraphics[width=0.32\textwidth]{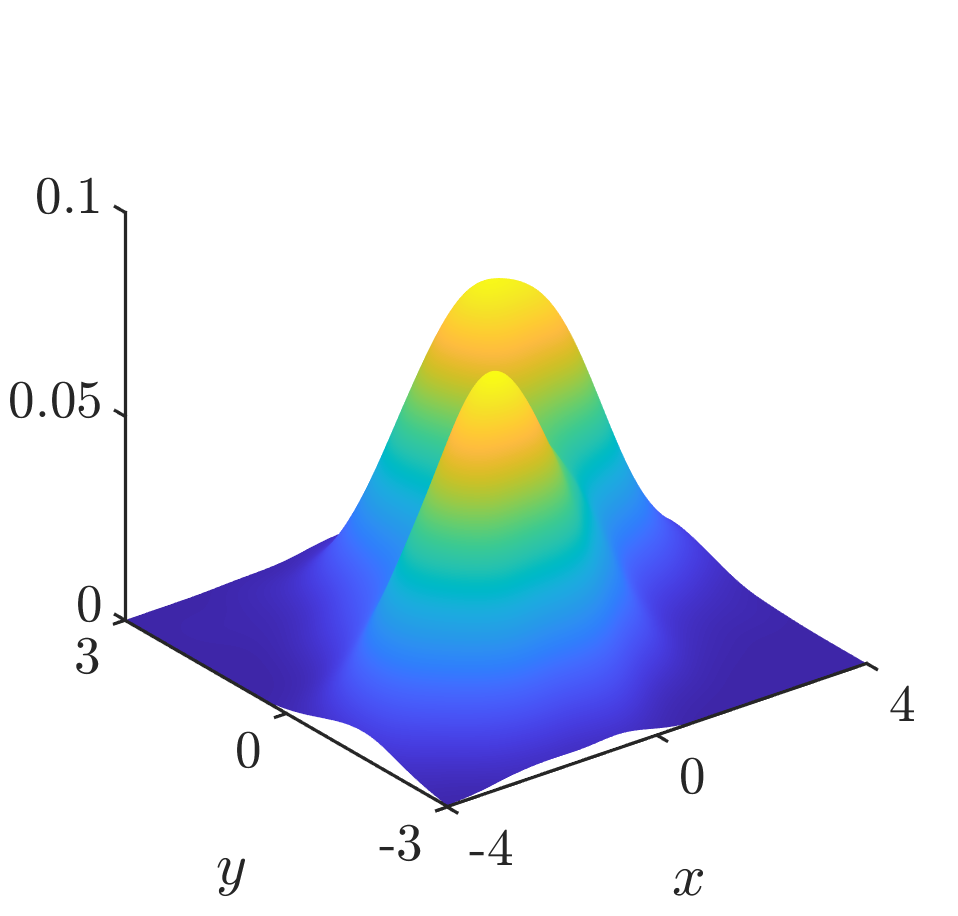}
    \caption{Penalized likelihood estimation of the density $f_2(x,y)$ constructed with $N=3000$ using $g=h=3$ and $\rho_x = 0.03$, $\rho_y=0.015$ in $L^2_0(\Omega)$ (left) and in $\mathcal{B}^2(\Omega)$ (middle), together with the kernel density estimate with $N=3000$ and bandwidth obtained with the Silverman's rule in $\mathcal{B}^2(\Omega)$.}
    \label{final_comp_2D_f2}
\end{figure}

\section{Supplementary materials for Application}
\label{app:application}
In this appendix we include supplementary outputs for Section \ref{Application}.

For univariate lead distribution, Table \ref{application_cv_d=2} reports cross-validated negative log-likelihood values with optimal values of penalization parameter $\rho$ given the considered number of inner knots $g$ and number of folds $K$ for second-order difference used in penalized maximum likelihood estimation. Table \ref{application_cv_kde_d=2} gives information about the cross-validated negative log-likelihood values of kernel smoothing with different bandwidth selection.

For bivariate lead and zinc distribution, Tables \ref{application_cv_2D_d2_K5} and \ref{application_cv_2D_d2_K10} give information about cross-validated negative log-likelihood values and optimal values of penalization parameters $(\rho_x,\rho_y)$ with second-order differences used in penalized maximum likelihood estimation and different number of inner knots $g$ and number of folds $K$.

\begin{table}[h]
\centering
\small
\setlength{\tabcolsep}{4pt}
\renewcommand{\arraystretch}{1.15}

\begin{tabular}{c|cccc|cccc}
\toprule
& \multicolumn{4}{c|}{$K=5$} 
& \multicolumn{4}{c}{$K=10$} \\

\cmidrule(lr){2-5}
\cmidrule(lr){6-9}

$g$
& Mean CV & Med CV & Mean $\rho$ & Med $\rho$
& Mean CV & Med CV & Mean $\rho$ & Med $\rho$ \\

\midrule
1  & 43.758 & 38.018 & 0.078 & 0.009 & 23.291 & 19.355 & 0.025 & 0.009 \\
2  & 39.524 & 36.285 & 0.019 & 0.007 & 19.758 & 17.671 & 0.028 & 0.007 \\
3  & 39.019 & 34.044 & 0.037 & 0.022 & 19.505 & 16.634 & 0.047 & 0.022 \\
4  & 38.922 & 34.203 & 0.082 & 0.045 & 19.457 & 16.865 & 0.087 & 0.057 \\
5  & 38.753 & 34.235 & 0.176 & 0.115 & 19.378 & 16.753 & 0.203 & 0.115 \\
6  & 38.603 & 33.669 & 0.3   & 0.233 & 19.291 & 16.522 & 0.363 & 0.233 \\
7  & 38.493 & 34.176 & 0.516 & 0.373 & 19.227 & 16.731 & 0.625 & 0.471 \\
8  & 38.523 & 34.182 & 0.81  & 0.596 & 19.251 & 16.75  & 0.95 & 0.596 \\
9  & 38.376 & \textbf{33.563} & 1.059 & 0.754 & 19.173 & \textbf{16.485} & 1.219 & 0.954 \\
10 & \textbf{38.339} & 33.966 & 1.474 & 0.954 & \textbf{19.158} & 16.674 & 1.732 & 1.207 \\

\bottomrule
\end{tabular}

\caption{
Cross-validated negative log-likelihood values and corresponding optimal penalization parameters for the proposed penalized likelihood estimator for $d=2$ with different numbers of inner knots $g$ and different numbers of folds $K$ with highlighted minimal mean and median CV values.
}

\label{application_cv_d=2}

\end{table}

\begin{table}[h]
\centering
\small
\setlength{\tabcolsep}{7pt}
\renewcommand{\arraystretch}{1.15}

\begin{tabular}{l|cc|cc}
\toprule
& \multicolumn{2}{c|}{$K=5$} 
& \multicolumn{2}{c}{$K=10$} \\

\cmidrule(lr){2-3}
\cmidrule(lr){4-5}
Method 
& Mean CV & Median CV
& Mean CV & Median CV \\
\midrule
Silverman & 55.727 & 41.721 & 26.307 & 19.892 \\
Scott     & \textbf{50.051} & \textbf{38.098} & 24.331 & \textbf{19.483} \\
Sheater-Jones        & 50.272 & 39.808 & \textbf{23.995} & 19.885 \\
\bottomrule
\end{tabular}
\caption{
Cross-validated negative log-likelihood values of kernel density estimators with different bandwidth selection rules with highlighter minimal mean and median CV values.
}
\label{application_cv_kde_d=2}
\end{table}

\begin{table}[h]
\centering
\small
\setlength{\tabcolsep}{4pt}
\renewcommand{\arraystretch}{1.15}

\begin{tabular}{c|cccccc}
\toprule

& \multicolumn{6}{c}{$\mathrm{d}_x=\mathrm{d}_y=2,\ K=5$} \\

\cmidrule(lr){2-7}

$g=h$
& Mean CV & Med CV
& Mean $\rho_x$ & Med $\rho_x$
& Mean $\rho_y$ & Med $\rho_y$ \\

\midrule

1 & 92.218 & 79.705 & 0.336 & 0.001 & 0.018 & 0.002 \\
2 & 76.013 & \textbf{66.902} & 0.014 & 0.002 & 0.025 & 0.007 \\
3 & 75.410 & 66.961 & 0.025 & 0.013 & 0.058 & 0.043 \\
4 & 74.810 & 68.189 & 0.052 & 0.023 & 0.114 & 0.078 \\
5 & 74.512 & 67.399 & 0.120 & 0.078 & 0.232 & 0.144 \\
6 & 74.384 & 67.764 & 0.234 & 0.144 & 0.437 & 0.264 \\
7 & 74.055 & 68.064 & 0.459 & 0.264 & 0.697 & 0.483 \\
8 & 74.196 & 69.294 & 0.642 & 0.483 & 1.140 & 0.886 \\
9 & 74.018 & 69.701 & 1.004 & 0.483 & 1.786 & 0.886 \\
10 & \textbf{73.943} & 70.360 & 1.721 & 0.886 & 2.444 & 1.624 \\

\bottomrule
\end{tabular}

\caption{
Cross-validated negative log-likelihood values together with the corresponding optimal penalization parameters $\rho_x$ and $\rho_y$ of the proposed penalized likelihood estimator for different numbers of inner knots $g=h$ using 5-fold cross-validation with $\mathrm{d}_x=\mathrm{d}_y=2$. The smallest mean and median cross-validation values are highlighted in bold.
}

\label{application_cv_2D_d2_K5}

\end{table}

\begin{table}[h]
\centering
\small
\setlength{\tabcolsep}{4pt}
\renewcommand{\arraystretch}{1.15}

\begin{tabular}{c|cccccc}
\toprule
& \multicolumn{6}{c}{$\mathrm{d}_x=\mathrm{d}_y=2,\ K=10$} \\
\cmidrule(lr){2-7}
$g=h$
& Mean CV & Med CV
& Mean $\rho_x$ & Med $\rho_x$
& Mean $\rho_y$ & Med $\rho_y$ \\
\midrule
1 & 43.714 & 37.986 & 0.267 & 0.001 & 0.150 & 0.001 \\
2 & 39.808 & 36.606 & 0.014 & 0.002 & 0.019 & 0.007 \\
3 & 39.460 & 36.883 & 0.021 & 0.007 & 0.052 & 0.023 \\
4 & 39.156 & 35.987 & 0.048 & 0.023 & 0.114 & 0.078 \\
5 & 39.100 & 36.056 & 0.125 & 0.078 & 0.241 & 0.144 \\
6 & 38.960 & 35.904 & 0.247 & 0.144 & 0.460 & 0.264 \\
7 & 38.724 & 34.991 & 0.398 & 0.264 & 0.837 & 0.483 \\
8 & 38.751 & 35.160 & 0.671 & 0.483 & 1.259 & 0.886 \\
9 & 38.637 & \textbf{34.941} & 1.029 & 0.483 & 1.745 & 1.624 \\
10 & \textbf{38.591} & 35.426 & 1.472 & 0.886 & 2.458 & 1.624 \\
\bottomrule
\end{tabular}
\caption{
Cross-validated negative log-likelihood values together with the corresponding optimal penalization parameters $\rho_x$ and $\rho_y$ of the proposed penalized likelihood estimator for different numbers of inner knots $g=h$ using 10-fold cross-validation with $\mathrm{d}_x=\mathrm{d}_y=2$. The smallest mean and median cross-validation values are highlighted in bold.
}
\label{application_cv_2D_d2_K10}

\end{table}

\end{appendices}

\bibliography{bibliography}

\end{document}